\documentclass{article}

\usepackage{nicematrix}

\usepackage{framed}

\usepackage{lipsum}

\usepackage{tikz, pgf, pgfplots}
\usetikzlibrary
	{
		arrows, automata, chains, datavisualization.formats.functions, fadings, fit, positioning, quantikz, quotes, shapes
	}
\usepackage{tikzpeople}
\usepackage{fontawesome5}
\newcounter{mathseed}
\pgfmathsetseed{\arabic{mathseed}}
\pgfdeclarelayer{background}
\pgfsetlayers{background,main}
\pgfdeclaredecoration{irregular fractal line}{init}
{
	\state{init}[width=\pgfdecoratedinputsegmentremainingdistance]
	{
		\pgfpathlineto{\pgfpoint{random*\pgfdecoratedinputsegmentremainingdistance}{(random*\pgfdecorationsegmentamplitude-0.02)*\pgfdecoratedinputsegmentremainingdistance}}
		\pgfpathlineto{\pgfpoint{\pgfdecoratedinputsegmentremainingdistance}{0pt}}
	}
}
\tikzset
{
	paper/.style =
	{
		draw = MyDarkBlue!10, blur shadow, every shadow/.style = { opacity = 1, MyDarkBlue }, shading = bilinear interpolation, lower left = MyDarkBlue!10, upper left = MyDarkBlue!5, upper right = GreenTeal!75, lower right = MyDarkBlue!5, fill=none
	},
	irregular cloudy border/.style =
	{
		decoration = { irregular fractal line, amplitude = 0.2 }, decorate,
	},
	irregular spiky border/.style =
	{
		decoration = { irregular fractal line, amplitude = -0.2 }, decorate,
	},
	ragged border/.style =
	{
		decoration = {random steps, segment length = 7mm, amplitude = 2mm }, decorate
	}
}
\def\tornpaper#1{%
	\ifthenelse{\isodd{\value{mathseed}}}
	{%
		\tikz
		{
			\node[inner sep = 1em] (A) {#1};		
			\begin{pgfonlayer}{background}			
				\fill[paper]						
				\pgfextra{\pgfmathsetseed{\arabic{mathseed}}\addtocounter{mathseed}{1}}%
				{decorate[irregular cloudy border]{decorate{decorate{decorate{decorate[ragged border]{
										(A.north west) -- (A.north east)
				}}}}}}
				-- (A.south east)
				\pgfextra{\pgfmathsetseed{\arabic{mathseed}}}%
				{decorate[irregular spiky border]{decorate{decorate{decorate{decorate[ragged border]{
										-- (A.south west)
				}}}}}}
				-- (A.north west);
			\end{pgfonlayer}
		}
	}
	{%
		\tikz{
			\node[inner sep=1em] (A) {#1};  
			\begin{pgfonlayer}{background}  
				\fill[paper] 
				\pgfextra{\pgfmathsetseed{\arabic{mathseed}}\addtocounter{mathseed}{1}}%
				{decorate[irregular spiky border]{decorate{decorate{decorate{decorate[ragged border]{
										(A.north east) -- (A.north west)
				}}}}}}
				-- (A.south west)
				\pgfextra{\pgfmathsetseed{\arabic{mathseed}}}%
				{decorate[irregular cloudy border]{decorate{decorate{decorate{decorate[ragged border]{
										-- (A.south east)
				}}}}}}
				-- (A.north east);
		\end{pgfonlayer}}
	}
}

\usepackage{amsthm}
\usepackage{amssymb}
\usepackage[bb = boondox]{mathalfa}
\usepackage{dsfont}
\usepackage{newtxmath}
\usepackage{mathtools}
\usepackage{multicol}
\usepackage{braket}
\usepackage{physics}
\numberwithin{equation}{section}

\usepackage{yquant}

\usepackage[a4paper, left = 2.5cm, right = 2.5 cm, top = 2.5cm, bottom = 3.0 cm]{geometry}

\definecolor{MyLightRed}{RGB}{244, 213, 245}
\definecolor{WordRed}{RGB}{255, 0, 102}
\definecolor{RedDarkLightest}{HTML}{ff0088}
\definecolor{RedDarkLight}{HTML}{ea005f}
\definecolor{RedPurple}{HTML}{aa007f}
\definecolor{Purple}{HTML}{911146}
\definecolor{WordLightGreen}{RGB}{140, 214, 192}
\definecolor{WordGreen}{RGB}{0, 176, 80}
\definecolor{GreenLightest}{HTML}{00ffa0}
\definecolor{GreenLighter1}{HTML}{00b383}
\definecolor{GreenLighter2}{HTML}{00aa7f}
\definecolor{GreenDark}{HTML}{225522}
\definecolor{GreenTeal}{HTML}{008080}
\definecolor{WordIceBlue}{RGB}{223, 227, 229}
\definecolor{MyVeryLightBlue}{RGB}{211, 245, 247}
\definecolor{WordBlueVeryLight}{RGB}{0, 176, 240}
\definecolor{WordBlueLight}{RGB}{0, 112, 192}
\definecolor{WordBlueDark}{RGB}{46, 116, 181}
\definecolor{WordBlueDarker}{RGB}{31, 78, 121}
\definecolor{WordBlueDarker25}{RGB}{54, 96, 146}
\definecolor{WordBlueDarker50}{RGB}{36, 64, 98}
\definecolor{WordBlueDarkest}{RGB}{0, 32, 96}
\definecolor{WordBlue}{RGB}{19, 65, 99}
\definecolor{MyBlue}{RGB}{0, 64, 128}
\definecolor{MyDarkBlue}{RGB}{0, 51, 102}
\definecolor{BlueVeryDark}{HTML}{222255}
\definecolor{MagentaVeryLight}{RGB}{178, 162, 201}
\definecolor{MagentaLighter}{RGB}{161, 106, 221}
\definecolor{MagentaLight}{RGB}{128, 100, 162}
\definecolor{MagentaDark}{RGB}{106, 65, 152}
\definecolor{MagentaVeryDark}{RGB}{97, 75, 128}
\definecolor{WordAquaLighter80}{RGB}{218, 238, 243}
\definecolor{WordAquaLighter60}{RGB}{183, 222, 232}
\definecolor{WordAquaLighter40}{RGB}{146, 205, 220}
\definecolor{WordAquaDarker25}{RGB}{49, 134, 155}
\definecolor{WordAquaDarker50}{RGB}{33, 89, 103}
\definecolor{WordVeryLightTeal}{RGB}{223, 236, 235}
\definecolor{WordLightTeal}{RGB}{160, 199, 197}
\definecolor{WordDarkTealLighter80}{RGB}{207, 223, 234}
\definecolor{WordDarkTeal}{RGB}{72, 123, 119}
\definecolor{WordDarkerTeal}{RGB}{48, 82, 80}
\definecolor{WordTurquoiseLighter80}{RGB}{209, 238, 249}
\definecolor{Brown}{HTML}{666633}

\usepackage{xcolor}
\usepackage{varwidth}
\usepackage[skins, breakable, listings, theorems] {tcolorbox}
\usepackage{graphicx}
\usepackage{hyperref}
\hypersetup
{
	colorlinks,
	linkcolor = {WordRed!75!black},
	citecolor = {WordBlueDarker25},
	urlcolor = {WordAquaDarker50}
}
\usepackage{colortbl}
\usepackage{float}
\usepackage{cellspace}
\usepackage{makecell}
\usepackage{multirow}
\usepackage[linesnumbered, tworuled, vlined]{algorithm2e}
\usepackage{appendix}
\usepackage{enumitem}
\usepackage{xintexpr}
\newcommand{\orcidicon}[1]{\href{https://orcid.org/#1}{\includegraphics[height=\fontcharht\font`\B]{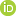}}}

\newtheorem{definition}{Definition}[section]
\newtheorem{theorem}{Theorem}[section]
\newtheorem{lemma}[theorem]{Lemma}
\newtheorem{proposition}[theorem]{Proposition}
\newtheorem{corollary}[theorem]{Corollary}
\newtheorem{example}{Example}[section]

\title
	{
		A Quantum Detectable Byzantine Agreement Protocol using only EPR pairs
	}

\author
	{
		Theodore Andronikos$^1$\orcidicon{0000-0002-3741-1271}
		and
		Alla Sirokofskich$^2$\\
		\\
		$^1$ \ Department of Informatics, Ionian University, \\
		7 Tsirigoti Square, 49100 Corfu, Greece; \\
		andronikos@ionio.gr \\
		$^2$ \ Department of History and Philosophy of Sciences, \\
		National and Kapodistrian University of Athens, \\
		Athens 15771, Greece; \\
		asirokof@math.uoa.gr
	}

\begin{document}

\maketitle

\begin{abstract}
	In this paper, we introduce a new quantum protocol for Detectable Byzantine Agreement. What distinguishes the proposed protocol among similar quantum protocols, is the fact that it uses only EPR pairs, and, in particular, $\ket{ \Psi^{ + } }$ pairs. There are many sophisticated quantum protocols that guarantee Detectable Byzantine Agreement, but they do not easily lend themselves to practical implementations, due to present-day technological limitations. For a large number $n$ of players, $\ket{ GHZ }$ $n$-tuples, or other more exotic entangled states, are not easy to produce, a fact which might complicate the scalability of such protocols. In contrast, Bell states are, undoubtedly, the easiest to generate among maximally entangled states. This will, hopefully, facilitate the scalability of the proposed protocol, as only EPR pairs are required, irrespective of the number $n$ of players. Finally, we mention that, even for arbitrary many players $n$, our protocol always completes in a constant number of rounds, namely $4$.
\end{abstract}
\section{Introduction} \label{sec:Introduction}

The term ``Byzantine Agreement'' has its origins in the landmark paper \cite{Lamport1982}. The paper studied the problem of using a playful terminology, involving Byzantium, a commanding general and many lieutenant generals, some loyal and some traitors. In another seminal work \cite{Pease1980} the same authors had previously tackled the same problem but in a more formal way. The approach in \cite{Pease1980} was perfectly distributed and symmetrical, whereas the approach in \cite{Lamport1982} has an slight asymmetry because the commanding general assumes the role of coordinating the lieutenant generals.

Currently, probably the most important application of Byzantine Agreement protocols would be in facilitating the implementation of a distributed consensus mechanism that enforces the consistency of data. Byzantine Agreement protocols are critical for successfully tackling distributed fault scenarios in the presence of adversaries (traitors) and enhancing the security of blockchain networks.

Today, it is clear that we have entered the quantum era. This is much more than a buzzword, as it promises to bring fundamental changes in our capabilities, indicated by the impressive progress in the construction of new, more powerful than before, quantum computers, such as IBM's $127$ qubit processor Eagle \cite{IBMEagle} and the more recent $433$ qubit Osprey \cite{IBMOsprey} quantum processor. At the same time, caution is warranted because these new powerful quantum computers bring us closer to the practical implementation of the quantum algorithms developed by Peter Shor and Lov Grover \cite{Shor1994, Grover1996}, that can compromise the security of the classical world.

It is fitting that the quantum paradigm can also be employed to provide us with novel quantum algorithms that adequately protect our critical information and achieve uncompromised security. The unique quantum phenomena are exploited in the design of secure protocols, e.g., for key distribution, as in \cite{Bennett1984, Ekert1991, Gisin2004, inoue2002differential, guan2015experimental, waks2006security, Ampatzis2021}, for secret sharing, as in \cite{Ampatzis2022, Ampatzis2023}, for cloud storage \cite{attasena2017secret, ermakova2013secret} or blockchain \cite{cha2021blockchain, Sun2020, Qu2023}. The emerging quantum setting has inspired researchers to approach the problem of Byzantine Agreement from a quantum perspective.

In an influential paper \cite{Fitzi2001} in 2001, the authors introduced a variant of Byzantine Agreement, called ``Detectable Byzantine Agreement'' and gave the first quantum protocol for Detectable Byzantine Agreement. Their protocol, which tackled only the case of $3$ players, used entangled qutrits in the Aharonov state. Since then, the concept of Detectable Byzantine Agreement has become widely accepted and has generated a plethora of similar works. Other early efforts include \cite{Cabello2003}, where another protocol was proposed utilizing a four-qubit singlet state, and \cite{Neigovzen2008} that introduced a continuous variable solution to the Byzantine Agreement problem with multipartite entangled Gaussian states. More recently, notable contributions were made by \cite{Feng2019}, presenting a quantum protocol based on tripartite GHZ-like states and homodyne measurements in the continuous variable scenario, by \cite{Sun2020}, devising a protocol without entanglement, with the additional attractive feature of achieving agreement in only three rounds, and by \cite{Qu2023}, introducing a general $n$-party quantum protocol based on GHZ states, with possible use in a quantum blockchain network.

In this work, we introduce a novel quantum protocol, named EPRQDBA, and formally prove that EPRQDBA achieves detectable byzantine agreement. Almost all previous analogous protocols have been cast in the form of a game, and, we, too, follow this tradition. In particular, in our protocol, the familiar protagonists Alice, Bob and Charlie make one more appearance. It is instructive to mention that the pedagogical nature of games often makes expositions of difficult and technical concepts easier to understand and appreciate. Ever since their introduction in 1999 \cite{Meyer1999, Eisert1999}, quantum games, have offered additional insight because quite often quantum strategies seem to achieve better results than classical ones \cite{Andronikos2018, Andronikos2021, Andronikos2022a}. The famous prisoners' dilemma game provides the most prominent example \cite{Eisert1999, Giannakis2019}, which also applies to other abstract quantum games \cite{Giannakis2015a}. The quantization of many classical systems can even apply to political structures, as was shown in \cite{Andronikos2022}.

\textbf{Contribution}. This paper presents a new quantum protocol for Detectable Byzantine Agreement, called EPRQDBA. Indeed there exist many sophisticated quantum protocols guaranteed to achieve Detectable Byzantine Agreement, so why a new one? The distinguishing feature of the EPRQDBA protocol, which does not utilize a quantum signature scheme, is the fact that it only relies on EPR pairs. Exotic multi-particle entangled states employed in other protocol are not so easy to produce with our current quantum apparatus. This increases preparation time and complicates their use in situations where the number $n$ of players is large. Similarly, although for small values of $n$ $\ket{ GHZ_{ n } }$ states are easily generated by contemporary quantum computers, when the number $n$ of players increases, it becomes considerably more difficult to prepare and distribute $\ket{ GHZ_{ n } }$ tuples. Thus, protocols that require $\ket{ GHZ_{ n } }$ tuples for $n$ players dot not facilitate scalability. In contrast, Bell states are, without a doubt the easiest to generate among maximally entangled states. The EPRQDBA only requires EPR pairs, specifically $\ket{ \Psi^{ + } }$ pairs, irrespective of the number $n$ of players. This leads to a reduction in preparation time, increases scalability, and, ultimately, practicability. Another strong point of the EPRQDBA protocol is the fact that it completes in a constant number of rounds, specifically $4$ rounds. Thus, the completion time of the protocol is independent of the number $n$ of players, which also enhances its scalability.

\subsection*{Organization} \label{subsec:Organization}

The paper is organized as follows. Section \ref{sec:Introduction} contains an introduction to the subject along with bibliographic pointers to related works. Section \ref{sec:Preliminaries} provides a brief exposition to all the concepts necessary for the understanding of our protocol. Section \ref{sec:The $3$ Player EPRQDBA Protocol} contains a detailed exposition of the EPRQDBA protocol for the special case of $3$ players, namely Alice, Bob and Charlie. Section \ref{sec:The $n$ Player EPRQDBA Protocol} gives a formal presentation of the EPRQDBA protocol in the general case of $n$ players. Finally, Section \ref{sec:Discussion and Conclusions} contains a summary and a discussion on some of the finer points of this protocol.

\section{Preliminaries} \label{sec:Preliminaries}

\subsection{$\ket{ \Psi^{ + } }$ EPR pairs} \label{subsec:Psi^{ + } EPR Pairs}

Quantum entanglement is one of the most celebrated properties of quantum mechanics and, without exaggeration, serves as the core of the majority of quantum protocols. Mathematically, entangled states of composite systems must be described as a linear combination of two or more product states of their subsystems, as a single product state will not suffice. The famous Bell states are special quantum states of two qubits, also called EPR pairs, that represent the simplest form of maximal entanglement. These states are succinctly described by the next formula from \cite{Nielsen2010}.

\begin{align} \label{eq:Bell States General Equation}
	\ket{ \beta_{ x, y } } = \frac { \ket{ 0 } \ket{ y } + (-1)^x \ket{ 1 } \ket{ \Bar{ y } } } { \sqrt{ 2 } } \ ,
\end{align}

where $\ket{\Bar{y}}$ is the negation of $\ket{y}$.

There are four Bell states and their specific mathematical expression is given below. The subscripts $A$ and $B$ are used to emphasize the subsystem to which the corresponding qubit belongs, that is, qubits $\ket{ \cdot }_{ A }$ belong to Alice and qubits $\ket{ \cdot }_{ B }$ belong to Bob.

\begin{tcolorbox}
	[
		grow to left by = 1.50 cm,
		grow to right by = 0.00 cm,
		colback = white,			
		enhanced jigsaw,			
		sharp corners,
		toprule = 0.1 pt,
		bottomrule = 0.1 pt,
		leftrule = 0.1 pt,
		rightrule = 0.1 pt,
		sharp corners,
		center title,
		fonttitle = \bfseries
	]
	\begin{minipage}[b]{0.475 \textwidth}
		\begin{align} \label{eq:Bell State Phi +}
			\ket{ \Phi^{ + } } = \ket{ \beta_{ 00 } } = \frac { \ket{ 0 }_{ A } \ket{ 0 }_{ B } + \ket{ 1 }_{ A } \ket{ 1 }_{ B } } { \sqrt{ 2 } }
		\end{align}
	\end{minipage} 
	\hfill
	\begin{minipage}[b]{0.45 \textwidth}
		\begin{align} \label{eq:Bell State Phi -}
			\ket{ \Phi^{ - } } = \ket{ \beta_{ 10 } } = \frac { \ket{ 0 }_{ A } \ket{ 0 }_{ B } - \ket{ 1 }_{ A } \ket{ 1 }_{ B } } { \sqrt{ 2 } }
		\end{align}
	\end{minipage}
	\begin{minipage}[b]{0.475 \textwidth}
		\begin{align} \label{eq:Bell State Psi +}
			\ket{ \Psi^{ + } } = \ket{ \beta_{ 01 } } = \frac { \ket{ 0 }_{ A } \ket{ 1 }_{ B } + \ket{ 1 }_{ A } \ket{ 0 }_{ B } } { \sqrt{ 2 } }
		\end{align}
	\end{minipage} 
	\hfill
	\begin{minipage}[b]{0.45 \textwidth}
		\begin{align} \label{eq:Bell State Psi -}
			\ket{ \Psi^{ - } } = \ket{ \beta_{ 11 } } = \frac { \ket{ 0 }_{ A } \ket{ 1 }_{ B } - \ket{ 1 }_{ A } \ket{ 0 }_{ B } } { \sqrt{ 2 } }
		\end{align}
	\end{minipage}
	
\end{tcolorbox}

For existing quantum computers that use the circuit model, it is quite easy to generate Bell states. The proposed protocol relies on $\ket{ \Psi^{ + } } = \frac { \ket{ 0 }_{ A } \ket{ 1 }_{ B } + \ket{ 1 }_{ A } \ket{ 0 }_{ B } } { \sqrt{ 2 } }$ pairs. Apart from $\ket{ \Psi^{ + } }$ pairs, we shall make use of qubits in another well-known state, namely $\ket{+}$. For convenience, we recall the definition of $\ket{+}$

\begin{align}
	\ket{ + } = H \ket{ 0 } = \frac { \ket{ 0 } + \ket{ 1 } } { \sqrt{ 2 } }
	\label{eq:Ket +}
	\ ,
\end{align}

which can be immediately produced by applying the Hadamard transform on $\ket{ 0 }$. Let us also clarify that all quantum registers are measured with respect to the computational basis $B = \{ \ket{ 0 }, \ket{ 1 } \}$.

\subsection{Detectable Byzantine Agreement} \label{subsec:Detectable Byzantine Agreement}

In the landmark paper \cite{Lamport1982}, the authors used the term the ``Byzantine Generals Problem'' to formulate the problem we study in this paper. A little later, two other influential works \cite{Dolev1983, Fischer1986} changed the terminology a bit and refereed to the same problem as the ``Byzantine Agreement Problem,'' a term that has been consistently used in the literature ever since. The problem itself involves $n$ generals of the Byzantine empire, one of them being the commanding general and the rest $n - 1$ being lieutenant generals. The commanding general must communicate his order to the lieutenant generals who are in different geographical locations. The overall situation is further complicated by the existence of traitors among the generals (possibly including even the commanding general). The notion of Byzantine Agreement was defined in \cite{Lamport1982} as follows.

\begin{definition} [Byzantine Agreement] \label{def:Byzantine Agreement} \
	A protocol achieves \emph{Byzantine Agreement} (BA) if it satisfies the following conditions.
\end{definition}

\begin{enumerate} [ left = 0.75 cm, labelsep = 0.50 cm ]
	\renewcommand\labelenumi{(\textbf{BA}$_\theenumi$)}
	\item	All loyal lieutenant generals follow the same order.
	\item	When the commanding general is loyal, all loyal lieutenant generals follow the commanding general's order.
\end{enumerate}

The authors in \cite{Lamport1982} not only presented an algorithm that, under certain assumptions, achieves BA, but also obtained the significant result that, under the assumption of pairwise authenticated channels among the generals, BA can be attained if and only if $t < \frac { n } { 3 }$, where $n$ is the number of generals and $t$ is the number of traitors. Later it was established that this bound can't be improved even if additional resources are assumed (see \cite{Fitzi2001a}). An important special illustration of this fact, is the case of $n = 3$ generals, exactly one of which is a traitor because it is impossible to achieve BA in such a case, by any classical protocol.

In the seminal paper \cite{Fitzi2001} a variant of BA called Detectable Byzantine Agreement (DBA from now on) was introduced. By slightly relaxing the requirements of the BA protocol, and, in particular, allowing the loyal generals to abort the protocol, an action that is not permitted in the original protocol, the authors were able to give the first quantum protocol for DBA. Their protocol relied on Aharonov states and improved the previous bound by achieving agreement in the special case of one traitor among three generals. Although \cite{Fitzi2001} considered only the case of $3$ generals, later works introduced DBA protocols for $n$ parties. In this work we use the Definition \ref{def:Detectable Byzantine Agreement} for DBA.

\begin{definition} [Detectable Byzantine Agreement] \label{def:Detectable Byzantine Agreement} \
	A protocol achieves \emph{Detectable Byzantine Agreement} (DBA) if it satisfies the following conditions.
\end{definition}

\begin{enumerate} [ left = 1.00 cm, labelsep = 0.50 cm ]
	\renewcommand\labelenumi{(\textbf{DBA}$_\theenumi$)}
	\item	If all generals are loyal, the protocol achieves Byzantine Agreement.
	\item	\textbf{Consistency}. All loyal generals either follow the same order or abort the protocol.
	\item	\textbf{Validity}. If the commanding general is loyal, then either all loyal lieutenant generals follow the commanding general’s order or abort the protocol.
\end{enumerate}

A comparison of Definition \ref{def:Byzantine Agreement} with Definition \ref{def:Detectable Byzantine Agreement} shows that the critical difference between BA and DBA lies in the extra capability of the generals to abort the protocol in the DBA case. As is common practice, we shall assume that the commanding general's orders are either $0$ or $1$, and we shall use the symbol $\bot$ to signify that the decision to abort.

\subsection{Assumptions and setting} \label{subsec:Assumptions and Setting}

In this section we begin the presentation of our EPR-based protocol for Detectable Byzantine Agreement, or EPRQDBA for short. For the sake of completeness, we explicitly state the assumptions that underlie the execution of the EPRQDBA protocol.

\begin{enumerate} [ left = 0.50 cm, labelsep = 0.50 cm ]
	\renewcommand\labelenumi{(\textbf{A}$_\theenumi$)}
	\item	There is a ``quantum source'' responsible for generating single qubits in the $\ket{ + }$ state and EPR pairs entangled in the $\ket{ \Psi^{ + } }$ state. The source distributes these qubits to each general through a quantum channel.
	\item	A complete network of pairwise authenticated classical channels connects all generals.
	\item	The classical network is synchronous coordinated by a global clock.
	\item	The protocol unfolds in rounds. In every round each general may receive messages, perform computations, and send messages. The messages sent during the current round are guaranteed to have arrived to their intended recipients by the beginning of the next round.
	\item	All measurements are performed with respect to the computational basis $B = \{ \ket{ 0 }, \ket{ 1 } \}$.
\end{enumerate}

We follow the tradition and describe the EPRQDBA protocol as a game. The players in this game are the $n$ spatially distributed generals. It will be convenient to divide them into two groups. The commanding general alone comprises the first group. From now on the famous Alice will play the role of the commanding general. The $n - 1$ lieutenant generals make up the second group. In the special case of $n = 3$ that we shall examine shortly, the venerable Bob and Charlie will play the $2$ lieutenant generals. In the general case, where $n > 3$, we will mostly employ generic names such LT$_0$, \dots, LT$_{n - 2}$, LT standing of course for ``lieutenant general.'' Occasionally, to emphasize some point, we may use again Bob or Charlie. To make the presentation easier to follow, we shall first show how the EPRQDBA protocol works in the special case of only $3$ players, namely Alice, Bod and Charlie.

\section{The $3$ player EPRQDBA protocol} \label{sec:The $3$ Player EPRQDBA Protocol}

\subsection{Entanglement distribution phase} \label{subsec:Entanglement Distribution Phase}

The EPRQDBA protocol can be conceptually organized into $3$ distinct phases. The first is the entanglement distribution phase, which refers to the generation and distribution of qubits and entangled EPR pairs. As we have briefly explained in assumption ($\mathbf{A}_{ 1 }$), we assume the existence of a trusted quantum source that undertakes this role. It is a relatively easy task, in view of the capabilities of modern quantum apparatus. Hence, the quantum source will have no difficulty in producing

\begin{itemize}
	\item	$2 m$ qubits in the $\ket{+}$ state by simply applying the Hadamard transform on $\ket{ 0 }$, and
	\item	$2 m$ EPR pairs in the $\ket{ \Psi^{ + } }$ state, which can be easily generated by contemporary quantum computers.
\end{itemize}

In the above scheme, the parameter $m$ is a properly chosen positive integer. Further suggestions regarding appropriate values of $m$ are provided in the mathematical analysis of the protocol.

Afterwards, the source transmits the qubits to the intended recipients through the quantum channels following the pattern outlined below.

\begin{enumerate} [ left = 0.70 cm, labelsep = 0.50 cm ]
	\renewcommand\labelenumi{(\textbf{$3$D}$_\theenumi$)}
	\item	The $2 m$ EPR pairs are numbered from $0$ to $2 m - 1$.
	\item	From every EPR pair $k$, where $0 \leq k \leq 2 m - 1$, the first qubit is sent to Alice.
	\item	The source sends to Bob $2 m$ qubits: $q'_{ 0 }$, $q'_{ 1 }$, \dots, $q'_{ 2 m - 2 }$, $q'_{ 2 m - 1 }$. The transmission alternates qubits in the $\ket{+}$ state (in the even positions of this sequence) with the second qubit of the odd-numbered EPR pairs (in the odd positions of this sequence).
	\item	Symmetrically, the source sends to Charlie $2 m$ qubits: $q''_{ 0 }$, $q''_{ 1 }$, \dots, $q''_{ 2 m - 2 }$, $q''_{ 2 m - 1 }$. In the even positions of this sequence, the source inserts the second qubit of the even-numbered EPR pairs and in the odd positions qubits in the $\ket{+}$ state.
\end{enumerate}

\begin{tcolorbox}
	[
		grow to left by = 1.25 cm,
		grow to right by = 1.25 cm,
		colback = MagentaVeryLight!12,			
		enhanced jigsaw,						
		sharp corners,
		toprule = 1.0 pt,
		bottomrule = 1.0 pt,
		leftrule = 0.1 pt,
		rightrule = 0.1 pt,
		sharp corners,
		center title,
		fonttitle = \bfseries
	]
	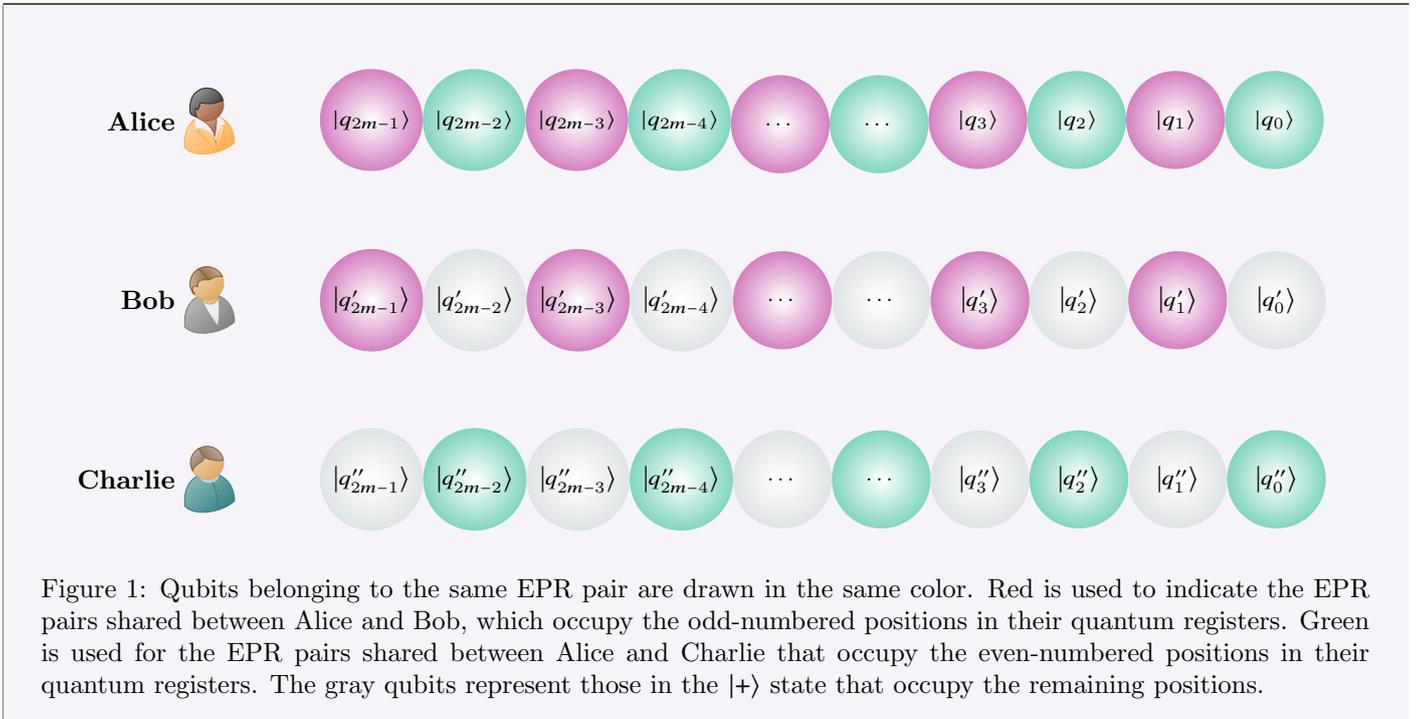
\begin{figure}[H]
		\centering
		\begin{tikzpicture} [ scale = 0.50 ]
			\node
			[
			alice,
			scale = 1.75,
			anchor = center,
			label = { [ label distance = 0.00 cm ] west: \textbf{Alice} }
			]
			(Alice) { };
			\matrix
			[
			matrix of nodes, nodes in empty cells,
			column sep = 0.000 mm, right = 1.00 of Alice,
			nodes = { circle, minimum size = 13 mm, semithick, font = \footnotesize },
			]
			(mat)
			{
				\node [ shade, outer color = RedPurple!50, inner color = white ] { $\ket{ q_{ 2 m - 1 } }$ }; &
				\node [ shade, outer color = GreenLighter2!50, inner color = white ] { $\ket{ q_{ 2 m - 2 } }$ }; &
				\node [ shade, outer color = RedPurple!50, inner color = white ] { $\ket{ q_{ 2 m - 3 } }$ }; &
				\node [ shade, outer color = GreenLighter2!50, inner color = white ] { $\ket{ q_{ 2 m - 4 } }$ }; &
				\node [ shade, outer color = RedPurple!50, inner color = white ] { \dots }; &
				\node [ shade, outer color = GreenLighter2!50, inner color = white ] { \dots }; &
				\node [ shade, outer color = RedPurple!50, inner color = white ] { $\ket{ q_{ 3 } }$ }; &
				\node [ shade, outer color = GreenLighter2!50, inner color = white ] { $\ket{ q_{ 2 } }$ }; &
				\node [ shade, outer color = RedPurple!50, inner color = white ] { $\ket{ q_{ 1 } }$ }; &
				\node [ shade, outer color = GreenLighter2!50, inner color = white ] { $\ket{ q_{ 0 } }$ };
				\\
			};
			\node
			[
			bob,
			scale = 1.75,
			anchor = center,
			below = 1.50 cm of Alice,
			label = { [ label distance = 0.00 cm ] west: \textbf{Bob} }
			]
			(Bob) { };
			\matrix
			[
			column sep = 0.000 mm, right = 1.00 of Bob,
			nodes = { circle, minimum size = 13 mm, semithick, font = \footnotesize },
			]
			(BobMatrix)
			{
				\node [ shade, outer color = RedPurple!50, inner color = white ] { $\ket{ q'_{ 2 m - 1 } }$ }; &
				\node [ shade, outer color = WordIceBlue, inner color = white ] { $\ket{ q'_{ 2 m - 2 } }$ }; &
				\node [ shade, outer color = RedPurple!50, inner color = white ] { $\ket{ q'_{ 2 m - 3 } }$ }; &
				\node [ shade, outer color = WordIceBlue, inner color = white ] { $\ket{ q'_{ 2 m - 4 } }$ }; &
				\node [ shade, outer color = RedPurple!50, inner color = white ] { \dots }; &
				\node [ shade, outer color = WordIceBlue, inner color = white ] { \dots }; &
				\node [ shade, outer color = RedPurple!50, inner color = white ] { $\ket{ q'_{ 3 } }$ }; &
				\node [ shade, outer color = WordIceBlue, inner color = white ] { $\ket{ q'_{ 2 } }$ }; &
				\node [ shade, outer color = RedPurple!50, inner color = white ] { $\ket{ q'_{ 1 } }$ }; &
				\node [ shade, outer color = WordIceBlue, inner color = white ] { $\ket{ q'_{ 0 } }$ }; \\
			};
			\node
			[
			charlie,
			scale = 1.75,
			anchor = center,
			below = 1.50 cm of Bob,
			label = { [ label distance = 0.00 cm ] west: \textbf{Charlie} }
			]
			(Charlie) { };
			\matrix
			[
			column sep = 0.000 mm, right = 1.00 of Charlie,
			nodes = { circle, minimum size = 13 mm, semithick, font = \footnotesize },
			]
			{
				\node [ shade, outer color = WordIceBlue, inner color = white ] { $\ket{ q''_{ 2 m - 1 } }$ }; &
				\node [ shade, outer color = GreenLighter2!50, inner color = white ] { $\ket{ q''_{ 2 m - 2 } }$ }; &
				\node [ shade, outer color = WordIceBlue, inner color = white ] { $\ket{ q''_{ 2 m - 3 } }$ }; &
				\node [ shade, outer color = GreenLighter2!50, inner color = white ] { $\ket{ q''_{ 2 m - 4 } }$ }; &
				\node [ shade, outer color = WordIceBlue, inner color = white ] { \dots }; &
				\node [ shade, outer color = GreenLighter2!50, inner color = white ] { \dots }; &
				\node [ shade, outer color = WordIceBlue, inner color = white ] { $\ket{ q''_{ 3 } }$ }; &
				\node [ shade, outer color = GreenLighter2!50, inner color = white ] { $\ket{ q''_{ 2 } }$ }; &
				\node [ shade, outer color = WordIceBlue, inner color = white ] { $\ket{ q''_{ 1 } }$ }; &
				\node [ shade, outer color = GreenLighter2!50, inner color = white ] { $\ket{ q''_{ 0 } }$ }; \\
			};
		\end{tikzpicture}
		\caption{Qubits belonging to the same EPR pair are drawn in the same color. Red is used to indicate the EPR pairs shared between Alice and Bob, which occupy the odd-numbered positions in their quantum registers. Green is used for the EPR pairs shared between Alice and Charlie that occupy the even-numbered positions in their quantum registers. The gray qubits represent those in the $\ket{+}$ state that occupy the remaining positions.}
		\label{fig:Alice Bob and Charlie's Quantum Registers}
	\end{figure}
\end{tcolorbox}

The end result is that the  quantum registers of Alice, Bob and Charlie are populated as shown in Figure \ref{fig:Alice Bob and Charlie's Quantum Registers}. Qubits of the same EPR pair are shown in the same color. Red is used to indicate the EPR pairs shared between Alice and Bob, which occupy the odd-numbered positions in their quantum registers. Analogously, green is used for the EPR pairs shared between Alice and Charlie that occupy the even-numbered positions in their quantum registers. The gray qubits represent those in the $\ket{+}$ state that occupy the remaining positions in Bob and Charlie's registers.

\subsection{Entanglement verification phase} \label{subsec:Entanglement Verification Phase}

This phase is crucial because the entire protocol is based on entanglement. If entanglement is not guaranteed, then agreement cannot not be guaranteed either. Obviously, the verification process can lead to two dramatically different outcomes. If entanglement verification is successfully established, then the EPRQDBA protocol is certain to achieve agreement. Failure of verification implies absence of the necessary entanglement. This~could be attributed either to noisy quantum channels or insidious sabotage by an active adversary. Whatever the true reason is, the only viable solution is to abort the current execution of the protocol, and initiate the whole procedure again from scratch, after taking some corrective measures.

This phase is very important because if entanglement is broken, then agreement can't be guaranteed. For this reason, the entanglement verification phase has been extensively analyzed in the relative literature. The~EPRQDBA protocol adheres to the previously established methods that have been introduced in previous works, such as~\cite{Fitzi2001,Cabello2003,Neigovzen2008,Feng2019,Qu2023}. In order to avoid repeating well-known techniques, we refer the reader to the these papers that describe in detail the implementation of this phase.

\subsection{Agreement phase} \label{subsec:Agreement Phase}

The EPRQDBA protocol achieves detectable agreement during its third and last phase, aptly named \textbf{agreement phase}. Alice, Bob and Charlie initiate the agreement phase by measuring their quantum registers. The distribution scheme, as analyzed in Subsection \ref{subsec:Entanglement Distribution Phase}, leads to some crucial correlations among the contents of Alice, Bob and Charlie's registers.

\begin{definition} \label{def:Measured Contents of the Registers} \
	Let the bit vectors $\mathbf { a }$, $\mathbf { b }$, and $\mathbf { c }$ denote the contents of Alice, Bob and Charlie's registers after the measurement, and assume that their explicit form is as given below:
	\begin{align}
		\mathbf { a }
		&=
		\underbrace { a_{ 2 m - 1 } a_{ 2 m - 2 } }_{ \text{ pair } m - 1 } \
		\underbrace { a_{ 2 m - 3 } a_{ 2 m - 4 } }_{ \text{ pair } m - 2 } \
		\dots \
		\underbrace { a_{ 3 } a_{ 2 } }_{ \text{ pair } 1 } \
		\underbrace { a_{ 1 } a_{ 0 } }_{ \text{ pair } 0 }
		\ , \label{eq:Alice's Bit Vector a - 1}
		\\
		\mathbf { b }
		&=
		\underbrace { b_{ 2 m - 1 } b_{ 2 m - 2 } }_{ \text{ pair } m - 1 } \
		\underbrace { b_{ 2 m - 3 } b_{ 2 m - 4 } }_{ \text{ pair } m - 2 } \
		\dots \
		\underbrace { b_{ 3 } b_{ 2 } }_{ \text{ pair } 1 } \
		\underbrace { b_{ 1 } b_{ 0 } }_{ \text{ pair } 0 }
		\ , \label{eq:Bob's Bit Vector b - 1}
		\\
		\mathbf { c }
		&=
		\underbrace { c_{ 2 m - 1 } c_{ 2 m - 2 } }_{ \text{ pair } m - 1 } \
		\underbrace { c_{ 2 m - 3 } c_{ 2 m - 4 } }_{ \text{ pair } m - 2 } \
		\dots \
		\underbrace { c_{ 3 } c_{ 2 } }_{ \text{ pair } 1 } \
		\underbrace { c_{ 1 } c_{ 0 } }_{ \text{ pair } 0 }
		\ . \label{eq:Charlie's Bit Vector c - 1}
	\end{align}
	The $k^{ th }$ pair of $\mathbf { a }$, $0 \leq k \leq m - 1$, is the pair of bits $a_{ 2 k + 1 } a_{ 2 k }$, and is designated by $\mathbf { a }_{ k }$. Similarly, the $k^{ th }$ pairs of $\mathbf { b }$ and $\mathbf { c }$ are designated by $\mathbf { b }_{ k }$ and $\mathbf { c }_{ k }$. Hence, $\mathbf { a }$, $\mathbf { b }$, and $\mathbf { c }$ can be written succinctly as:
	\begin{align}
		\mathbf { a }
		&=
		\mathbf { a }_{ m - 1 } \ \mathbf { a }_{ m - 2 } \ \dots \ \mathbf { a }_{ 1 } \ \mathbf { a }_{ 0 }
		\ , \label{eq:Alice's Bit Vector a - 2}
		\\
		\mathbf { b }
		&=
		\mathbf { b }_{ m - 1 } \ \mathbf { b }_{ m - 2 } \ \dots \ \mathbf { b }_{ 1 } \ \mathbf { b }_{ 0 }
		\ , \label{eq:Bob's Bit Vector b - 2}
		\\
		\mathbf { c }
		&=
		\mathbf { c }_{ m - 1 } \ \mathbf { c }_{ m - 2 } \ \dots \ \mathbf { c }_{ 1 } \ \mathbf { c }_{ 0 }
		\ . \label{eq:Charlie's Bit Vector c - 2}
	\end{align}
	We also define a special type of pair, termed the \emph{uncertain} pair, that is denoted by $\mathbf { u } = \sqcup \sqcup$, where $\sqcup$ is a new symbol, different from $0$ and $1$. In contrast, pairs consisting of bits $0$ and/or $1$ are called \emph{definite} pairs.
\end{definition}

When we want refer to either Bob's or Charlie's bit vector, but without specifying precisely which one, we will designate it by
\begin{align}
	\mathbf { l }
	&=
	\mathbf { l }_{ m - 1 } \ \mathbf { l }_{ m - 2 } \ \dots \ \mathbf { l }_{ 1 } \ \mathbf { l }_{ 0 }
	\nonumber \\
	&=
	\underbrace { l_{ 2 m - 1 } l_{ 2 m - 2 } }_{ \text{ pair } m - 1 } \
	\underbrace { l_{ 2 m - 3 } l_{ 2 m - 4 } }_{ \text{ pair } m - 2 } \
	\dots \
	\underbrace { l_{ 3 } l_{ 2 } }_{ \text{ pair } 1 } \
	\underbrace { l_{ 1 } l_{ 0 } }_{ \text{ pair } 0 }
	\ . \label{eq:Bob or Charlie's Bit Vector}
\end{align}

According to the distribution scheme, Alice and Bob share the odd-numbered $\ket{ \Psi^{ + } } = \frac { \ket{ 0 }_{ A } \ket{ 1 }_{ B } + \ket{ 1 }_{ A } \ket{ 0 }_{ B } } { \sqrt{ 2 } }$ pairs, and Alice and Charlie share the even-numbered $\ket{ \Psi^{ + } } = \frac { \ket{ 0 }_{ A } \ket{ 1 }_{ C } + \ket{ 1 }_{ A } \ket{ 0 }_{ C } } { \sqrt{ 2 } }$ pairs. Therefore, the next Lemma \ref{lem:Pair Differentiation Property} holds. Its proof is trivial and is omitted.

\begin{lemma} [Pair Differentiation Property] \label{lem:Pair Differentiation Property}
	The next property, termed \emph{pair differentiation} property, characterizes the corresponding bits and pairs of the bit vectors $\mathbf { a }$, $\mathbf { b }$, and $\mathbf { c }$.
\end{lemma}
\begin{enumerate}
	\item	For every odd-numbered bit $b_{ k }$ in $\mathbf { b }$, where $k = 1, 3, \dots, 2 m - 1$, it holds that
	\begin{align}
		b_{ k } &= \overline{ a_{ k } }, \ k = 1, 3, \dots, 2 m - 1
		\label{eq:Alice - Bod Odd Bit Relation}
		\ .
	\end{align}
	\item	For every even-numbered bit $c_{ k }$ in $\mathbf { c }$, where $k = 0, 2, \dots, 2 m - 2$, it holds that
	\begin{align}
		c_{ k } &= \overline{ a_{ k } }, \ k = 0, 2, \dots, 2 m - 2
		\label{eq:Alice - Charlie Odd Bit Relation}
		\ .
	\end{align}
	\item	Every $\mathbf { a }_{ k }$ pair, where $0 \leq k \leq m - 1$, differs from the corresponding pairs $\mathbf { b }_{ k }$ and $\mathbf { c }_{ k }$.
	\begin{align}
		\mathbf { a }_{ k } \neq \mathbf { b }_{ k }
		\ \text{ and } \
		\mathbf { a }_{ k } \neq \mathbf { c }_{ k }
		\ , \text{ for every } k, \ 0 \leq k \leq m - 1 \ .
		\label{eq:Alice - Bob - Charlie Pair Differentiation Property}
	\end{align}
\end{enumerate}

In the above relations (\ref{eq:Alice - Bod Odd Bit Relation}) and (\ref{eq:Alice - Charlie Odd Bit Relation}), $\overline{ a_{ k } }$ is the negation of $a_{ k }$.

Visually, the situation regarding the correlations among corresponding bits is shown in Figure \ref{fig:Alice Bob and Charlie's Quantum Registrs After Measurement}. Again bits resulting from the same EPR pair are shown in the same color.

\begin{tcolorbox}
	[
		grow to left by = 1.25 cm,
		grow to right by = 1.25 cm,
		colback = Brown!03,				
		enhanced jigsaw,				
		sharp corners,
		toprule = 1.0 pt,
		bottomrule = 1.0 pt,
		leftrule = 0.1 pt,
		rightrule = 0.1 pt,
		sharp corners,
		center title,
		fonttitle = \bfseries
	]
	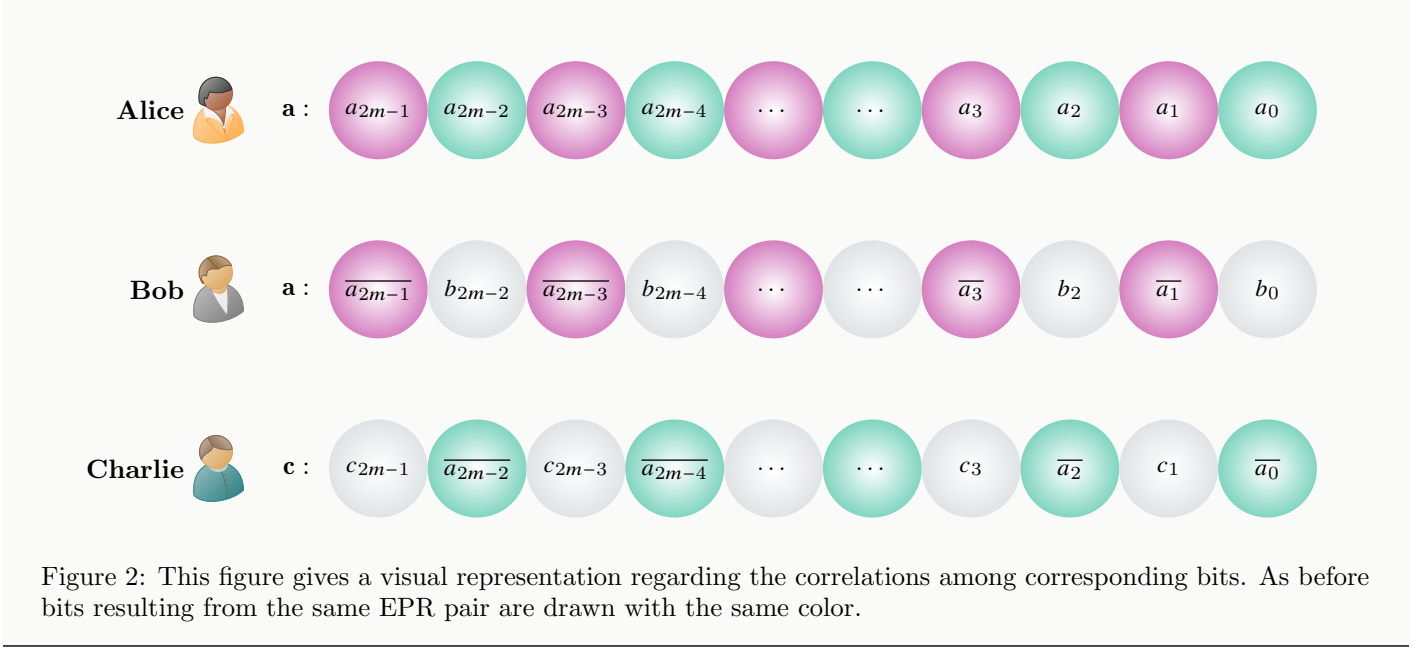
\begin{figure}[H]
		\centering
		\begin{tikzpicture} [ scale = 0.50 ]
			\node
			[
			alice,
			scale = 1.75,
			anchor = center,
			label = { [ label distance = 0.00 cm ] west: \textbf{Alice} }
			]
			(Alice) { };
			\matrix
			[
			column sep = 0.000 mm, right = 1.00 of Alice,
			nodes = { circle, minimum size = 13 mm, font = \small },
			label = { [ label distance = 0.00 cm ] west: $\mathbf{a}:$ }
			]
			{
				\node [ shade, outer color = RedPurple!50, inner color = white ] { $a_{ 2 m - 1 }$ }; &
				\node [ shade, outer color = GreenLighter2!50, inner color = white ] { $a_{ 2 m - 2 }$ }; &
				\node [ shade, outer color = RedPurple!50, inner color = white ] { $a_{ 2 m - 3 }$ }; &
				\node [ shade, outer color = GreenLighter2!50, inner color = white ] { $a_{ 2 m - 4 }$ }; &
				\node [ shade, outer color = RedPurple!50, inner color = white ] { \dots }; &
				\node [ shade, outer color = GreenLighter2!50, inner color = white ] { \dots }; &
				\node [ shade, outer color = RedPurple!50, inner color = white ] { $a_{ 3 }$ }; &
				\node [ shade, outer color = GreenLighter2!50, inner color = white ] { $a_{ 2 }$ }; &
				\node [ shade, outer color = RedPurple!50, inner color = white ] { $a_{ 1 }$ }; &
				\node [ shade, outer color = GreenLighter2!50, inner color = white ] { $a_{ 0 }$ }; &
				\\
			};
			\node
			[
			bob,
			scale = 1.75,
			anchor = center,
			below = 1.50 cm of Alice,
			label = { [ label distance = 0.00 cm ] west: \textbf{Bob} }
			]
			(Bob) { };
			\matrix
			[
			column sep = 0.000 mm, right = 1.00 of Bob,
			nodes = { circle, minimum size = 13 mm, font = \small },
			label = { [ label distance = 0.00 cm ] west: $\mathbf{a}:$ }
			]
			{
				\node [ shade, outer color = RedPurple!50, inner color = white ] { $\overline{ a_{ 2 m - 1 } }$ }; &
				\node [ shade, outer color = WordIceBlue, inner color = white ] { $b_{ 2 m - 2 }$ }; &
				\node [ shade, outer color = RedPurple!50, inner color = white ] { $\overline{ a_{ 2 m - 3 } }$ }; &
				\node [ shade, outer color = WordIceBlue, inner color = white ] { $b_{ 2 m - 4 }$ }; &
				\node [ shade, outer color = RedPurple!50, inner color = white ] { \dots }; &
				\node [ shade, outer color = WordIceBlue, inner color = white ] { \dots }; &
				\node [ shade, outer color = RedPurple!50, inner color = white ] { $\overline{ a_{ 3 } }$ }; &
				\node [ shade, outer color = WordIceBlue, inner color = white ] { $b_{ 2 }$ }; &
				\node [ shade, outer color = RedPurple!50, inner color = white ] { $\overline{ a_{ 1 } }$ }; &
				\node [ shade, outer color = WordIceBlue, inner color = white ] { $b_{ 0 }$ }; &
				\\
			};
			\node
			[
			charlie,
			scale = 1.75,
			anchor = center,
			below = 1.50 cm of Bob,
			label = { [ label distance = 0.00 cm ] west: \textbf{Charlie} }
			]
			(Charlie) { };
			\matrix
			[
			column sep = 0.000 mm, right = 1.00 of Charlie,
			nodes = { circle, minimum size = 13 mm, font = \small },
			label = { [ label distance = 0.00 cm ] west: $\mathbf{c}:$ }
			]
			{
				\node [ shade, outer color = WordIceBlue, inner color = white ] { $c_{ 2 m - 1 }$ }; &
				\node [ shade, outer color = GreenLighter2!50, inner color = white ] { $\overline{ a_{ 2 m - 2 } }$ }; &
				\node [ shade, outer color = WordIceBlue, inner color = white ] { $c_{ 2 m - 3 }$ }; &
				\node [ shade, outer color = GreenLighter2!50, inner color = white ] { $\overline{ a_{ 2 m - 4 } }$ }; &
				\node [ shade, outer color = WordIceBlue, inner color = white ] { \dots }; &
				\node [ shade, outer color = GreenLighter2!50, inner color = white ] { \dots }; &
				\node [ shade, outer color = WordIceBlue, inner color = white ] { $c_{ 3 }$ }; &
				\node [ shade, outer color = GreenLighter2!50, inner color = white ] { $\overline{ a_{ 2 } }$ }; &
				\node [ shade, outer color = WordIceBlue, inner color = white ] { $c_{ 1 }$ }; &
				\node [ shade, outer color = GreenLighter2!50, inner color = white ] { $\overline{ a_{ 0 } }$ }; &
				\\
			};
		\end{tikzpicture}
		\caption{This figure gives a visual representation regarding the correlations among corresponding bits. As before bits resulting from the same EPR pair are drawn with the same color.}
		\label{fig:Alice Bob and Charlie's Quantum Registrs After Measurement}
	\end{figure}
\end{tcolorbox}

\begin{definition} [Command Vectors] \label{def:Command Vectors} \
	Alice sends to Bob and Charlie either the command $0$ or the command $1$. Besides her command, as a ``proof'', she also sends an appropriate \emph{command vector}. The idea is that the command vector for Bob is always different from the command vector for Charlie, even when the command is the same. For the command $0$, the command vectors for Bob and Charlie are $\vmathbb{ 0 }_{ B }$ and $\vmathbb{ 0 }_{ C }$, respectively, whereas for the command $1$ the corresponding command vectors are $\mathds{ 1 }_{ B }$ and $\mathds{ 1 }_{ C }$. The explicit form of the command vectors is the following.
	\begin{align}
		\vmathbb{ 0 }_{ B }
		&=
		\mathbf { v }_{ m - 1 } \ \mathbf { v }_{ m - 2 } \ \dots \ \mathbf { v }_{ 1 } \ \mathbf { v }_{ 0 }
		\ ,
		\ \text{ where } \
		\mathbf { v }_{ k }
		=
		\left\{
		\begin{matrix*}[l]
			\mathbf { a }_{ k } = a_{ 2 k + 1 } a_{ 2 k } & \text{ if } a_{ 2 k + 1 } = 0 \\
			\mathbf { u } = \sqcup \sqcup & \text{ if } a_{ 2 k + 1 } \neq 0
		\end{matrix*}
		\right.
		\ , \ 0 \leq k \leq m - 1 \ .
		\label{eq:Bob's Command Vector 0}
		\\
		\vmathbb{ 0 }_{ C }
		&=
		\mathbf { v }_{ m - 1 } \ \mathbf { v }_{ m - 2 } \ \dots \ \mathbf { v }_{ 1 } \ \mathbf { v }_{ 0 }
		\ ,
		\ \text{ where } \
		\mathbf { v }_{ k }
		=
		\left\{
		\begin{matrix*}[l]
			\mathbf { a }_{ k } = a_{ 2 k + 1 } a_{ 2 k } & \text{ if } a_{ 2 k } = 0 \\
			\mathbf { u } = \sqcup \sqcup & \text{ if } a_{ 2 k } \neq 0
		\end{matrix*}
		\right.
		\ , \ 0 \leq k \leq m - 1 \ .
		\label{eq:Charlie's Command Vector 0}
		\\
		\mathds{ 1 }_{ B }
		&=
		\mathbf { v }_{ m - 1 } \ \mathbf { v }_{ m - 2 } \ \dots \ \mathbf { v }_{ 1 } \ \mathbf { v }_{ 0 }
		\ ,
		\ \text{ where } \
		\mathbf { v }_{ k }
		=
		\left\{
		\begin{matrix*}[l]
			\mathbf { a }_{ k } = a_{ 2 k + 1 } a_{ 2 k } & \text{ if } a_{ 2 k + 1 } = 1 \\
			\mathbf { u } = \sqcup \sqcup & \text{ if } a_{ 2 k + 1 } \neq 1
		\end{matrix*}
		\right.
		\ , \ 0 \leq k \leq m - 1 \ .
		\label{eq:Bob's Command Vector 1}
		\\
		\mathds{ 1 }_{ C }
		&=
		\mathbf { v }_{ m - 1 } \ \mathbf { v }_{ m - 2 } \ \dots \ \mathbf { v }_{ 1 } \ \mathbf { v }_{ 0 }
		\ ,
		\ \text{ where } \
		\mathbf { v }_{ k }
		=
		\left\{
		\begin{matrix*}[l]
			\mathbf { a }_{ k } = a_{ 2 k + 1 } a_{ 2 k } & \text{ if } a_{ 2 k } = 1 \\
			\mathbf { u } = \sqcup \sqcup & \text{ if } a_{ 2 k } \neq 1
		\end{matrix*}
		\right.
		\ , \ 0 \leq k \leq m - 1 \ .
		\label{eq:Charlie's Command Vector 1}
	\end{align}
	A command vector, besides pairs containing $0$ and $1$ bits, also contains an approximately equal number of pairs consisting of $\sqcup$ characters. When we want refer to a command vector, but without providing further details, we will designate it by
	\begin{align}
		\mathbf { v }
		&=
		\mathbf { v }_{ m - 1 } \ \mathbf { v }_{ m - 2 } \ \dots \ \mathbf { v }_{ 1 } \ \mathbf { v }_{ 0 }
		\nonumber \\
		&=
		\underbrace { v_{ 2 m - 1 } v_{ 2 m - 2 } }_{ \text{ pair } m - 1 } \
		\underbrace { v_{ 2 m - 3 } v_{ 2 m - 4 } }_{ \text{ pair } m - 2 } \
		\dots \
		\underbrace { v_{ 3 } v_{ 2 } }_{ \text{ pair } 1 } \
		\underbrace { v_{ 1 } v_{ 0 } }_{ \text{ pair } 0 }
		\ . \label{eq:Generic Command Vector for Bob and Charlie}
	\end{align}
\end{definition}

Given a command vector or a bit vector, we define the set containing the positions of the pairs that consist of a given combination of bits.

\begin{definition} [Pair Designation] \label{def:Pair Designation} \
	Given a command vector $\mathbf { v }$ or a bit vector $\mathbf { l }$, we define the set $\mathbb{ P }_{ x, y } ( \mathbf { v } )$ and $\mathbb{ P }_{ x, y } ( \mathbf { l } )$, respectively, of the positions of those pairs consisting precisely of the bits $x, y$.
\end{definition}

\begin{example} [Illustrating the concepts] \label{xmp:Illustration of Concepts}
	This first example is designed to illustrate all the previous concepts. For practical purposes, i.e., to fit in a page, we take $m = 12$.
	\begin{tcolorbox}
		[
			grow to left by = 0.00 cm,
			grow to right by = 0.00 cm,
			colback = Brown!03,			
			enhanced jigsaw,			
			sharp corners,
			toprule = 1.0 pt,
			bottomrule = 1.0 pt,
			leftrule = 0.1 pt,
			rightrule = 0.1 pt,
			sharp corners,
			center title,
			fonttitle = \bfseries
		]
		\begin{figure}[H]
			\centering
			\begin{tikzpicture} [ scale = 0.40 ]
				\node
				[
				shade, top color = GreenTeal, bottom color = black, rectangle, text width = 4.45 cm, align = center
				]
				(Label)
				{ \color{white} Alice, Bob \& Charlie's registers after measurement };
				\node
				[
				alice,
				scale = 1.50,
				anchor = center,
				below = 1.00 cm of Label,
				label = { [ label distance = 0.00 cm ] north: \textbf{Alice} }
				]
				(Alice) { };
				\node
				[
				bob,
				scale = 1.50,
				anchor = center,
				left = 3.50 cm of Alice,
				label = { [ label distance = 0.00 cm ] north: \textbf{Bob} }
				]
				(Bob) { };
				\node
				[
				charlie,
				scale = 1.50,
				anchor = center,
				right = 3.50 cm of Alice,
				label = { [ label distance = 0.00 cm ] north: \textbf{Charlie} }
				]
				(Charlie) { };
				\matrix
				[
				column sep = 0.00 mm, below = 0.50 cm of Alice,
				nodes = { draw = black, fill = none, minimum size = 7 mm, semithick, align = center, font = \scriptsize }
				]
				(AREG)
				{
					\node [ shade, outer color = GreenLighter2!50, inner color = white ] { 1 }; \\
					\node [ shade, outer color = RedPurple!50, inner color = white ] { 0 }; \\
					\node [ shade, outer color = GreenLighter2!50, inner color = white ] { 0 }; \\
					\node [ shade, outer color = RedPurple!50, inner color = white ] { 1 }; \\
					\node [ shade, outer color = GreenLighter2!50, inner color = white ] { 0 }; \\
					\node [ shade, outer color = RedPurple!50, inner color = white ] { 0 }; \\
					\node [ shade, outer color = GreenLighter2!50, inner color = white ] { 1 }; \\
					\node [ shade, outer color = RedPurple!50, inner color = white ] { 1 }; \\
					\node [ shade, outer color = GreenLighter2!50, inner color = white ] { 0 }; \\
					\node [ shade, outer color = RedPurple!50, inner color = white ] { 0 }; \\
					\node [ shade, outer color = GreenLighter2!50, inner color = white ] { 0 }; \\
					\node [ shade, outer color = RedPurple!50, inner color = white ] { 0 }; \\
					\node [ shade, outer color = GreenLighter2!50, inner color = white ] { 1 }; \\
					\node [ shade, outer color = RedPurple!50, inner color = white ] { 1 }; \\
					\node [ shade, outer color = GreenLighter2!50, inner color = white ] { 1 }; \\
					\node [ shade, outer color = RedPurple!50, inner color = white ] { 0 }; \\
					\node [ shade, outer color = GreenLighter2!50, inner color = white ] { 0 }; \\
					\node [ shade, outer color = RedPurple!50, inner color = white ] { 1 }; \\
					\node [ shade, outer color = GreenLighter2!50, inner color = white ] { 0 }; \\
					\node [ shade, outer color = RedPurple!50, inner color = white ] { 0 }; \\
					\node [ shade, outer color = GreenLighter2!50, inner color = white ] { 1 }; \\
					\node [ shade, outer color = RedPurple!50, inner color = white ] { 1 }; \\
					\node [ shade, outer color = GreenLighter2!50, inner color = white ] { 0 }; \\
					\node [ shade, outer color = RedPurple!50, inner color = white ] { 1 }; \\
				};
				\matrix
				[
				column sep = 0.000 mm, left = 0.05 cm of AREG,
				nodes = { fill = none, minimum size = 7.20 mm, semithick, align = center, font = \scriptsize }
				]
				(Index)
				{
					\node { 0 }; \\
					\node { 1 }; \\
					\node { 2 }; \\
					\node { 3 }; \\
					\node { 4 }; \\
					\node { 5 }; \\
					\node { 6 }; \\
					\node { 7 }; \\
					\node { 8 }; \\
					\node { 9 }; \\
					\node { 10 }; \\
					\node { 11 }; \\
					\node { 12 }; \\
					\node { 13 }; \\
					\node { 14 }; \\
					\node { 15 }; \\
					\node { 16 }; \\
					\node { 17 }; \\
					\node { 18 }; \\
					\node { 19 }; \\
					\node { 20 }; \\
					\node { 21 }; \\
					\node { 22 }; \\
					\node { 23 }; \\
				};
				\matrix
				[
				column sep = 0.00 mm, below = 0.50 cm of Bob,
				nodes = { draw = black, fill = none, minimum size = 7 mm, semithick, align = center, font = \scriptsize }
				]
				(BREG)
				{
					\node [ shade, outer color = WordIceBlue, inner color = white ] { 1 }; \\
					\node [ shade, outer color = RedPurple!50, inner color = white ] { 1 }; \\
					\node [ shade, outer color = WordIceBlue, inner color = white ] { 0 }; \\
					\node [ shade, outer color = RedPurple!50, inner color = white ] { 0 }; \\
					\node [ shade, outer color = WordIceBlue, inner color = white ] { 0 }; \\
					\node [ shade, outer color = RedPurple!50, inner color = white ] { 1 }; \\
					\node [ shade, outer color = WordIceBlue, inner color = white ] { 1 }; \\
					\node [ shade, outer color = RedPurple!50, inner color = white ] { 0 }; \\
					\node [ shade, outer color = WordIceBlue, inner color = white ] { 1 }; \\
					\node [ shade, outer color = RedPurple!50, inner color = white ] { 1 }; \\
					\node [ shade, outer color = WordIceBlue, inner color = white ] { 0 }; \\
					\node [ shade, outer color = RedPurple!50, inner color = white ] { 1 }; \\
					\node [ shade, outer color = WordIceBlue, inner color = white ] { 1 }; \\
					\node [ shade, outer color = RedPurple!50, inner color = white ] { 0 }; \\
					\node [ shade, outer color = WordIceBlue, inner color = white ] { 0 }; \\
					\node [ shade, outer color = RedPurple!50, inner color = white ] { 1 }; \\
					\node [ shade, outer color = WordIceBlue, inner color = white ] { 1 }; \\
					\node [ shade, outer color = RedPurple!50, inner color = white ] { 0 }; \\
					\node [ shade, outer color = WordIceBlue, inner color = white ] { 1 }; \\
					\node [ shade, outer color = RedPurple!50, inner color = white ] { 1 }; \\
					\node [ shade, outer color = WordIceBlue, inner color = white ] { 1 }; \\
					\node [ shade, outer color = RedPurple!50, inner color = white ] { 0 }; \\
					\node [ shade, outer color = WordIceBlue, inner color = white ] { 0 }; \\
					\node [ shade, outer color = RedPurple!50, inner color = white ] { 0 }; \\
				};
				\matrix
				[
				column sep = 0.00 mm, below = 0.50 cm of Charlie,
				nodes = { draw = black, fill = none, minimum size = 7 mm, semithick, align = center, font = \scriptsize }
				]
				(CREG)
				{
					\node [ shade, outer color = GreenLighter2!50, inner color = white ] { 0 }; \\
					\node [ shade, outer color = WordIceBlue, inner color = white ] { 1 }; \\
					\node [ shade, outer color = GreenLighter2!50, inner color = white ] { 1 }; \\
					\node [ shade, outer color = WordIceBlue, inner color = white ] { 0 }; \\
					\node [ shade, outer color = GreenLighter2!50, inner color = white ] { 1 }; \\
					\node [ shade, outer color = WordIceBlue, inner color = white ] { 1 }; \\
					\node [ shade, outer color = GreenLighter2!50, inner color = white ] { 0 }; \\
					\node [ shade, outer color = WordIceBlue, inner color = white ] { 0 }; \\
					\node [ shade, outer color = GreenLighter2!50, inner color = white ] { 1 }; \\
					\node [ shade, outer color = WordIceBlue, inner color = white ] { 1 }; \\
					\node [ shade, outer color = GreenLighter2!50, inner color = white ] { 1 }; \\
					\node [ shade, outer color = WordIceBlue, inner color = white ] { 0 }; \\
					\node [ shade, outer color = GreenLighter2!50, inner color = white ] { 0 }; \\
					\node [ shade, outer color = WordIceBlue, inner color = white ] { 1 }; \\
					\node [ shade, outer color = GreenLighter2!50, inner color = white ] { 0 }; \\
					\node [ shade, outer color = WordIceBlue, inner color = white ] { 0 }; \\
					\node [ shade, outer color = GreenLighter2!50, inner color = white ] { 1 }; \\
					\node [ shade, outer color = WordIceBlue, inner color = white ] { 1 }; \\
					\node [ shade, outer color = GreenLighter2!50, inner color = white ] { 1 }; \\
					\node [ shade, outer color = WordIceBlue, inner color = white ] { 1 }; \\
					\node [ shade, outer color = GreenLighter2!50, inner color = white ] { 0 }; \\
					\node [ shade, outer color = WordIceBlue, inner color = white ] { 0 }; \\
					\node [ shade, outer color = GreenLighter2!50, inner color = white ] { 1 }; \\
					\node [ shade, outer color = WordIceBlue, inner color = white ] { 0 }; \\
				};
			\end{tikzpicture}
			\caption{This figure shows the contents of Alice, Bob and Charlie's registers after the measurement.}
			\label{fig: Example of Register Contents after Measurement}
		\end{figure}
	\end{tcolorbox}

	\noindent In a real implementation $m$ should certainly be greater (see Table \refeq{tbl:Numerical Characteristics of the EPRQDBA Protocol}). Let us assume that the contents of Alice, Bob and Charlie's registers are those shown in Figure \ref{fig: Example of Register Contents after Measurement}. By Definition \ref{def:Measured Contents of the Registers}, we may write that
	{\small
		\begin{align}
			\mathbf { a }
			&=
			\underbrace { 1 0 }_{ \text{ pair } 11 } \
			\underbrace { 1 1 }_{ \text{ pair } 10 } \
			\underbrace { 0 0 }_{ \text{ pair } 9 } \
			\underbrace { 1 0 }_{ \text{ pair } 8 } \
			\underbrace { 0 1 }_{ \text{ pair } 7 } \
			\underbrace { 1 1 }_{ \text{ pair } 6 } \
			\underbrace { 0 0 }_{ \text{ pair } 5 } \
			\underbrace { 0 0 }_{ \text{ pair } 4 } \
			\underbrace { 1 1 }_{ \text{ pair } 3 } \
			\underbrace { 0 0 }_{ \text{ pair } 2 } \
			\underbrace { 1 0 }_{ \text{ pair } 1 } \
			\underbrace { 0 1 }_{ \text{ pair } 0 } \
			\ , \label{eq:Alice's Bit Vector a - Example 1 - 1}
			\\
			\mathbf { b }
			&=
			\underbrace { 0 0 }_{ \text{ pair } 11 } \
			\underbrace { 0 1 }_{ \text{ pair } 10 } \
			\underbrace { 1 1 }_{ \text{ pair } 9 } \
			\underbrace { 0 1 }_{ \text{ pair } 8 } \
			\underbrace { 1 0 }_{ \text{ pair } 7 } \
			\underbrace { 0 1 }_{ \text{ pair } 6 } \
			\underbrace { 1 0 }_{ \text{ pair } 5 } \
			\underbrace { 1 1 }_{ \text{ pair } 4 } \
			\underbrace { 0 1 }_{ \text{ pair } 3 } \
			\underbrace { 1 0 }_{ \text{ pair } 2 } \
			\underbrace { 0 0 }_{ \text{ pair } 1 } \
			\underbrace { 1 1 }_{ \text{ pair } 0 } \
			\ , \label{eq:Bob's Bit Vector b - Example 1}
			\\
			\mathbf { c }
			&=
			\underbrace { 0 1 }_{ \text{ pair } 11 } \
			\underbrace { 0 0 }_{ \text{ pair } 10 } \
			\underbrace { 1 1 }_{ \text{ pair } 9 } \
			\underbrace { 1 1 }_{ \text{ pair } 8 } \
			\underbrace { 0 0 }_{ \text{ pair } 7 } \
			\underbrace { 1 0 }_{ \text{ pair } 6 } \
			\underbrace { 0 1 }_{ \text{ pair } 5 } \
			\underbrace { 1 1 }_{ \text{ pair } 4 } \
			\underbrace { 0 0 }_{ \text{ pair } 3 } \
			\underbrace { 1 1 }_{ \text{ pair } 2 } \
			\underbrace { 0 1 }_{ \text{ pair } 1 } \
			\underbrace { 1 0 }_{ \text{ pair } 0 } \
			\ , \label{eq:Charlie's Bit Vector c - Example 1}
		\end{align}
	}
	where the bit vectors $\mathbf { a }$, $\mathbf { b }$, and $\mathbf { c }$ denote the contents of Alice, Bob and Charlie's registers after the measurement.

	In this figure, the correlated pairs of bits are drawn with the same color. Specifically, as in all previous figures, red is used to indicate the pairs of bits that originated from the measurement of the EPR pairs shared between Alice and Bob, and which occupy the odd-numbered positions in their registers. Analogously, green depicts the pairs of bits that originated from the measurement of the EPR pairs shared between Alice and Charlie, which occupy the even-numbered positions in their registers. The remaining gray bits that appear in Bob and Charlie's registers come from qubits that were initially in the $\ket{+}$ state.

	By a simple observation of the red and green pairs of bits, we immediately see that Lemma \ref{lem:Pair Differentiation Property} is verified. In every red pair, the two bits are complementary, i.e., the second bit is the negation of the first. This is due to the fact that their values are the result of the measurement of a $\ket{ \Psi^{ + } }$ pair. The red pairs correlate Alice and Bob's registers. Exactly the same holds for the green pairs, which correlate Alice and Charlie's registers. The initial presence of the $\ket{ \Psi^{ + } }$ pairs in the entanglement distribution phase, has manifested into classical correlation among the contents of Alice, Bob, and Charlie's registers during the agreement phase.

	There are $12$ pairs in $\mathbf { a }$, $\mathbf { b }$, and $\mathbf { c }$. In each of them, the $k^{ th }$ pair, $0 \leq k \leq 11$, contains the bits that occupy positions $2 k + 1$ and $2 k$, as shown in Figure \ref{fig: Example of Register Contents after Measurement}. The bit in position $2 k + 1$ is the most significant bit of the $k^{ th }$ pair and the bit in position $2 k$ is the least significant bit.

	Let us suppose that Alice sends the order $0$. According to Definition \ref{def:Command Vectors}, Alice must prove her order by sending the command vectors $\vmathbb{ 0 }_{ B }$ and $\vmathbb{ 0 }_{ C }$ to Bob and Charlie, respectively. The command vector $\vmathbb{ 0 }_{ B }$ is constructed by including the pairs in $\mathbf { a }$ with most significant bit $0$ (irrespective of what the least significant bit is) and filling the gaps with the uncertain pair. Symmetrically, $\vmathbb{ 0 }_{ C }$ is constructed by including the pairs in $\mathbf { a }$ with least significant bit $0$ (irrespective of what the most significant bit is) and filling the gaps with the uncertain pair. Thus, $\vmathbb{ 0 }_{ B }$ and $\vmathbb{ 0 }_{ C }$ are as given below (we have also repeated $\mathbf { a }$ to facilitate the comparison).
	{\small
		\begin{align}
			\mathbf { a }
			&=
			\underbrace { 1 0 }_{ \text{ pair } 11 } \
			\underbrace { 1 1 }_{ \text{ pair } 10 } \
			\underbrace { 0 0 }_{ \text{ pair } 9 } \
			\underbrace { 1 0 }_{ \text{ pair } 8 } \
			\underbrace { 0 1 }_{ \text{ pair } 7 } \
			\underbrace { 1 1 }_{ \text{ pair } 6 } \
			\underbrace { 0 0 }_{ \text{ pair } 5 } \
			\underbrace { 0 0 }_{ \text{ pair } 4 } \
			\underbrace { 1 1 }_{ \text{ pair } 3 } \
			\underbrace { 0 0 }_{ \text{ pair } 2 } \
			\underbrace { 1 0 }_{ \text{ pair } 1 } \
			\underbrace { 0 1 }_{ \text{ pair } 0 } \
			\ , \label{eq:Alice's Bit Vector a - Example 1 - 2}
			\\
			\vmathbb{ 0 }_{ B }
			&=
			\underbrace { \sqcup \sqcup }_{ \text{ pair } 11 } \
			\underbrace { \sqcup \sqcup }_{ \text{ pair } 10 } \
			\underbrace { 0 0 }_{ \text{ pair } 9 } \
			\underbrace { \sqcup \sqcup }_{ \text{ pair } 8 } \
			\underbrace { 0 1 }_{ \text{ pair } 7 } \
			\underbrace { \sqcup \sqcup }_{ \text{ pair } 6 } \
			\underbrace { 0 0 }_{ \text{ pair } 5 } \
			\underbrace { 0 0 }_{ \text{ pair } 4 } \
			\underbrace { \sqcup \sqcup }_{ \text{ pair } 3 } \
			\underbrace { 0 0 }_{ \text{ pair } 2 } \
			\underbrace { \sqcup \sqcup }_{ \text{ pair } 1 } \
			\underbrace { 0 1 }_{ \text{ pair } 0 } \
			\ , \label{eq:Alice's Command Vector 0_B - Example 1}
			\\
			\vmathbb{ 0 }_{ C }
			&=
			\underbrace { 1 0 }_{ \text{ pair } 11 } \
			\underbrace { \sqcup \sqcup }_{ \text{ pair } 10 } \
			\underbrace { 0 0 }_{ \text{ pair } 9 } \
			\underbrace { 1 0 }_{ \text{ pair } 8 } \
			\underbrace { \sqcup \sqcup }_{ \text{ pair } 7 } \
			\underbrace { \sqcup \sqcup }_{ \text{ pair } 6 } \
			\underbrace { 0 0 }_{ \text{ pair } 5 } \
			\underbrace { 0 0 }_{ \text{ pair } 4 } \
			\underbrace { \sqcup \sqcup }_{ \text{ pair } 3 } \
			\underbrace { 0 0 }_{ \text{ pair } 2 } \
			\underbrace { 1 0 }_{ \text{ pair } 1 } \
			\underbrace { \sqcup \sqcup }_{ \text{ pair } 0 } \
			\ . \label{eq:Bob's Command Vector 0_C - Example 1}
		\end{align}
	}
	Recalling Definition \ref{def:Pair Designation}, we see that the sets $\mathbb{ P }_{ x, y } ( \vmathbb{ 0 }_{ B } )$ and $\mathbb{ P }_{ x, y } ( \vmathbb{ 0 }_{ C } )$ contain the positions in $\vmathbb{ 0 }_{ B }$ and $\vmathbb{ 0 }_{ C }$, respectively, of the pairs consisting of the bits $x, y$. In view of (\ref{eq:Alice's Command Vector 0_B - Example 1}) and (\ref{eq:Bob's Command Vector 0_C - Example 1}), we derive that

	\begin{minipage}{0.485 \textwidth}
		\begin{itemize}
			\item	$\mathbb{ P }_{ 0, 0 } ( \vmathbb{ 0 }_{ B } ) = \{ 2, 4, 5, 9 \}$
			\item	$\mathbb{ P }_{ 0, 1 } ( \vmathbb{ 0 }_{ B } ) = \{ 0, 7 \}$
		\end{itemize}
	\end{minipage}
	\begin{minipage}{0.485 \textwidth}
		\begin{itemize}
			\item	$\mathbb{ P }_{ 0, 0 } ( \vmathbb{ 0 }_{ C } ) = \{ 2, 4, 5, 9 \}$
			\item	$\mathbb{ P }_{ 1, 0 } ( \vmathbb{ 0 }_{ C } ) = \{ 1, 8, 11 \}$
		\end{itemize}
	\end{minipage}

	Let us now suppose that Alice sends the order $1$. In this case, Alice must prove her order by sending the command vectors $\mathds{ 1 }_{ B }$ and $\mathds{ 1 }_{ C }$ to Bob and Charlie, respectively. The command vector $\mathds{ 1 }_{ B }$ is constructed by including the pairs in $\mathbf { a }$ with most significant bit $1$ (irrespective of what the least significant bit is) and filling the gaps with the uncertain pair. Symmetrically, $\mathds{ 1 }_{ C }$ is constructed by including the pairs in $\mathbf { a }$ with least significant bit $1$ (irrespective of what the most significant bit is) and filling the gaps with the uncertain pair. The resulting $\mathds{ 1 }_{ B }$ and $\mathds{ 1 }_{ C }$ are shown below (again, we have also repeated $\mathbf { a }$ to facilitate the comparison).
	{\small
		\begin{align}
			\mathbf { a }
			&=
			\underbrace { 1 0 }_{ \text{ pair } 11 } \
			\underbrace { 1 1 }_{ \text{ pair } 10 } \
			\underbrace { 0 0 }_{ \text{ pair } 9 } \
			\underbrace { 1 0 }_{ \text{ pair } 8 } \
			\underbrace { 0 1 }_{ \text{ pair } 7 } \
			\underbrace { 1 1 }_{ \text{ pair } 6 } \
			\underbrace { 0 0 }_{ \text{ pair } 5 } \
			\underbrace { 0 0 }_{ \text{ pair } 4 } \
			\underbrace { 1 1 }_{ \text{ pair } 3 } \
			\underbrace { 0 0 }_{ \text{ pair } 2 } \
			\underbrace { 1 0 }_{ \text{ pair } 1 } \
			\underbrace { 0 1 }_{ \text{ pair } 0 } \
			\ , \label{eq:Alice's Bit Vector a - Example 1 - 3}
			\\
			\mathds{ 1 }_{ B }
			&=
			\underbrace { 1 0 }_{ \text{ pair } 11 } \
			\underbrace { 1 1 }_{ \text{ pair } 10 } \
			\underbrace { \sqcup \sqcup }_{ \text{ pair } 9 } \
			\underbrace { 1 0 }_{ \text{ pair } 8 } \
			\underbrace { \sqcup \sqcup }_{ \text{ pair } 7 } \
			\underbrace { 1 1 }_{ \text{ pair } 6 } \
			\underbrace { \sqcup \sqcup }_{ \text{ pair } 5 } \
			\underbrace { \sqcup \sqcup }_{ \text{ pair } 4 } \
			\underbrace { 1 1 }_{ \text{ pair } 3 } \
			\underbrace { \sqcup \sqcup }_{ \text{ pair } 2 } \
			\underbrace { 1 0 }_{ \text{ pair } 1 } \
			\underbrace { \sqcup \sqcup }_{ \text{ pair } 0 } \
			\ , \label{eq:Alice's Command Vector 1_B - Example 1}
			\\
			\mathds{ 1 }_{ C }
			&=
			\underbrace { \sqcup \sqcup }_{ \text{ pair } 11 } \
			\underbrace { 1 1 }_{ \text{ pair } 10 } \
			\underbrace { \sqcup \sqcup }_{ \text{ pair } 9 } \
			\underbrace { \sqcup \sqcup }_{ \text{ pair } 8 } \
			\underbrace { 0 1 }_{ \text{ pair } 7 } \
			\underbrace { 1 1 }_{ \text{ pair } 6 } \
			\underbrace { \sqcup \sqcup }_{ \text{ pair } 5 } \
			\underbrace { \sqcup \sqcup }_{ \text{ pair } 4 } \
			\underbrace { 1 1 }_{ \text{ pair } 3 } \
			\underbrace { \sqcup \sqcup }_{ \text{ pair } 2 } \
			\underbrace { \sqcup \sqcup }_{ \text{ pair } 1 } \
			\underbrace { 0 1 }_{ \text{ pair } 0 } \
			\ . \label{eq:Bob's Command Vector 1_C - Example 1}
		\end{align}
	}
	The sets $\mathbb{ P }_{ x, y } ( \mathds{ 1 }_{ B } )$ and $\mathbb{ P }_{ x, y } ( \mathds{ 1 }_{ C } )$ contain the positions in $\mathds{ 1 }_{ B }$ and $\mathds{ 1 }_{ C }$, respectively, of the $x, y$ pairs. Taking into account (\ref{eq:Alice's Command Vector 1_B - Example 1}) and (\ref{eq:Bob's Command Vector 1_C - Example 1}), we see that

	\begin{minipage}{0.485 \textwidth}
		\begin{itemize}
			\item	$\mathbb{ P }_{ 1, 0 } ( \mathds{ 1 }_{ B } ) = \{ 1, 8, 11 \}$
			\item	$\mathbb{ P }_{ 1, 1 } ( \mathds{ 1 }_{ B } ) = \{ 3, 6, 10 \}$
		\end{itemize}
	\end{minipage}
	\begin{minipage}{0.485 \textwidth}
		\begin{itemize}
			\item	$\mathbb{ P }_{ 0, 1 } ( \mathds{ 1 }_{ C } ) = \{ 0, 7 \}$
			\item	$\mathbb{ P }_{ 1, 1 } ( \mathds{ 1 }_{ C } ) = \{ 3, 6, 10 \}$
		\end{itemize}
	\end{minipage}
	
	In the present example, the value of the parameter $m$ is $12$, which implies that $\frac { m } { 4 } = 3$. The cardinality of all the sets $\mathbb{ P }_{ 0, 0 } ( \vmathbb{ 0 }_{ B } )$, $\mathbb{ P }_{ 0, 1 } ( \vmathbb{ 0 }_{ B } )$, $\mathbb{ P }_{ 0, 0 } ( \vmathbb{ 0 }_{ C } )$, $\mathbb{ P }_{ 1, 0 } ( \vmathbb{ 0 }_{ C } )$, $\mathbb{ P }_{ 1, 0 } ( \mathds{ 1 }_{ B } )$, $\mathbb{ P }_{ 1, 1 } ( \mathds{ 1 }_{ B } )$, $\mathbb{ P }_{ 0, 1 } ( \mathds{ 1 }_{ C } )$, and $\mathbb{ P }_{ 1, 1 } ( \mathds{ 1 }_{ C } )$ is $\approx 3 = \frac { m } { 4 }$, which verifies the next Lemma \ref{lem:Loyal Command Vector Properties}. Of course, the exact cardinality of some of these sets deviates slightly from the expected value $3$, but this is not unexpected. As $m$ increases, the ratio of the deviation to $m$ will tend to $0$.
	\hfill $\triangleleft$
\end{example}

Based on Definition \ref{def:Pair Designation}, it is trivial to see that the next Lemma \ref{lem:Loyal Command Vector Properties} holds. The notation $| S |$ is employed to designate the cardinality, i.e., the number of elements, of an arbitrary set $S$.

\begin{lemma} [Loyal Command Vector Properties] \label{lem:Loyal Command Vector Properties} \
	If Alice is loyal, her command vectors $\vmathbb{ 0 }_{ B }$, $\vmathbb{ 0 }_{ C }$, $\mathds{ 1 }_{ B }$ and $\mathds{ 1 }_{ C }$ satisfy the following properties.
\end{lemma}

\begin{minipage}{0.485 \textwidth}
	\begin{itemize}
		\item	$| \mathbb{ P }_{ 0, 0 } ( \vmathbb{ 0 }_{ B } ) | \approx \frac { m } { 4 }$.
		\item	$| \mathbb{ P }_{ 0, 1 } ( \vmathbb{ 0 }_{ B } ) | \approx \frac { m } { 4 }$.
		\item	$| \mathbb{ P }_{ 1, 0 } ( \mathds{ 1 }_{ B } ) | \approx \frac { m } { 4 }$.
		\item	$| \mathbb{ P }_{ 1, 1 } ( \mathds{ 1 }_{ B } ) | \approx \frac { m } { 4 }$.
	\end{itemize}
\end{minipage}
\begin{minipage}{0.485 \textwidth}
	\begin{itemize}
		\item	$| \mathbb{ P }_{ 0, 0 } ( \vmathbb{ 0 }_{ C } ) | \approx \frac { m } { 4 }$.
		\item	$| \mathbb{ P }_{ 1, 0 } ( \vmathbb{ 0 }_{ C } ) | \approx \frac { m } { 4 }$.
		\item	$| \mathbb{ P }_{ 0, 1 } ( \mathds{ 1 }_{ C } ) | \approx \frac { m } { 4 }$.
		\item	$| \mathbb{ P }_{ 1, 1 } ( \mathds{ 1 }_{ C } ) | \approx \frac { m } { 4 }$.
	\end{itemize}
\end{minipage}

The agreement phase of the EPRQDBA protocol evolves in rounds. The actions taken by the Alice, Bob and Charlie in each round are described below.

\begin{tcolorbox}
	[
		enhanced,
		breakable,
		grow to left by = 0.00 cm,
		grow to right by = 0.00 cm,
		colback = WordIceBlue,			
		enhanced jigsaw,				
		sharp corners,
		toprule = 0.01 pt,
		bottomrule = 0.01 pt,
		leftrule = 0.1 pt,
		rightrule = 0.1 pt,
		sharp corners,
		center title,
		fonttitle = \bfseries
	]
	\begin{algorithm}[H]
		\SetAlgorithmName{Protocol}{ }{ }
		\setcounter{algocf}{0}
		\caption { \textsc { The $3$ Player EPRQDBA Protocol } }
		\label{alg:The $3$ Player EPRQDBA Protocol}
	\end{algorithm}
	\begin{enumerate} [ left = 0.90 cm, labelsep = 0.50 cm ]
		\renewcommand\labelenumi{(\textbf{Round}$_\theenumi$)}
		\item	Alice sends to Bob and Charlie her order and the appropriate command vector as proof.
		\item	Bob uses the \textsc{CheckAlice} Algorithm \ref{alg:The $3$ Player CheckAlice Algorithm} to check the received command vector against his own bit vector $\mathbf { b }$ for inconsistencies.
		\begin{itemize}
			\item[$\star$]	If no inconsistencies are found, Bob preliminary accepts the received order.
			\item[$\star$]	Otherwise, Bob's preliminary decision is to abort ($\bot$).
		\end{itemize}
		Upon completion of the verification procedure, Bob sends to Charlie his preliminary decision together with the received command vector.
		
		Symmetrically and in parallel, Charlie performs the same actions.
		\item	Bob compares his own initial decision with that of Charlie's.
		\begin{itemize}
			\item[$\star$]	\colorbox{WordLightGreen}{\textbf{Rule}$_{ 3, 1 }$}: If Charlie's decisions coincides with Bob's preliminary decision, then Bob accepts as final his preliminary decision and terminates the protocol on his end.
			\item[$\star$]	In case Charlie's initial decision is different, then the following cases are considered.
			\begin{itemize}
				\item[$\square$]	\colorbox{WordLightGreen}{\textbf{Rule}$_{ 3, 2 }$}: If Bob's decision is $0$ ($1$) and Charlie's decision is to abort, then Bob sticks to his initial decision and terminates the protocol.
				\item[$\square$]	If Bob's decision is $0$ ($1$) and Charlie's decision is $1$, then Bob uses the \textsc{CheckWCV} Algorithm \ref{alg:The $3$ Player CheckWCV Algorithm} to check Charlie's command vector against his own command vector for inconsistencies.
				\begin{itemize}
					\item[$\diamond$]	\colorbox{WordLightGreen}{\textbf{Rule}$_{ 3, 3 }$}: If no inconsistencies are detected, then Bob aborts and terminates the protocol.
					\item[$\diamond$]	\colorbox{WordLightGreen}{\textbf{Rule}$_{ 3, 4 }$}: Otherwise, Bob sticks to his preliminary decision and terminates the protocol.
				\end{itemize}
				\item[$\square$]	If Bob's decision is to abort and Charlie's decision is $0$ (or $1$), then Bob uses the \textsc{CheckWBV} Algorithm \ref{alg:The $3$ Player CheckWBV Algorithm} to check Charlie's command vector against his own bit vector for inconsistencies.
				\begin{itemize}
					\item[$\diamond$]	\colorbox{WordLightGreen}{\textbf{Rule}$_{ 3, 5 }$}: If no inconsistencies are found, then Bob changes his final decision to $0$ ($1$) and terminates the protocol.
					\item[$\diamond$]	\colorbox{WordLightGreen}{\textbf{Rule}$_{ 3, 6 }$}: Otherwise, Bob sticks to his initial decision to abort and terminates the protocol.
				\end{itemize}
			\end{itemize}
		\end{itemize}
		Charlie's actions mirror Bob's actions.
	\end{enumerate}
\end{tcolorbox}

Algorithms \ref{alg:The $3$ Player CheckAlice Algorithm} -- \ref{alg:The $3$ Player CheckWBV Algorithm} apply Lemmata \ref{lem:Pair Differentiation Property} and \ref{lem:Loyal Command Vector Properties} to determine whether a command vector is consistent, in which case they return TRUE. If they determine inconsistencies, they terminate and return FALSE. The names of Algorithms \ref{alg:The $3$ Player CheckWCV Algorithm} and \ref{alg:The $3$ Player CheckWBV Algorithm} are mnemonic abbreviations for ``Check with Command Vector'' and ``Check with Bit Vector.'' It is important to explain their difference.

Without loss of generality, let us suppose that Bob received the order $0$ and a consistent command vector $\mathbf { v }_{ A }$ from Alice, while Charlie claims that he received the order $1$. Charlie must convince Bob by sending Bob the command vector $\mathbf { v }$ he claims he received from Alice. Bob checks $\mathbf { v }$ for inconsistencies against his own $\mathbf { v }_{ A }$ by invoking the Algorithm \ref{alg:The $3$ Player CheckWCV Algorithm}. Now, imagine that Bob received an inconsistent command vector $\mathbf { v }_{ A }$ from Alice and his initial decision is to abort, while Charlie claims that he received the order $1$. Charlie must convince Bob by sending Bob the command vector $\mathbf { v }$ he claims he received from Alice. In this situation, Bob does not have a consistent command vector $\mathbf { v }_{ A }$ that he could use, so Bob must check $\mathbf { v }$ for inconsistencies against his own bit vector $\mathbf { b }$ by invoking the Algorithm \ref{alg:The $3$ Player CheckWBV Algorithm}.

In Algorithms \ref{alg:The $3$ Player CheckAlice Algorithm}, \ref{alg:The $3$ Player CheckWCV Algorithm}, and \ref{alg:The $3$ Player CheckWBV Algorithm}, we employ the following notation.
\begin{itemize}
	\item	$i, j \in \{ 0, 1 \}$ are the indices of Bob ($1$) and Charlie ($0$).
	\item	$c$, where $c = 0$ or $c = 1$, is the command being checked for consistency.
	\item	$\mathbf { v }_{ A } = v_{ 2 m - 1 } v_{ 2 m - 2 } \ \dots \ v_{ 1 } v_{ 0 }$ is the command vector sent by Alice.
	\item	$\mathbf { v }$, is the command vector sent by Bod to Charlie (or vice versa).
	\item	$\mathbf { l } = l_{ 2 m - 1 } l_{ 2 m - 2 } \ \dots \ l_{ 1 } l_{ 0 }$ is the bit vector of Bod (Charlie) who does the consistency checking.
	\item	$\bigtriangleup$ is the symmetric difference of two sets, i.e., $ S \bigtriangleup S' = \left( S \setminus S' \right) \bigcup \left( S' \setminus S \right)$, for given sets $S$ and $S'$.
\end{itemize}

\begin{algorithm}[H]
	\setcounter{algocf}{0}
	\SetKw{Break}{break}
	\caption{ \textsc{ CheckAlice } ( i, c, $\mathbf { v }_{ A }$, $\mathbf { l }$ ) }
	\label{alg:The $3$ Player CheckAlice Algorithm}
	\# When invoked by Bob $i = 1$
	\\
	\# When invoked by Charlie $i = 0$
	\\
	\If
	{
		$! \ ( \ | \ \mathbb{ P }_{ c, c } ( \mathbf { v }_{ A } ) \ | \approx \frac { m } { 4 } \ )$
	}
	{
		\Return FALSE
	}
	\If
	{
		$! \ ( \ | \ \mathbb{ P }_{ \overline{ c } \oplus i, c \oplus i } ( \mathbf { v }_{ A } ) \ | \approx \frac { m } { 4 } \ )$
	}
	{
		\Return FALSE
	}
	\For {$k = 0$ \KwTo $m - 1$}
	{
		\If { $( v_{ 2 k + i } == l_{ 2 k + i } )$ }
		{
			\Return FALSE
		}
	}
	\Return TRUE
\end{algorithm}

\begin{algorithm}[H]
	\SetKw{Break}{break}
	\caption{ \textsc{CheckWCV} (i, j, c, $\mathbf { v }$, $\mathbf { v }_{ A }$) }
	\label{alg:The $3$ Player CheckWCV Algorithm}
	\# When invoked by Bob to check Charlie $i = 1, j = 0$
	\\
	\# When invoked by Charlie to check Bob $i = 0, j = 1$
	\\
	\If
	{
		$! \ ( \ | \ \mathbb{ P }_{ c, c } ( \mathbf { v } ) \ | \approx \frac { m } { 4 } \ )$
	}
	{
		\Return FALSE
	}
	\If
	{
		$! \ ( \ | \ \mathbb{ P }_{ \overline{ c } \oplus j, c \oplus j } ( \mathbf { v } ) \ | \approx \frac { m } { 4 } \ )$
	}
	{
		\Return FALSE
	}
	\If
	{
		$! \ ( \ | \
		\mathbb{ P }_{ \overline{ c } \oplus j, c \oplus j } ( \mathbf { v }_{ A } )
		\bigtriangleup
		\mathbb{ P }_{ \overline{ c } \oplus j, c \oplus j } ( \mathbf { v } )
		\ | \approx 0 \ )$
	}
	{
		\Return FALSE
	}
	\Return TRUE
\end{algorithm}

\begin{algorithm}[H]
	\SetKw{Break}{break}
	\caption{ \textsc{CheckWBV} (i, j, c, $\mathbf { v }$, $\mathbf { l }$) }
	\label{alg:The $3$ Player CheckWBV Algorithm}
	\# When invoked by Bob to check Charlie $i = 1, j = 0$
	\\
	\# When invoked by Charlie to check Bob $i = 0, j = 1$
	\\
	\If
	{
		$! \ ( \ | \ \mathbb{ P }_{ c, c } ( \mathbf { v } ) \ | \approx \frac { m } { 4 } \ )$
	}
	{
		\Return FALSE
	}
	\If
	{
		$! \ ( \ | \ \mathbb{ P }_{ \overline{ c } \oplus j, c \oplus j } ( \mathbf { v } ) \ | \approx \frac { m } { 4 } \ )$
	}
	{
		\Return FALSE
	}
	\For {$k = 0$ \KwTo $m - 1$}
	{
		\If { $( v_{ 2 k + i } == l_{ 2 k + i } )$ }
		{
			\Return FALSE
		}
	}
	\Return TRUE
\end{algorithm}

In order to streamline the exposition of the text, all lengthy proofs are relocated in the Appendix \ref{appsec:Appendix - Main Results Proofs}.

In the literature, all works focusing on the $3$ player setting, assume that there is exactly one traitor among the $3$, who can either be the commanding general or one of the lieutenant generals. If we assume that there is precisely one traitor among the $3$, we can prove the following Propositions \ref{prp:Loyal Alice} and \ref{prp:Traitor Alice}.

\begin{proposition} [Loyal Alice] \label{prp:Loyal Alice}
	If Alice is loyal, the $3$ player EPRQDBA protocol will enable the loyal lieutenant general to agree with Alice. Specifically, if Alice and Bob are loyal and Charlie is a traitor, Bob will follow Alice's order. Symmetrically, if Alice and Charlie are loyal and Bob is a traitor, Charlie will follow Alice's order.
\end{proposition}

\begin{proposition} [Traitor Alice] \label{prp:Traitor Alice}
	If Bob and Charlie are loyal and Alice is a traitor, the $3$ player EPRQDBA protocol will enable Bob and Charlie to reach agreement, in the sense of both following the same order or both aborting.
\end{proposition}

\begin{example} [Illustrating the $3$ player protocol] \label{xmp:Illustration of the $3$ Player EPRQDBA Protocol}
	This second example is a continuation of the first. It aims to illustrate the operation of the $3$ Player EPRQDBA Protocol and give an intuitive explanation of why Propositions \ref{prp:Loyal Alice} and \ref{prp:Traitor Alice} are true.
	
	Let us continue Example \ref{xmp:Illustration of Concepts}, and examine first the scenario where a loyal Alice sends order $0$ to both Bob and Charlie accompanied by the command vectors $\vmathbb{ 0 }_{ B }$ and $\vmathbb{ 0 }_{ C }$ as given by (\ref{eq:Alice's Command Vector 0_B - Example 1}) and (\ref{eq:Bob's Command Vector 0_C - Example 1}), respectively. We may distinguish the following cases.
	\begin{itemize}
		\item	\emph{Both Bob and Charlie are loyal}. Then, according to \textbf{Rule}$_{ 3, 1 }$ of the $3$ player EPRQDBA Protocol \ref{alg:The $3$ Player EPRQDBA Protocol}, all $3$ players agree to execute order $0$.
		\item	\emph{There is one traitor among Bob and Charlie}. Without loss of generality we may assume that Charlie is the traitor, who tries to sabotage the agreement. What can he do?
		\begin{itemize}
			\item[$\diamond$]	Charlie may claim that he decided to abort because Alice sent him an inconsistent command vector. Then, according to \textbf{Rule}$_{ 3, 2 }$ of the $3$ player EPRQDBA Protocol \ref{alg:The $3$ Player EPRQDBA Protocol}, Bob will stick to his initial decision, and, together with Alice, will agree to execute order $0$. In this way, assumption (\textbf{DBA}$_{ 3 }$) of the detectable Byzantine agreement protocol, stipulating that commanding general is loyal, then either all loyal lieutenant generals follow the commanding general’s order or abort the protocol (recall Definition \ref{def:Detectable Byzantine Agreement}), is satisfied. The rationale behind this rule is that Bob, having already received a consistent command vector from Alice, is suspicious of Charlie's decision to abort, and decides to cling to his preliminary decision.
			\item[$\diamond$]	Charlie may claim that Alice gave the order $1$, along with a consistent command vector $\mathds{ 1 }_{ C }$. Bob will then use the \textsc{CheckWCV} Algorithm \ref{alg:The $3$ Player CheckWCV Algorithm} to check Charlie's command vector against his own command vector $\vmathbb{ 0 }_{ B }$ for inconsistencies. If Bob finds no inconsistencies, then, according to \textbf{Rule}$_{ 3, 3 }$ of the $3$ player EPRQDBA Protocol \ref{alg:The $3$ Player EPRQDBA Protocol}, he will abort. Obviously, in such a case, Charlie will have won in sabotaging the agreement. Let's analyze the probability of this event. First, let's recall from (\ref{eq:Bob's Command Vector 0_C - Example 1}) what Charlie knows with certainty.
			{\small
				\begin{align}
					\vmathbb{ 0 }_{ C }
					&=
					\underbrace { 1 0 }_{ \text{ pair } 11 } \
					\underbrace { \sqcup \sqcup }_{ \text{ pair } 10 } \
					\underbrace { 0 0 }_{ \text{ pair } 9 } \
					\underbrace { 1 0 }_{ \text{ pair } 8 } \
					\underbrace { \sqcup \sqcup }_{ \text{ pair } 7 } \
					\underbrace { \sqcup \sqcup }_{ \text{ pair } 6 } \
					\underbrace { 0 0 }_{ \text{ pair } 5 } \
					\underbrace { 0 0 }_{ \text{ pair } 4 } \
					\underbrace { \sqcup \sqcup }_{ \text{ pair } 3 } \
					\underbrace { 0 0 }_{ \text{ pair } 2 } \
					\underbrace { 1 0 }_{ \text{ pair } 1 } \
					\underbrace { \sqcup \sqcup }_{ \text{ pair } 0 } \
					\ . \label{eq:Bob's Command Vector 0_C - Example 2}
				\end{align}
			}
			He partially knows $\vmathbb{ 0 }_{ B }$ given by (\ref{eq:Alice's Command Vector 0_B - Example 1}), and in particular only the pairs consisting of $0 0$,
			{\small
				\begin{align}
					\vmathbb{ 0 }_{ B }
					&=
					\underbrace { \sqcup \sqcup }_{ \text{ pair } 11 } \
					\underbrace { \sqcup \sqcup }_{ \text{ pair } 10 } \
					\underbrace { 0 0 }_{ \text{ pair } 9 } \
					\underbrace { \sqcup \sqcup }_{ \text{ pair } 8 } \
					\underbrace { 0 1 }_{ \text{ pair } 7 } \
					\underbrace { \sqcup \sqcup }_{ \text{ pair } 6 } \
					\underbrace { 0 0 }_{ \text{ pair } 5 } \
					\underbrace { 0 0 }_{ \text{ pair } 4 } \
					\underbrace { \sqcup \sqcup }_{ \text{ pair } 3 } \
					\underbrace { 0 0 }_{ \text{ pair } 2 } \
					\underbrace { \sqcup \sqcup }_{ \text{ pair } 1 } \
					\underbrace { 0 1 }_{ \text{ pair } 0 } \
					\ , \label{eq:Alice's Command Vector 0_B - Example 2}
				\end{align}
			}
			and his problem is to construct something resembling the real $\mathds{ 1 }_{ C }$, as given from (\ref{eq:Bob's Command Vector 1_C - Example 1}).
			{\small
				\begin{align}
					\mathds{ 1 }_{ C }
					&=
					\underbrace { \sqcup \sqcup }_{ \text{ pair } 11 } \
					\underbrace { 1 1 }_{ \text{ pair } 10 } \
					\underbrace { \sqcup \sqcup }_{ \text{ pair } 9 } \
					\underbrace { \sqcup \sqcup }_{ \text{ pair } 8 } \
					\underbrace { 0 1 }_{ \text{ pair } 7 } \
					\underbrace { 1 1 }_{ \text{ pair } 6 } \
					\underbrace { \sqcup \sqcup }_{ \text{ pair } 5 } \
					\underbrace { \sqcup \sqcup }_{ \text{ pair } 4 } \
					\underbrace { 1 1 }_{ \text{ pair } 3 } \
					\underbrace { \sqcup \sqcup }_{ \text{ pair } 2 } \
					\underbrace { \sqcup \sqcup }_{ \text{ pair } 1 } \
					\underbrace { 0 1 }_{ \text{ pair } 0 } \
					\ . \label{eq:Bob's Command Vector 1_C - Example 2}
				\end{align}
			}
			Charlie knows that the least significant bit of the uncertain pairs in $\vmathbb{ 0 }_{ C }$ is $1$ and the most significant bit is $0$ or $1$, with equal probability $0.5$. However, Charlie cannot know with certainty if the most significant bit of a specific uncertain pair is $0$ or $1$. Therefore, when guessing $\mathds{ 1 }_{ C }$, Charlie can make two detectable mistakes: (i) place a $0$ in a wrong pair that is not actually contained in $\mathbb{ P }_{ 0, 1 } ( \vmathbb{ 0 }_{ B } )$, e.g., set the pair $10$ as $01$, or (ii) place a $1$ in a wrong pair that appears in $\mathbb{ P }_{ 0, 1 } ( \vmathbb{ 0 }_{ B } )$, e.g., set pair 7 as $11$. The situation from a probabilistic point of view resembles the probability of picking the one correct configuration out of many. The total number of configurations is equal to the number of ways to place $2$ or $3$ identical objects ($0$) into $5$ distinguishable boxes. The probability that Charlie places all the $0$s correctly is
			\begin{align}
				P( \text{ Charlie places all $0$ correctly } )
				=
				\frac { 1 }
				{ \binom { \ 5 \ } { \ 3 \ } }
				=
				0.1
				\ .
			\end{align}
			Thus, the probability that Charlie cheats successfully is just $0.1$, even in this toy scale example with $m = 12$. Consequently, Bob is bound to detect inconsistencies, which will lead him to apply \textbf{Rule}$_{ 3, 4 }$ of the $3$ player EPRQDBA Protocol \ref{alg:The $3$ Player EPRQDBA Protocol}, cling to his initial decision, and, ultimately, achieve agreement with Alice.
		\end{itemize}
	\end{itemize}
	
	The situation where a loyal Alice sends order $1$ to both Bob and Charlie is identical.
	
	Let us now examine the second major scenario where a traitor Alice intends to prevent the loyal Bob and Charlie from recanting agreement. In view of Definition \ref{def:Detectable Byzantine Agreement}, this precludes the cases where Alice sends to both of them inconsistent command vectors, or the same order, accompanied by appropriate consistent command vectors. In reviewing her remaining options, we distinguish the following cases.
	\begin{itemize}
		\item	\emph{Alice sends an order, say $0$, along with a consistent command vector to Bob and an inconsistent one to Charlie}. In this case, Bob will stick to his initial decision, but Charlie will use the \textsc{CheckWBV} Algorithm \ref{alg:The $3$ Player CheckWBV Algorithm} to verify Bob's command vector. In doing so, Charlie will crosscheck $\vmathbb{ 0 }_{ B }$, given by (\ref{eq:Alice's Command Vector 0_B - Example 1}), with his own bit vector $\mathbf { c }$ from (\ref{eq:Charlie's Bit Vector c - Example 1}).
		{\small
			\begin{align}
				\vmathbb{ 0 }_{ B }
				&=
				\underbrace { \sqcup \sqcup }_{ \text{ pair } 11 } \
				\underbrace { \sqcup \sqcup }_{ \text{ pair } 10 } \
				\underbrace { 0 0 }_{ \text{ pair } 9 } \
				\underbrace { \sqcup \sqcup }_{ \text{ pair } 8 } \
				\underbrace { 0 1 }_{ \text{ pair } 7 } \
				\underbrace { \sqcup \sqcup }_{ \text{ pair } 6 } \
				\underbrace { 0 0 }_{ \text{ pair } 5 } \
				\underbrace { 0 0 }_{ \text{ pair } 4 } \
				\underbrace { \sqcup \sqcup }_{ \text{ pair } 3 } \
				\underbrace { 0 0 }_{ \text{ pair } 2 } \
				\underbrace { \sqcup \sqcup }_{ \text{ pair } 1 } \
				\underbrace { 0 1 }_{ \text{ pair } 0 } \
				\ , \label{eq:Alice's Command Vector 0_B - Example 2 - 2}
				\\
				\mathbf { c }
				&=
				\underbrace { 0 1 }_{ \text{ pair } 11 } \
				\underbrace { 0 0 }_{ \text{ pair } 10 } \
				\underbrace { 1 1 }_{ \text{ pair } 9 } \
				\underbrace { 1 1 }_{ \text{ pair } 8 } \
				\underbrace { 0 0 }_{ \text{ pair } 7 } \
				\underbrace { 1 0 }_{ \text{ pair } 6 } \
				\underbrace { 0 1 }_{ \text{ pair } 5 } \
				\underbrace { 1 1 }_{ \text{ pair } 4 } \
				\underbrace { 0 0 }_{ \text{ pair } 3 } \
				\underbrace { 1 1 }_{ \text{ pair } 2 } \
				\underbrace { 0 1 }_{ \text{ pair } 1 } \
				\underbrace { 1 0 }_{ \text{ pair } 0 } \
				\ . \label{eq:Charlie's Bit Vector c - Example 2}
			\end{align}
		}
		Charlie will count the $00$ and $01$ pairs in $\vmathbb{ 0 }_{ B }$ to find that they are $4$ and $2$, pretty close to the expected value $3$, and in all these pairs the least significant bit is the negation of his corresponding bit, as dictated by the original $\ket{ \Psi^{ + } }$ entanglement. Hence, Charlie will successfully verify Bob's command vector, and, according to the \textbf{Rule}$_{ 3, 5 }$ of the $3$ player EPRQDBA Protocol \ref{alg:The $3$ Player EPRQDBA Protocol}, he will change his decision to that of Bob, thereby reaching the desired agreement. The rationale behind this rule is that Charlie, having already received an inconsistent command vector from Alice, is suspicious of Alice. His suspicion is confirmed when he verifies Bob's command vector, considering that it high unlikely for Bob to have forge a consistent command vector. Thus, he decides to change his preliminary decision to that of Bob.
		\item	\emph{Alice sends to Bob and Charlie different orders with consistent command vectors}. Let us assume, without loss of generality, that Alice sends the order $0$ together with a consistent command vector $\vmathbb{ 0 }_{ B }$ to Bob and the order $1$ together with a consistent command vector $\mathds{ 1 }_{ C }$ to Charlie. Bob knows the positions of almost all $00$ and $01$ pairs in $\mathbf { a }$. If Alice had forged even a single $10$ or $11$ pair, claiming to be either $00$ or $01$, then Bob, when using the \textsc{CheckAlice} Algorithm \ref{alg:The $3$ Player CheckAlice Algorithm}, he would have immediately detected the inconsistency. Symmetrically, Charlie knows the positions of almost all $01$ and $11$ pairs in $\mathbf { a }$. According to the rules of the protocol, Bob will sent to Charlie the command vector $\vmathbb{ 0 }_{ B }$, and, simultaneously, Charlie will send to Bob the command vector $\mathds{ 1 }_{ C }$. Both will use the \textsc{CheckWCV} Algorithm \ref{alg:The $3$ Player CheckWCV Algorithm} to verify the consistency of the others' command vector. Then, according to \textbf{Rule}$_{ 3, 3 }$, both Bob and Charlie will abort, again fulfilling the requirements of DBA.\hfill $\triangleleft$
	\end{itemize}
\end{example}

Let us ponder the question of what will happen if there are $2$ traitors and only one loyal general. It might be that the $2$ traitors are the $2$ lieutenant generals, or the commanding general and one of the lieutenant generals. This case is actually very easy as explained in the next Corollary \ref{crl:The Case of $2$ Traitors}.


\begin{corollary} [The case of $2$ traitors] \label{crl:The Case of $2$ Traitors}
	The $3$ player EPRQDBA protocol achieves detectable Byzantine agreement even when there are $2$ traitors among the $3$ generals.
\end{corollary}
\begin{proof}
	The fundamental assumption (\textbf{DBA}$_{ 2 }$) of any detectable Byzantine agreement protocol, stipulates that all loyal generals either follow the same order or abort the protocol (recall Definition \ref{def:Detectable Byzantine Agreement}). When there is only one loyal general, no matter whether he is the commanding general or a lieutenant general, whatever he decides is perfectly fine, as there is no other loyal party that he must agree with.
\end{proof}

The next theorem is an obvious consequence of Propositions \ref{prp:Loyal Alice} and \ref{prp:Traitor Alice}, and Corollary \ref{crl:The Case of $2$ Traitors}.

\begin{theorem} [$3$ Player Detectable Byzantine Agreement] \label{thr:$3$ Player Detectable Byzantine Agreement}
	The $3$ player EPRQDBA protocol achieves detectable Byzantine agreement in any eventuality.
\end{theorem}

\section{The $n$ player EPRQDBA protocol} \label{sec:The $n$ Player EPRQDBA Protocol}

In this section we present the general form of the EPRQDBA protocol that can handle any number $n > 3$ generals who are spatially distributed. To be precise, we assume that Alice is the unique commanding general, who tries to coordinate her $n - 1$ lieutenant generals LT$_0$, \dots, LT$_{n - 2}$, LT being an abbreviation for ``lieutenant general.'' As in the $3$ player case, the $n$ player form of the EPRQDBA protocol can be conceptually divided into $3$ distinct phases, which we describe below.

\subsection{The $n$ player entanglement distribution phase} \label{subsec:The $n$ Player Entanglement Distribution Phase}

The first phase is the entanglement distribution phase that involves the generation and distribution of qubits and entangled EPR pairs. As we have mentioned in assumption ($\mathbf{A}_{ 1 }$), we assume the existence of a trusted quantum source that undertakes this role. It is a relatively easy task, in view of the capabilities of modern quantum apparatus. Hence, the quantum source will have no difficulty in producing

\begin{itemize}
	\item	$( n - 2 ) ( n - 1 ) m$ qubits in the $\ket{+}$ state by simply applying the Hadamard transform on $\ket{ 0 }$, and
	\item	$( n - 1 ) m$ EPR pairs in the $\ket{ \Psi^{ + } }$ state, which can be easily generated by modern quantum computers.
\end{itemize}

The parameter $m$ plays an important role and should be a sufficiently large positive integer. The forthcoming mathematical analysis of the protocol contains additional discussion regarding $m$. The source distributes the produced qubits according to the next scheme.

\begin{enumerate} [ left = 0.70 cm, labelsep = 0.50 cm ]
	\renewcommand\labelenumi{(\textbf{$n$D}$_\theenumi$)}
	\item	The $( n - 1 ) m$ EPR pairs are numbered from $0$ to $( n - 1 ) m - 1$.
	\item	The source sends a sequence
	$q_{ 0 }^{ A }$, $q_{ 1 }^{ A }$, \dots, $q_{ ( n - 1 ) m - 2 }^{ A }$, $q_{ ( n - 1 ) m - 1 }^{ A }$ of $( n - 1 ) m$ qubits to Alice
	and a sequence
	$q_{ 0 }^{ i }$, $q_{ 1 }^{ i }$, \dots, $q_{ ( n - 1 ) m - 2 }^{ i }$, $q_{ ( n - 1 ) m - 1 }^{ i }$ of $( n - 1 ) m$ qubits to lieutenant general LT$_i$, $0 \leq i \leq n - 2$.
	\item	Alice's sequence is constructed by inserting in position $k$ the first qubit of the EPR pair $k$, $0 \leq k \leq ( n - 1 ) m - 1$.
	\item	LT$_i$'s sequence, $0 \leq i \leq n - 2$, is constructed as follows. Position $k$, $0 \leq k \leq ( n - 1 ) m - 1$, contains the second qubit of the EPR pair $k$ if and only if $k \equiv i \bmod ( n - 1 )$; otherwise position $k$ contains a qubit in the $\ket{+}$ state.
\end{enumerate}

The end result is that the quantum registers of Alice and her $n - 1$ lieutenant generals LT$_0$, \dots, LT$_{n - 2}$, are populated as shown in Figure \ref{fig:Alice's and her LTs' Quantum Registers}. Qubits of the same EPR pair are shown in the same color. Green is used to indicate the EPR pairs shared between Alice and LT$_0$, which occupy those positions $k$ that satisfy the relation $k \equiv 0 \bmod ( n - 1 )$ in their quantum registers. Analogously, red is used for the EPR pairs shared between Alice and LT$_1$ occupying positions $k$ that satisfy the relation $k \equiv 1 \bmod ( n - 1 )$ in their quantum registers, and blue indicates EPR pairs shared between Alice and LT$_{ n - 2 }$ occupying positions $k$ that satisfy the relation $k \equiv n - 2 \bmod ( n - 1 )$. The gray qubits represent those in the $\ket{+}$ state that occupy the remaining positions in the quantum registers of the lieutenant generals.

\begin{tcolorbox}
	[
		grow to left by = 1.25 cm,
		grow to right by = 1.25 cm,
		colback = MagentaVeryLight!12,			
		enhanced jigsaw,						
		sharp corners,
		toprule = 1.0 pt,
		bottomrule = 1.0 pt,
		leftrule = 0.1 pt,
		rightrule = 0.1 pt,
		sharp corners,
		center title,
		fonttitle = \bfseries
	]
	\begin{figure}[H]
		\centering
		\begin{tikzpicture} [ scale = 0.30 ]
			\node
			[
			alice,
			scale = 1.75,
			anchor = center,
			label = { [ label distance = 0.00 cm ] west: \textbf{Alice} }
			]
			(Alice) { };
			\matrix
			[
			matrix of nodes, nodes in empty cells,
			column sep = 0.000 mm, right = 1.00 of Alice,
			nodes = { circle, minimum size = 10 mm, semithick, font = \footnotesize },
			]
			{
				\node [ shade, outer color = WordBlueVeryLight!50, inner color = white ] { }; &
				\node [ shade, outer color = orange!50, inner color = white ] { \dots }; &
				\node [ shade, outer color = RedPurple!50, inner color = white ] { }; &
				\node [ shade, outer color = GreenLighter2!50, inner color = white ] { }; &
				\node { \dots }; &
				\node [ shade, outer color = WordBlueVeryLight!50, inner color = white ] { }; &
				\node [ shade, outer color = orange!50, inner color = white ] { \dots }; &
				\node [ shade, outer color = RedPurple!50, inner color = white ] { }; &
				\node [ shade, outer color = GreenLighter2!50, inner color = white ] { }; &
				\node [ shade, outer color = WordBlueVeryLight!50, inner color = white ] { }; &
				\node [ shade, outer color = orange!50, inner color = white ] { \dots }; &
				\node [ shade, outer color = RedPurple!50, inner color = white ] { }; &
				\node [ shade, outer color = GreenLighter2!50, inner color = white ] { };
				\\
			}; %
			\node
			[
			businessman,
			scale = 1.75,
			anchor = center,
			below = 1.00 cm of Alice,
			label = { [ label distance = 0.00 cm ] west: \textbf{LT$_{ n - 2 }$} }
			]
			(Bob) { };
			\matrix
			[
			column sep = 0.000 mm, right = 1.00 of Bob,
			nodes = { circle, minimum size = 10 mm, semithick, font = \footnotesize },
			]
			{
				\node [ shade, outer color = WordBlueVeryLight!50, inner color = white ] { }; &
				\node [ shade, outer color = WordIceBlue, inner color = white ] { \dots }; &
				\node [ shade, outer color = WordIceBlue, inner color = white ] { }; &
				\node [ shade, outer color = WordIceBlue, inner color = white ] { }; &
				\node { \dots }; &
				\node [ shade, outer color = WordBlueVeryLight!50, inner color = white ] { }; &
				\node [ shade, outer color = WordIceBlue, inner color = white ] { \dots }; &
				\node [ shade, outer color = WordIceBlue, inner color = white ] { }; &
				\node [ shade, outer color = WordIceBlue, inner color = white ] { }; &
				\node [ shade, outer color = WordBlueVeryLight!50, inner color = white ] { }; &
				\node [ shade, outer color = WordIceBlue, inner color = white ] { \dots }; &
				\node [ shade, outer color = WordIceBlue, inner color = white ] { }; &
				\node [ shade, outer color = WordIceBlue, inner color = white ] { };
				\\
			};
			\node
			[
			anchor = center,
			below = 1.00 cm of Bob
			]
			(Dots) { \Large \dots };
			\matrix
			[
			column sep = 0.000 mm, right = 1.00 of Dots,
			nodes = { circle, minimum size = 10 mm, semithick, font = \footnotesize },
			]
			(BobMatrix)
			{
				\node { }; &
				\node { }; &
				\node { }; &
				\node { }; &
				\node { \Large \dots }; &
				\node { }; &
				\node { }; &
				\node { }; &
				\node { \Large \dots }; &
				\node { }; &
				\node { }; &
				\node { }; &
				\node { };
				\\
			};
			\node
			[
			criminal,
			scale = 1.75,
			anchor = center,
			below = 1.00 cm of Dots,
			label = { [ label distance = 0.00 cm ] west: \textbf{LT$_{ 1 }$} }
			]
			(Charlie) { };
			\matrix
			[
			column sep = 0.000 mm, right = 1.00 of Charlie,
			nodes = { circle, minimum size = 10 mm, semithick, font = \footnotesize },
			]
			{
				\node [ shade, outer color = WordIceBlue, inner color = white ] { }; &
				\node [ shade, outer color = WordIceBlue, inner color = white ] { \dots }; &
				\node [ shade, outer color = RedPurple!50, inner color = white ] { }; &
				\node [ shade, outer color = WordIceBlue, inner color = white ] { }; &
				\node { \dots }; &
				\node [ shade, outer color = WordIceBlue, inner color = white ] { }; &
				\node [ shade, outer color = WordIceBlue, inner color = white ] { \dots }; &
				\node [ shade, outer color = RedPurple!50, inner color = white ] { }; &
				\node [ shade, outer color = WordIceBlue, inner color = white ] { }; &
				\node [ shade, outer color = WordIceBlue, inner color = white ] { }; &
				\node [ shade, outer color = WordIceBlue, inner color = white ] { \dots }; &
				\node [ shade, outer color = RedPurple!50, inner color = white ] { }; &
				\node [ shade, outer color = WordIceBlue, inner color = white ] { };
				\\
			};
			\node
			[
			maninblack,
			scale = 1.75,
			anchor = center,
			below = 1.00 cm of Charlie,
			label = { [ label distance = 0.00 cm ] west: \textbf{LT$_{ 0 }$} }
			]
			(Dave) { };
			\matrix
			[
			column sep = 0.000 mm, right = 1.00 of Dave,
			nodes = { circle, minimum size = 10 mm, semithick, font = \footnotesize },
			]
			{
				\node [ shade, outer color = WordIceBlue, inner color = white ] { }; &
				\node [ shade, outer color = WordIceBlue, inner color = white ] { \dots }; &
				\node [ shade, outer color = WordIceBlue, inner color = white ] { }; &
				\node [ shade, outer color = GreenLighter2!50, inner color = white ] { }; &
				\node { \dots }; &
				\node [ shade, outer color = WordIceBlue, inner color = white ] { }; &
				\node [ shade, outer color = WordIceBlue, inner color = white ] { \dots }; &
				\node [ shade, outer color = WordIceBlue, inner color = white ] { }; &
				\node [ shade, outer color = GreenLighter2!50, inner color = white ] { }; &
				\node [ shade, outer color = WordIceBlue, inner color = white ] { }; &
				\node [ shade, outer color = WordIceBlue, inner color = white ] { \dots }; &
				\node [ shade, outer color = WordIceBlue, inner color = white ] { }; &
				\node [ shade, outer color = GreenLighter2!50, inner color = white ] { };
				\\
			};
		\end{tikzpicture}
		\caption{Qubits of the same EPR pair are shown in the same color. Green is used to indicate the EPR pairs shared between Alice and LT$_0$, which occupy those positions $k$ that satisfy the relation $k \equiv 0 \bmod ( n - 1 )$ in their quantum registers. Analogously, red is used for the EPR pairs shared between Alice and LT$_1$ occupying positions $k$ that satisfy the relation $k \equiv 1 \bmod ( n - 1 )$ in their quantum registers, and blue indicates EPR pairs shared between Alice and LT$_{ n - 2 }$ occupying positions $k$ that satisfy the relation $k \equiv n - 2 \bmod ( n - 1 )$. The gray qubits represent those in the $\ket{+}$ state that occupy the remaining positions in the quantum registers of the lieutenant generals.}
		\label{fig:Alice's and her LTs' Quantum Registers}
	\end{figure}
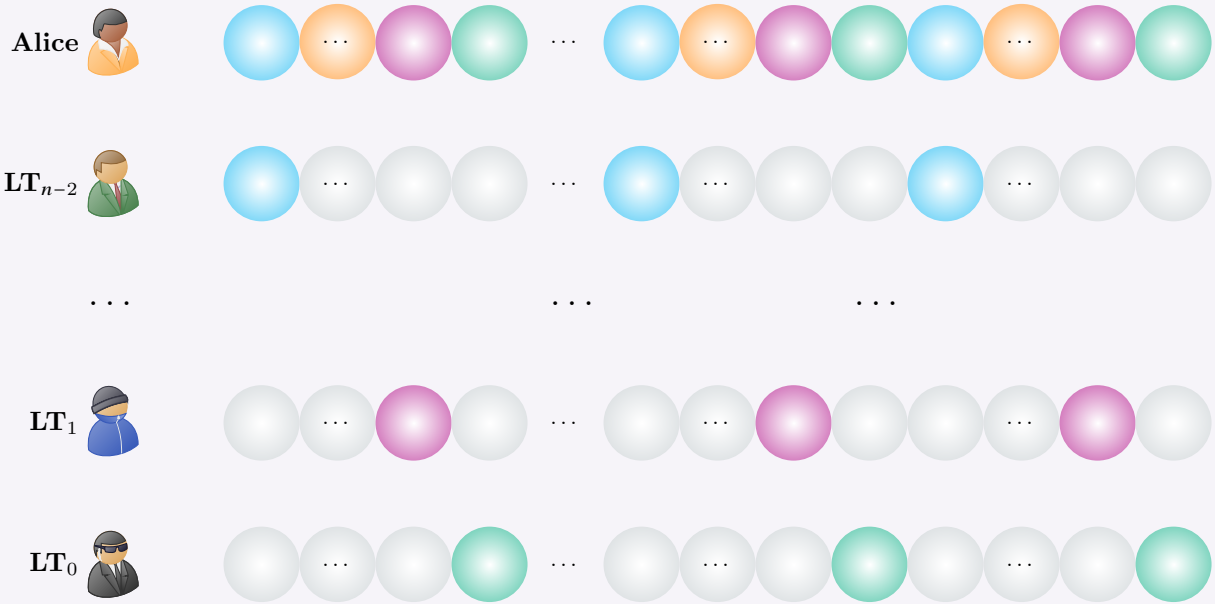
\end{tcolorbox}

\subsection{The $n$ player entanglement verification phase} \label{subsec:The $n$ Player Entanglement Verification Phase}

This phase is extremely important because the entire protocol is based on entanglement. If entanglement is not guaranteed, then agreement cannot not be guaranteed either. The~verification process can lead to two dramatically different outcomes. If entanglement verification is successfully established, then the EPRQDBA protocol is certain to achieve agreement. Failure of verification implies absence of the necessary entanglement. This~may be attributed either to noisy quantum channels or insidious sabotage by an active adversary. Whatever the true reason is, the only viable solution is to abort the current execution of the protocol, and initiate the whole procedure again from scratch, after taking some corrective~measures.

One should not underestimate the significance of this phase. Many previous papers have comprehensively and thoroughly studied this phase. The~EPRQDBA protocol subscribes to the sophisticated methods that have already been described in the literature, such as~\cite{Fitzi2001,Cabello2003,Neigovzen2008,Feng2019,Qu2023}. The interested reader may consult these works and the references therein for further analysis.

\subsection{The $n$ player agreement phase} \label{subsec:The $n$ Player Agreement Phase}

The EPRQDBA protocol achieves detectable agreement during its third and last phase, aptly named \textbf{agreement phase}. The agreement phase begins when Alice and her $n - 1$ lieutenant generals LT$_0$, \dots, LT$_{n - 2}$ measure their quantum registers. By recalling the distribution scheme of Subsection \ref{subsec:The $n$ Player Entanglement Distribution Phase}, we see that certain important correlations among Alice's register and the registers of her lieutenant generals have been established.

\begin{definition} \label{def:$n$ Player Measured Contents of the Registers} \
	Let the bit vectors $\mathbf { a }$ and $\mathbf { l }^{ i }$ denote the contents of Alice and LT$_i$'s, $0 \leq i \leq n - 2$, registers after the measurement, and assume that their explicit form is as given below:
	\begin{align}
		\mathbf { a }
		&=
		\underbrace { a_{ ( n - 1 ) m - 1 } \dots a_{ ( n - 1 ) m - ( n - 1 ) } }_{ ( n - 1 ) \text{-tuple} } \
		\dots \
		\underbrace { a_{ n - 2 } \dots a_{ 0 } }_{ ( n - 1 ) \text{-tuple} }
		\ , \label{eq:Alice's Bit Vector a - 3}
		\\
		\mathbf { l }^{ i }
		&=
		\underbrace { l_{ ( n - 1 ) m - 1 }^{ i } \dots l_{ ( n - 1 ) m - ( n - 1 ) }^{ i } }_{ ( n - 1 ) \text{-tuple} } \
		\dots \
		\underbrace { l_{ n - 2 }^{ i } \dots l_{ 0 }^{ i } }_{ ( n - 1 ) \text{-tuple} }
		\ . \label{eq:Lieutenant General's i Bit Vector - 1}
	\end{align}
	The $k^{ th }$ $( n - 1 )$-tuple of $\mathbf { a }$, $0 \leq k \leq m - 1$, is the $( n - 1 )$-tuple of bits $a_{ ( n - 1 ) k + ( n - 2 ) }$ \dots $a_{ ( n - 1 ) k }$, and is designated by $\mathbf { a }_{ k }$. Similarly, the $k^{ th }$ $( n - 1 )$-tuple of $\mathbf { l }^{ i }$ is designated by $\mathbf { l }_{ k }^{ i }$. Hence, $\mathbf { a }$ and $\mathbf { l }^{ i }$ can be written succinctly as:
	\begin{align}
		\mathbf { a }
		&=
		\mathbf { a }_{ m - 1 } \ \mathbf { a }_{ m - 2 } \ \dots \ \mathbf { a }_{ 1 } \ \mathbf { a }_{ 0 }
		\ , \label{eq:Alice's Bit Vector a - 4}
		\\
		\mathbf { l }^{ i }
		&=
		\mathbf { l }_{ m - 1 }^{ i } \ \mathbf { l }_{ m - 2 }^{ i } \ \dots \ \mathbf { l }_{ 1 }^{ i } \ \mathbf { l }_{ 0 }^{ i }
		\ . \label{eq:Lieutenant General's i Bit Vector - 1}
	\end{align}
	We also use the notion of the \emph{uncertain} $( n - 1 )$-tuple, denoted by $\mathbf { u } = \sqcup \dots \sqcup$, where $\sqcup$ is a new symbol, different from $0$ and $1$. In contrast, $( n - 1 )$-tuples consisting of bits $0$ and/or $1$ are called \emph{definite} tuples.
\end{definition}

According to the distribution scheme, each $( n - 1 )$-tuple in Alice's bit vector $\mathbf { a }$ shares one $\ket{ \Psi^{ + } } = \frac { \ket{ 0 }_{ A } \ket{ 1 }_{ i } + \ket{ 1 }_{ A } \ket{ 0 }_{ i } } { \sqrt{ 2 } }$ pair with every lieutenant general LT$_i$, $0 \leq i \leq n - 2$. This is captured by the next Lemma \ref{lem:$n$ Player Tuple Differentiation Property}. Its proof is obvious and is omitted.

\begin{lemma} [Tuple Differentiation Property] \label{lem:$n$ Player Tuple Differentiation Property}
	The next property, termed \emph{tuple differentiation} property characterizes the corresponding bits and tuples of the bit vectors $\mathbf { a }$ and $\mathbf { l }^{ i }$, $0 \leq i \leq n - 2$.
	\begin{align}
		{ l }_{ k }^{ i } &= \overline{ a_{ k } }
		\ , \text{ iff } k \equiv i \bmod ( n - 1 ) \ ,
		\ k = 0, \dots, ( n - 1 ) m - 1 \ ,
		\ 0 \leq i \leq n - 2
		\ .
		\label{eq:Alice - Lieutenant General's i Bit Relation}
	\end{align}
	Consequently, each $\mathbf { a }_{ k }$ tuple, where $0 \leq k \leq m - 1$, is different from the corresponding tuple of every $\mathbf { l }^{ i }$, $0 \leq i \leq n - 2$.
	\begin{align}
		\mathbf { a }_{ k } &\neq \mathbf { l }_{ k }^{ i }
		\ , \text{ for every } k, \ 0 \leq k \leq m - 1 \ .
		\label{eq:Alice - Lieutenant General's i Tuple Differentiation Property}
	\end{align}
\end{lemma}

\begin{definition} [$n$ Player Command Vectors] \label{def:$n$ Player Command Vectors} \
	Alice sends to every lieutenant general LT$_i$, $0 \leq i \leq n - 2$, either the command $0$ or the command $1$. Besides her order, as a ``proof,'' she also sends an appropriate \emph{command vector}. The idea is that the command vector for every lieutenant general is always different, even when the command is the same. For the order $0$, the command vector sent to every lieutenant general LT$_i$, $0 \leq i \leq n - 2$, is $\vmathbb{ 0 }_{ i }$, whereas for the order $1$ the corresponding command vector is $\mathds{ 1 }_{ i }$. The explicit form of the command vectors is the following.
	\begin{align}
		\vmathbb{ 0 }_{ i }
		&=
		\mathbf { v }_{ m - 1 } \ \dots \ \mathbf { v }_{ 0 }
		\ ,
		\ \text{ where } \
		\mathbf { v }_{ k }
		=
		\left\{
		\begin{matrix*}[l]
			\mathbf { a }_{ k } & \text{ if } a_{ ( n - 1 ) k + i } = 0 \\
			\mathbf { u } = \sqcup \dots \sqcup & \text{ if } a_{ ( n - 1 ) k + i } \neq 0
		\end{matrix*}
		\right.
		\ , \ 0 \leq k \leq m - 1 \ .
		\label{eq:Lieutenant General's i Command Vector 0}
		\\
		\mathds{ 1 }_{ i }
		&=
		\mathbf { v }_{ m - 1 } \ \dots \ \mathbf { v }_{ 0 }
		\ ,
		\ \text{ where } \
		\mathbf { v }_{ k }
		=
		\left\{
		\begin{matrix*}[l]
			\mathbf { a }_{ k } & \text{ if } a_{ ( n - 1 ) k + i } = 1 \\
			\mathbf { u } = \sqcup \dots \sqcup & \text{ if } a_{ ( n - 1 ) k + i } \neq 1
		\end{matrix*}
		\right.
		\ , \ 0 \leq k \leq m - 1 \ .
		\label{eq:Lieutenant General's i Command Vector 1}
	\end{align}
	A command vector besides tuples comprised entirely of $0$ and $1$ bits, also contains an approximately equal number of tuples consisting of $\sqcup$ characters. When we want refer to a command vector, but without providing further details, we will designate it by
	\begin{align}
		\mathbf { v }
		=
		\mathbf { v }_{ m - 1 } \ \dots \ \mathbf { v }_{ 0 }
		=
		\underbrace { v_{ ( n - 1 ) m - 1 } \dots v_{ ( n - 1 ) m - ( n - 1 ) } }_{ ( n - 1 ) \text{-tuple} } \
		\dots \
		\underbrace { v_{ n - 2 } \dots v_{ 0 } }_{ ( n - 1 ) \text{-tuple} }
		\ . \label{eq:Generic Command Vector for Lieutenant General i}
	\end{align}
\end{definition}

Given a command vector or a bit vector, we gather in a set the positions of the $( n - 1 )$-tuples that contain a specific combination of bits.

\begin{definition} [Tuple Designation] \label{def:$n$ Player Tuple Designation} \
	Given a command vector $\mathbf { v }$ or a bit vector $\mathbf { l }^{ i }$, $0 \leq i \leq n - 2$, we define the set $\mathbb{ T }_{ j \rightarrow y }^{ i \rightarrow x } ( \mathbf { v } )$ and $\mathbb{ T }_{ j \rightarrow y }^{ i \rightarrow x } ( \mathbf { l }^{ i } )$, respectively, of the positions of the $( n - 1 )$-tuples that contain bit $x \in \{ 0, 1 \}$ in the $i^{ th }$ place and bit $y \in \{ 0, 1 \}$ in the $j^{ th }$ place, $0 \leq i \neq j \leq n - 2$.
	
	Likewise, the set of the positions of their $( n - 1 )$-tuples that contain bit $x \in \{ 0, 1 \}$ in the $i^{ th }$ place is denoted by $\mathbb{ T }^{ i \rightarrow x } ( \mathbf { v } )$ and $\mathbb{ T }^{ i \rightarrow x } ( \mathbf { l }^{ i } )$, respectively.
\end{definition}

In view of Definition \ref{def:$n$ Player Command Vectors}, it is trivial to see that the assertions of the next Lemma \ref{lem:$n$ Player Loyal Command Vector Properties} are valid, where the notation $| S |$ is employed to designate the cardinality, i.e., the number of elements, of an arbitrary set $S$.

\begin{lemma} [$n$ Player Loyal Command Vector Properties] \label{lem:$n$ Player Loyal Command Vector Properties} \
	If Alice is loyal, her command vectors $\vmathbb{ 0 }_{ i }$ and $\mathds{ 1 }_{ i }$, $0 \leq i \leq n - 2$, satisfy the following properties.
\end{lemma}

\begin{minipage}{0.485 \textwidth}
	\begin{itemize}
		\item	$| \mathbb{ T }^{ i \rightarrow 0 } ( \vmathbb{ 0 }_{ i } ) | \approx \frac { m } { 2 }$.
		\item	$| \mathbb{ T }_{ j \rightarrow 0 }^{ i \rightarrow 0 } ( \vmathbb{ 0 }_{ i } ) | \approx \frac { m } { 4 }$.
		\item	$| \mathbb{ T }_{ j \rightarrow 1 }^{ i \rightarrow 0 } ( \vmathbb{ 0 }_{ i } ) | \approx \frac { m } { 4 }$.
	\end{itemize}
\end{minipage}
\begin{minipage}{0.485 \textwidth}
	\begin{itemize}
		\item	$| \mathbb{ T }^{ i \rightarrow 1 } ( \mathds{ 1 }_{ i } ) | \approx \frac { m } { 2 }$.
		\item	$| \mathbb{ T }_{ j \rightarrow 0 }^{ i \rightarrow 1 } ( \mathds{ 1 }_{ i } ) | \approx \frac { m } { 4 }$.
		\item	$| \mathbb{ T }_{ j \rightarrow 1 }^{ i \rightarrow 1 } ( \mathds{ 1 }_{ i } ) | \approx \frac { m } { 4 }$.
	\end{itemize}
\end{minipage}

As expected the protocol evolves in rounds. The actions of the $n$ generals in each round are explained below. In the exposition of Protocol \ref{alg:The $n$ Player EPRQDBA Protocol}, we have made a few plausible assumptions regarding the behavior of traitors, which we state below.

\begin{enumerate} [ left = 0.50 cm, labelsep = 0.50 cm ]
	\renewcommand\labelenumi{(\textbf{T}$_\theenumi$)}
		\item	If Alice is loyal, then the traitors do not possess any inside information. This means that for a disloyal lieutenant general it will be probabilistically impossible to forge a consistent command vector. On this fact actually relies the soundness of the protocol when Alice is loyal.
		\item	The disloyal lieutenant general manifest their traitorous behavior at the end of ($\mathbf{Round}_{ 2 }$), either by falsely claiming that their decision is to abort, or by sending contradictory orders accompanied by consistent command vectors to different lieutenant generals. However, their behavior is consistent in the sense that they always send identical information to the same lieutenant general, i.e., at the end of ($\mathbf{Round}_{ 3 }$) they transmit exactly the same decision they sent at the end of ($\mathbf{Round}_{ 2 }$).
\end{enumerate}

\begin{tcolorbox}
	[
		enhanced,
		breakable,
		grow to left by = 0.00 cm,
		grow to right by = 0.00 cm,
		colback = WordIceBlue,			
		enhanced jigsaw,				
		sharp corners,
		toprule = 0.01 pt,
		bottomrule = 0.01 pt,
		leftrule = 0.1 pt,
		rightrule = 0.1 pt,
		sharp corners,
		center title,
		fonttitle = \bfseries
	]
	\begin{algorithm}[H]
		\SetAlgorithmName{Protocol}{ }{ }
		\setcounter{algocf}{1}
		\caption { \textsc { The $n$ Player EPRQDBA Protocol } }
		\label{alg:The $n$ Player EPRQDBA Protocol}
	\end{algorithm}
	\begin{enumerate} [ left = 0.90 cm, labelsep = 0.50 cm ]
		\renewcommand\labelenumi{(\textbf{Round}$_\theenumi$)}
		\item	\textbf{Send} $\triangleright$ Alice sends to all of her $n - 1$ lieutenant generals LT$_0$, \dots, LT$_{n - 2}$ her order $c_{ i }$ (which is either $0$ or $1$) and the appropriate command vector $\mathbf { v }_{ i }$ as proof.
		\item	\textbf{Receive} $\triangleright$ Every LT$_i$, $0 \leq i \leq n - 2$, receives Alice's order $c_{ i }$ and command vector $\mathbf { v }_{ i }$.

		Subsequently, each LT$_i$, calls the \textsc{CheckAlice} Algorithm \ref{alg:The $n$ Player CheckAlice Algorithm} to check $\mathbf { v }_{ i }$ against his bit vector for inconsistencies.
		\begin{itemize}
			\item[$\star$]	If no inconsistencies are found, LT$_i$'s initial decision is $d_{ i }^{ ( 2 ) } = c_{ i }$.
			\item[$\star$]	If LT$_i$ detects inconsistencies, his initial decision is $d_{ i }^{ ( 2 ) } = \bot$.
		\end{itemize}
		\textbf{Send} $\triangleright$ Upon completion of the verification procedure, LT$_i$ sends to all other lieutenant generals LT$_j$, $0 \leq i \neq j \leq n - 2$, his initial decision $d_{ i }^{ ( 2 ) }$ together with $\mathbf { v }_{ i }$ as proof.
		\item	\textbf{Receive} $\triangleright$ Every LT$_i$, $0 \leq i \leq n - 2$, receives from every other LT$_j$, $0 \leq i \neq j \leq n - 2$, LT$_j$'s initial decision $d_{ j }^{ ( 2 ) }$ together with $\mathbf { v }_{ j }$ as proof.

		Afterwards, each LT$_i$, compares his $d_{ i }^{ ( 2 ) }$ to all other $d_{ j }^{ ( 2 ) }$, $0 \leq i \neq j \leq n - 2$.
		\begin{itemize}
			\item[$\star$]	\colorbox{WordLightGreen}{\textbf{Rule}$_{ 3, 1 }$}: If all $d_{ j }^{ ( 2 ) }$ coincide with $d_{ i }^{ ( 2 ) }$, then LT$_i$'s intermediary decision is $d_{ i }^{ ( 3 ) } = d_{ i }^{ ( 2 ) }$.
			\item[$\star$]	If there is \emph{at least one} $d_{ j }^{ ( 2 ) }$ such that $d_{ i }^{ ( 2 ) } \neq d_{ j }^{ ( 2 ) }$, then the following cases are considered.
			\begin{itemize}
				\item[$\square$]	\colorbox{WordLightGreen}{\textbf{Rule}$_{ 3, 2 }$}: If $d_{ i }^{ ( 2 ) } = 0 \ ( 1 )$ and \emph{all} different decisions are $\bot$, then LT$_i$'s intermediary decision is $d_{ i }^{ ( 3 ) } = d_{ i }^{ ( 2 ) }$.
				\item[$\square$]	If $d_{ i }^{ ( 2 ) } = 0 \ (1)$ and there exist $d_{ j }^{ ( 2 ) } = 1 \ ( 0 )$, then LT$_i$ uses the \textsc{CheckLTwCV} Algorithm \ref{alg:The $n$ Player CheckLTwCV Algorithm} to check the corresponding $\mathbf { v }_{ j }$.
				\begin{itemize}
					\item[$\diamond$]	\colorbox{WordLightGreen}{\textbf{Rule}$_{ 3, 3 }$}: If there is \emph{at least one} consistent $\mathbf { v }_{ j }$, then LT$_i$'s intermediary decision is $d_{ i }^{ ( 3 ) } = \bot$.
					\item[$\diamond$]	\colorbox{WordLightGreen}{\textbf{Rule}$_{ 3, 4 }$}: Otherwise, LT$_i$'s intermediary decision is $d_{ i }^{ ( 3 ) } = d_{ i }^{ ( 2 ) }$.
				\end{itemize}
				\item[$\square$]	If $d_{ i }^{ ( 2 ) } = \bot$ and there exist different $d_{ j }^{ ( 2 ) }$, then LT$_i$ uses the \textsc{CheckLTwBV} Algorithm \ref{alg:The $n$ Player CheckLTwBV Algorithm} to check the corresponding $\mathbf { v }_{ j }$.
				\begin{itemize}
					\item[$\diamond$]	\colorbox{WordLightGreen}{\textbf{Rule}$_{ 3, 5 }$}: If all the consistent $\mathbf { v }_{ j }$ correspond to the \emph{same} $d_{ j }^{ ( 2 ) }$, then LT$_i$'s intermediary decision is $d_{ i }^{ ( 3 ) } = d_{ j }^{ ( 2 ) }$.
					\item[$\diamond$]	\colorbox{WordLightGreen}{\textbf{Rule}$_{ 3, 6 }$}: Otherwise, LT$_i$'s intermediary decision is $d_{ i }^{ ( 3 ) } = d_{ i }^{ ( 2 ) }$.
				\end{itemize}
			\end{itemize}
		\end{itemize}
		\textbf{Send} $\triangleright$ Upon completion of the above comparison, LT$_i$ sends to all other lieutenant generals LT$_j$, $0 \leq i \neq j \leq n - 2$, his intermediary decision $d_{ i }^{ ( 3 ) }$ with proof those $\mathbf { v }_{ k }, \mathbf { v }_{ l }$ that enabled him to arrive at his decision.

		\textbf{Note.} The number of transmitted $\mathbf { v }_{ j }$(s) can be $0$ (\textbf{Rule}$_{ 3, 1 }$, \textbf{Rule}$_{ 3, 2 }$, \textbf{Rule}$_{ 3, 4 }$), $1$ (\textbf{Rule}$_{ 3, 3 }$, \textbf{Rule}$_{ 3, 5 }$), or even $2$, in case there are two opposite orders, but with consistent command vectors (\textbf{Rule}$_{ 3, 6 }$). It should be clear that this time LT$_i$ transmits his own intermediary decision, together with command vectors sent to him from other lieutenant generals.
		\item	\textbf{Receive} $\triangleright$ Every LT$_i$, $0 \leq i \leq n - 2$, receives from every other LT$_j$, $0 \leq i \neq j \leq n - 2$, LT$_j$'s intermediary decision $d_{ j }^{ ( 3 ) }$ together with a list of $0, 1$ or $2$ command vectors $\mathbf { v }_{ k }, \mathbf { v }_{ l }$ as proof.

		Subsequently, each LT$_i$, compares his $d_{ i }^{ ( 3 ) }$ to all other $d_{ j }^{ ( 3 ) }$, $0 \leq i \neq j \leq n - 2$.
		\begin{itemize}
			\item[$\star$]	\colorbox{WordLightGreen}{\textbf{Rule}$_{ 4, 1 }$}: If $d_{ i }^{ ( 3 ) } = \bot$ because LT$_i$ has been sent two consistent but contradictory command vectors, then LT$_i$'s final decision is $d_{ i }^{ ( 4 ) } = d_{ i }^{ ( 3 ) }$.
			\item[$\star$]	\colorbox{WordLightGreen}{\textbf{Rule}$_{ 4, 2 }$}: If $d_{ i }^{ ( 3 ) } = d_{ j }^{ ( 3 ) }$, for every $j, 0 \leq i \neq j \leq n - 2$, then $d_{ i }^{ ( 4 ) } = d_{ i }^{ ( 3 ) }$.
			\item[$\star$]	If If $d_{ i }^{ ( 3 ) } = 0 \ ( 1 )$ and there is \emph{at least one} $d_{ j }^{ ( 3 ) }$ such that $d_{ i }^{ ( 3 ) } \neq d_{ j }^{ ( 3 ) }$, then the following cases are considered.
			\begin{itemize}
				\item[$\square$]	If there exist $d_{ j }^{ ( 3 ) } = \bot \neq d_{ j }^{ ( 2 ) }$ (meaning that LT$_j$ has revised his intermediary decision from $0$ or $1$ to $\bot$), then, for every such $j$, LT$_i$ uses the \textsc{CheckLTwCV} Algorithm \ref{alg:The $n$ Player CheckLTwCV Algorithm} to check the command vectors supplied as proof.
				\begin{itemize}
					\item[$\diamond$]	\colorbox{WordLightGreen}{\textbf{Rule}$_{ 4, 3 }$}: If, for at least one $j$, they are consistent, then $d_{ i }^{ ( 4 ) } = \bot$.
					\item[$\diamond$]	\colorbox{WordLightGreen}{\textbf{Rule}$_{ 4, 4 }$}: Otherwise, LT$_i$'s final decision is $d_{ i }^{ ( 4 ) } = d_{ i }^{ ( 3 ) }$.
				\end{itemize}
				\item[$\square$]	If LT$_i$'s decision is $0$ ($1$) and \emph{at least one} different decision is $1$ ($0$), then LT$_i$ uses the \textsc{CheckLTwCV} Algorithm \ref{alg:The $n$ Player CheckLTwCV Algorithm} to check the command vectors of all lieutenant generals whose intermediary decision is $1$ ($0$) for inconsistencies.
				\begin{itemize}
					\item[$\diamond$]	\colorbox{WordLightGreen}{\textbf{Rule}$_{ 4, 5 }$}: If there exists \emph{at least one} command vector with no inconsistencies, then LT$_i$ aborts and terminates the protocol.
					\item[$\diamond$]	\colorbox{WordLightGreen}{\textbf{Rule}$_{ 4, 6 }$}: Otherwise, LT$_i$ clings to his intermediary decision.
				\end{itemize}
			\end{itemize}
		\end{itemize}
	\end{enumerate}
\end{tcolorbox}

The following algorithms apply Lemmata \ref{lem:$n$ Player Tuple Differentiation Property} and \ref{lem:$n$ Player Loyal Command Vector Properties} to determine whether a command vector is consistent, in which case they return TRUE. If they determine inconsistencies, they terminate and return FALSE. The names of Algorithms \ref{alg:The $n$ Player CheckLTwCV Algorithm} and \ref{alg:The $n$ Player CheckLTwBV Algorithm} are mnemonic abbreviations for ``Check LT with Command Vector'' and ``Check LT with Bit Vector.'' To clear up things, we explain their difference.

Without loss of generality, let us suppose that LT$_i$ received the order $0$ and a consistent command vector $\mathbf { v }_{ A }$ from Alice, while LT$_j$ claims that he received the order $1$. According to the EPRQDBA protocol, LT$_j$ must convince LT$_i$ by sending the command vector $\mathbf { v }$ he claims he received from Alice. LT$_i$ checks $\mathbf { v }$ for inconsistencies against his own $\mathbf { v }_{ A }$ by invoking the Algorithm \ref{alg:The $n$ Player CheckLTwCV Algorithm}. Now, suppose that LT$_i$ received an inconsistent command vector $\mathbf { v }_{ A }$ from Alice and his initial decision is to abort, while LT$_j$ claims that he received the order $1$. LT$_j$ must convince LT$_i$ by sending the command vector $\mathbf { v }$ he claims he received from Alice. In this situation, LT$_i$ does not have a consistent command vector $\mathbf { v }_{ A }$ that he could use, so LT$_i$ must check $\mathbf { v }$ for inconsistencies against his own bit vector $\mathbf { b }$ by invoking the Algorithm \ref{alg:The $n$ Player CheckLTwBV Algorithm}.

In Algorithms \ref{alg:The $n$ Player CheckAlice Algorithm}, \ref{alg:The $n$ Player CheckLTwCV Algorithm}, and \ref{alg:The $n$ Player CheckLTwBV Algorithm}, we employ the following notation.
\begin{itemize}
	\item	$i$, $0 \leq i \leq n - 2$, is the index of the lieutenant general who does the consistency checking.
	\item	$j$, $0 \leq i \neq j \leq n - 2$, is the index of the lieutenant general whose command vector is being checked for consistency.
	\item	$c$, where $c = 0$ or $c = 1$, is the command being checked for consistency.
	\item	$\mathbf { v }_{ A } = v_{ ( n - 1 ) m - 1 } \dots v_{ ( n - 1 ) m - ( n - 1 ) } \ \dots \ v_{ n - 2 } \dots v_{ 0 }$, is the command vector sent by Alice.
	\item	$\mathbf { v }$ is the command vector sent by the lieutenant general LT$_j$.
	\item	$\mathbf { l } = l_{ ( n - 1 ) m - 1 } \dots l_{ ( n - 1 ) m - ( n - 1 ) } \ \dots \ l_{ n - 2 } \dots l_{ 0 }$ is the bit vector of the lieutenant general LT$_i$ who does the consistency checking.
	\item	$\bigtriangleup$ is the symmetric difference of two sets, i.e., $ S \bigtriangleup S' = \left( S \setminus S' \right) \bigcup \left( S' \setminus S \right)$, for given sets $S$ and $S'$.
\end{itemize}

\begin{algorithm}[H]
	\setcounter{algocf}{3}
	\SetKw{Break}{break}
	\caption{ \textsc{ CheckAlice } ( i, c, $\mathbf { v }_{ A }$, $\mathbf { l }$ ) }
	\label{alg:The $n$ Player CheckAlice Algorithm}
	\If
	{
		$! \ ( \ | \ \mathbb{ T }^{ i \rightarrow c } ( \mathbf { v }_{ A } ) \ | \approx \frac { m } { 2 } \ )$
	}
	{
		\Return FALSE
	}
	\For {$k = 0$ \KwTo $m - 1$}
	{
		\If { $( v_{ ( n - 1 ) k + i } == l_{ ( n - 1 ) k + i } )$ }
		{
			\Return FALSE
		}
	}
	\Return TRUE
\end{algorithm}

\begin{algorithm}[H]
	\SetKw{Break}{break}
	\caption{ \textsc{CheckLTwCV} (i, j, c, $\mathbf { v }$, $\mathbf { v }_{ A }$) }
	\label{alg:The $n$ Player CheckLTwCV Algorithm}
	\If
	{
		$! \ ( \ | \ \mathbb{ T }_{ j \rightarrow c }^{ i \rightarrow c } ( \mathbf { v } ) \ | \approx \frac { m } { 4 } \ )$
	}
	{
		\Return FALSE
	}
	\If
	{
		$! \ ( \ | \ \mathbb{ T }_{ j \rightarrow c }^{ i \rightarrow \overline{ c } } ( \mathbf { v } ) \ | \approx \frac { m } { 4 } \ )$
	}
	{
		\Return FALSE
	}
	\If
	{
		$! \ ( \ | \
		\mathbb{ T }_{ j \rightarrow c }^{ i \rightarrow \overline{ c } } ( \mathbf { v }_{ A } )
		\bigtriangleup
		\mathbb{ T }_{ j \rightarrow c }^{ i \rightarrow \overline{ c } } ( \mathbf { v } )
		\ | \approx 0 \ )$
	}
	{
		\Return FALSE
	}
	\Return TRUE
\end{algorithm}

\begin{algorithm}[H]
	\SetKw{Break}{break}
	\caption{ \textsc{CheckLTwBV} (i, j, c, $\mathbf { v }$, $\mathbf { l }$) }
	\label{alg:The $n$ Player CheckLTwBV Algorithm}
	\If
	{
		$! \ ( \ | \ \mathbb{ T }_{ j \rightarrow c }^{ i \rightarrow c } ( \mathbf { v } ) \ | \approx \frac { m } { 4 } \ )$
	}
	{
		\Return FALSE
	}
	\If
	{
		$! \ ( \ | \ \mathbb{ T }_{ j \rightarrow c }^{ i \rightarrow \overline{ c } } ( \mathbf { v } ) \ | \approx \frac { m } { 4 } \ )$
	}
	{
		\Return FALSE
	}
	\For {$k = 0$ \KwTo $m - 1$}
	{
		\If { $( v_{ ( n - 1 ) k + i } == l_{ ( n - 1 ) k + i } )$ }
		{
			\Return FALSE
		}
	}
	\Return TRUE
\end{algorithm}

We now proceed to formally prove that the EPRQDBA protocol enables all the loyal generals to reach agreement.

\begin{proposition} [$n$ Player Loyal Alice] \label{prp:Loyal Alice General Case}
	If Alice is loyal, then the $n$ player EPRQDBA protocol will enable all loyal lieutenant generals to follow Alice's order.
\end{proposition}

\begin{proposition} [$n$ Player Traitor Alice] \label{prp:Traitor Alice General Case}
	If Alice is a traitor, then the $n$ player EPRQDBA protocol will enable all loyal lieutenant generals to reach agreement, in the sense of all following the same order or all aborting.
\end{proposition}

The next Theorem \ref{thr:$n$ Player Detectable Byzantine Agreement} is an immediate consequence of Propositions \ref{prp:Loyal Alice General Case} and \ref{prp:Traitor Alice General Case}.

\begin{theorem} [$n$ Player Detectable Byzantine Agreement] \label{thr:$n$ Player Detectable Byzantine Agreement}
	The $n$ player EPRQDBA protocol achieves detectable Byzantine agreement.
\end{theorem}

In closing we may ask what will happen if there are $n - 1$ traitors and only one loyal general. It might be that the $n - 1$ traitors are the $n - 1$ lieutenant generals, or Alice and $n - 2$ of her lieutenant generals. Again, this case turns to be very easy. The crucial assumption (\textbf{DBA}$_{ 2 }$) of any detectable Byzantine agreement protocol, stipulates that all loyal generals either follow the same order or abort the protocol (recall Definition \ref{def:Detectable Byzantine Agreement}). When there is only one loyal general, maybe Alice herself, or one of the lieutenant generals, any decision taken would be fine, as there is no other loyal party that to agree with.

\section{Discussion and conclusions} \label{sec:Discussion and Conclusions}

This paper introduced the quantum protocol EPRQDBA that achieves Detectable Byzantine Agreement for an arbitrary number $n$ of players.

The EPRQDBA protocol has a number of advantages to offer.

\begin{itemize}
	\item	It is completely general because it can handle any number $n$ of players.
	\item	It takes only a constant number of rounds, namely $3$, to complete, no matter how many players there are.
	\item	It requires only EPR (specifically $\ket{ \Psi^{ + } }$) pairs, irrespective of the number $n$ of players. This is probably the most important quality of the protocol. Without a doubt, EPR pairs are the easiest to produce among all maximally entangled states. In comparison, even $\ket{ GHZ_{ n } }$ states, would not be able to scale well as $n$ increases, not to mention other more difficult to create entangled states. 
\end{itemize}

All the above qualities advocate for its scalability and its potential applicability in practice. To drive further this point, we present the numerical characteristics of the EPRQDBA protocol in the following Table \ref{tbl:Numerical Characteristics of the EPRQDBA Protocol}.

Parameter $m$ is independent of the number of players $n$, which is another plus for the scalability of the protocol. Obviously, a proper value for must be selected in order to ensure that the probability that a traitor can forge a command vector is negligible. The reason why $m$ does not scale with $n$ is that the protocol always uses EPR pairs, that is, bipartite entanglement. The rules of the protocol dictate that every consistency check takes place between two command vectors, or between a command vector and a bit vector. Therefore, in every case the comparison involves just $2$ strings of bits, irrespective of the number of generals. Moreover, this comparison, even in the most general case, involves just $2$ bits, say $i$ and $j$, in every $n - 1$ tuple. Thus, the situation from a probabilistic point of view is always identical with the $3$ player case. In a final analysis, the probability that a traitor deceives a loyal general, is the probability of picking the one correct configuration out of many. The total number of configurations is equal to the number of ways to place $\approx \frac { m } { 4 }$ identical objects (either $0$ or $1$, depending on the order) into $\approx \frac { m } { 2 }$ distinguishable boxes (the uncertain tuples). The probability that a traitor places \emph{all} the $\approx \frac { m } { 4 }$ bits correctly in the $\approx \frac { m } { 2 }$ uncertain tuples is
\begin{align}
	P( \text{ traitor cheats } )
	\approx
	\frac { 1 }
	{ \binom { \ m / 2 \ } { \ m / 4 \ } }
	\ , \label{eq:Traitor Cheats}
\end{align}
which tends to zero as $m$ increases.

\begin{table}[H]
	\renewcommand{\arraystretch}{1.50}
	\caption{This table summarizes the numerical characteristics of the EPRQDBA protocol. The abbreviations ``bpm'' and ``tb'' stand for ``bits per message'' and ``total bits,'' respectively.}
	\label{tbl:Numerical Characteristics of the EPRQDBA Protocol}
	\centering
	\begin{tabular}
		{
			>{\centering\arraybackslash} m{3.00 cm} !{\vrule width 0.5 pt}
			>{\centering\arraybackslash} m{3.00 cm} 
		}
		\Xhline{4\arrayrulewidth}
		\multicolumn{2}{c}{Qubits}
		\\
		\Xhline{\arrayrulewidth}
		\# of $\ket{ \Psi^{ + } }$ pairs
		&
		\# of $\ket{ + }$ qubits
		\\
		\Xhline{3\arrayrulewidth}
		$( n - 1 ) m$
		&
		$( n - 2 ) ( n - 1 ) m$
		\\
		\Xhline{\arrayrulewidth}
	\end{tabular}
	\begin{tabular}
		{
			>{\centering\arraybackslash} m{1.00 cm} !{\vrule width 0.5 pt}
			>{\centering\arraybackslash} m{3.00 cm} !{\vrule width 0.5 pt}
			>{\centering\arraybackslash} m{3.50 cm} !{\vrule width 0.5 pt}
			>{\centering\arraybackslash} m{4.25 cm}
		}
		\Xhline{4\arrayrulewidth}
		\multicolumn{4}{c}{Bits \& Messages}
		\\
		\Xhline{\arrayrulewidth}
		Round
		&
		\# of messages
		&
		\# bits per message
		&
		Total \# of bits
		\\
		\Xhline{3\arrayrulewidth}
		$1$
		&
		$n - 1$
		&
		$( n - 1 ) m$
		&
		$( n - 1 )^{ 2 } m$
		\\
		\Xhline{\arrayrulewidth}
		$2$
		&
		$( n - 2 ) ( n - 1 )$
		&
		$( n - 1 ) m$
		&
		$( n - 2 ) ( n - 1 )^{ 2 } m$
		\\
		\Xhline{\arrayrulewidth}
		$3$
		&
		$( n - 2 ) ( n - 1 )$
		&
		$0 \leq \text{\# bpm} \leq 2 ( n - 1 ) m$
		&
		$0 \leq \text{\# tb} \leq 2 ( n - 2 ) ( n - 1 )^{ 2 } m$
		\\
		\Xhline{\arrayrulewidth}
		$4$
		&
		$0$
		&
		$0$
		&
		$0$
		\\
		\Xhline{1\arrayrulewidth}
	\end{tabular}
	\begin{tabular}
		{
			>{\centering\arraybackslash} m{1.00 cm} !{\vrule width 0.5 pt}
			>{\centering\arraybackslash} m{3.00 cm} !{\vrule width 0.5 pt}
			>{\centering\arraybackslash} m{3.50 cm} !{\vrule width 0.5 pt}
			>{\centering\arraybackslash} m{4.25 cm}
		}
		\Xhline{4\arrayrulewidth}
		\multicolumn{4}{c}{$m$ \& Probability}
		\\
		\Xhline{\arrayrulewidth}
		$m$
		&
		$\frac { m } { 4 }$
		&
		$\frac { m } { 2 }$
		&
		$P( \text{ traitor cheats } )$
		\\
		\Xhline{3\arrayrulewidth}
		$4$
		&
		\xinteval { 4 // 4 }
		&
		\xinteval { 4 // 2 }
		&
		\xintdeffloatfunc	Combinations ( m ) := ( m / 2 )! // ( ( m / 4 )! ( m / 4 )! );
		\xintdeffloatfunc	Probability ( m ) := 1 / Combinations ( m );
		\xintfloateval{ Probability ( 4 ) }
		\\
		\Xhline{\arrayrulewidth}
		$8$
		&
		\xinteval { 8 // 4 }
		&
		\xinteval { 8 // 2 }
		&
		\xintDigits*:=3;
		\xintdeffloatfunc	Combinations ( m ) := ( m / 2 )! // ( ( m / 4 )! ( m / 4 )! );
		\xintdeffloatfunc	Probability ( m ) := 1 / Combinations ( m );
		\xintfloateval{ Probability ( 8 ) }
		\\
		\Xhline{\arrayrulewidth}
		$16$
		&
		\xinteval { 16 // 4 }
		&
		\xinteval { 16 // 2 }
		&
		\xintDigits*:=3;
		\xintdeffloatfunc	Combinations ( m ) := ( m / 2 )! // ( ( m / 4 )! ( m / 4 )! );
		\xintdeffloatfunc	Probability ( m ) := 1 / Combinations ( m );
		$\xintfloateval{ Probability ( 16 ) }$
		\\
		\Xhline{\arrayrulewidth}
		$32$
		&
		\xinteval { 32 // 4 }
		&
		\xinteval { 32 // 2 }
		&
		\xintDigits*:=4;
		\xintdeffloatfunc	Combinations ( m ) := ( m / 2 )! // ( ( m / 4 )! ( m / 4 )! );
		\xintdeffloatfunc	Probability ( m ) := 1 / Combinations ( m );
		\def\xintfloatexprPrintOne[#1]#2{\xintTeXFromSci{\xintFloat[#1]{#2}}}
		$\xintfloateval{ Probability ( 32 ) }$
		\\
		\Xhline{\arrayrulewidth}
		$64$
		&
		\xinteval { 64 // 4 }
		&
		\xinteval { 64 // 2 }
		&
		\xintDigits*:=4;
		\xintdeffloatfunc	Combinations ( m ) := ( m / 2 )! // ( ( m / 4 )! ( m / 4 )! );
		\xintdeffloatfunc	Probability ( m ) := 1 / Combinations ( m );
		\def\xintfloatexprPrintOne[#1]#2{\xintTeXFromSci{\xintFloat[#1]{#2}}}
		$\xintfloateval{ Probability ( 64 ) }$
		\\
		\Xhline{4\arrayrulewidth}
	\end{tabular}
	\renewcommand{\arraystretch}{1.0}
\end{table}
\appendix \label{Appendix A}

\appendixpage

\section{Proofs of the Main Results} \label{appsec:Appendix - Main Results Proofs}

In this Appendix we formally prove the main results of this paper.

\begin{proposition} [Loyal Alice] \label{prp:Loyal Alice Appendix}
	If Alice is loyal, the $3$ player EPRQDBA protocol will enable the loyal lieutenant general to agree with Alice. Specifically, if Alice and Bob are loyal and Charlie is a traitor, Bob will follow Alice's order. Symmetrically, if Alice and Charlie are loyal and Bob is a traitor, Charlie will follow Alice's order.
\end{proposition}

\begin{proof} [Proof of Proposition \ref{prp:Loyal Alice Appendix}] \

	We assume that Alice is loyal and, without loss of generality, that Alice conveys the command $0$ to both Bob and Charlie, along with the appropriate command vectors as proof. These command vectors satisfy the properties asserted in Lemma \ref{lem:Loyal Command Vector Properties}.

	Let us also assume that Bob is loyal. Consequently, Bob will use the Algorithm \ref{alg:The $3$ Player CheckAlice Algorithm} to verify the validity of the received command vector $\vmathbb{ 0 }_{ B }$.

	According to the rules of the $3$ player EPRQDBA Protocol \ref{alg:The $3$ Player EPRQDBA Protocol}, the only way for Bob to change his initial decision and abort is if Charlie sends to Bob a valid command vector $\mathds{ 1 }_{ C }$. In such a case Charlie would be a traitor that has in reality received the command vector $\vmathbb{ 0 }_{ C }$ during ($\mathbf{Round}_{ 1 }$). Can Charlie construct a convincing $\mathds{ 1 }_{ C }$?

	Bob, having received the command vector $\vmathbb{ 0 }_{ B }$, knows the positions of all the pairs in $\mathbf { a }$ of the form $a_{ 2 k + 1 } a_{2 k}$, such that $a_{ 2 k + 1 } = 0$ and $a_{2 k} = 1$, $0 \leq k \leq m - 1$.
	\begin{align}
		\mathbb{ P }_{ 0, 1 } ( \vmathbb{ 0 }_{ B } )
		=
		\{ p_{ 1 }, p_{ 2 }, \dots, p_{ t } \}
		\ , \text{ where } t \approx \frac { m } { 4 }
		\ . \tag{ P\ref{prp:Loyal Alice Appendix}.i }
	\end{align}
	Having $\vmathbb{ 0 }_{ C }$, Charlie knows that the least significant bit of the $\approx \frac { m } { 2 }$ uncertain pairs in $\vmathbb{ 0 }_{ C }$ is $1$ and the most significant bit is $0$ or $1$, with equal probability $0.5$. However, Charlie cannot know with certainty if the most significant bit of a specific uncertain pair is $0$ or $1$. Therefore, when constructing $\mathds{ 1 }_{ C }$, Charlie can make two detectable mistakes.
	\begin{enumerate}
		\renewcommand\labelenumi{(\textbf{M}$_\theenumi$)}
		\item	Place a $0$ in a wrong pair that is not actually contained in $\mathbb{ P }_{ 0, 1 } ( \vmathbb{ 0 }_{ B } )$.
		\item	Place a $1$ in a wrong pair that appears in $\mathbb{ P }_{ 0, 1 } ( \vmathbb{ 0 }_{ B } )$.
	\end{enumerate}
	The situation from a probabilistic point of view resembles the probability of picking the one correct configuration out of many. The total number of configurations is equal to the number of ways to place $\approx \frac { m } { 4 }$ identical objects ($0$) into $\approx \frac { m } { 2 }$ distinguishable boxes. The probability that Charlie places \emph{all} the $\approx \frac { m } { 4 }$ bits $0$ correctly in the $\approx \frac { m } { 2 }$ uncertain pairs is
	\begin{align}
		P( \text{ Charlie places all $0$ correctly } )
		\approx
		\frac { 1 }
		{ \binom { \ m / 2 \ } { \ m / 4 \ } }
		\ , \tag{ P\ref{prp:Loyal Alice Appendix}.ii }
	\end{align}
	which tends to zero as $m$ increases. Put it another way, the cardinality of the symmetric difference $\mathbb{ P }_{ 0, 1 } ( \vmathbb{ 0 }_{ B } ) \bigtriangleup \mathbb{ P }_{ 0, 1 } ( \mathds{ 1 }_{ C } )$ will be significantly greater than $0$ from a statistical point of view.

	Hence, when Bob uses the Algorithm \ref{alg:The $3$ Player CheckWCV Algorithm} during ($\mathbf{Round}_{ 3 }$) to verify the validity of the command vector $\mathds{ 1 }_{ C }$ sent by Charlie, he will detect inconsistencies and stick to his preliminary decision. Ergo, Bob follows Alice's order.

	The case where Alice conveys the command $1$ is virtually identical.

	Finally, the situation where Alice and Charlie are loyal and Bob is a traitor, is completely symmetrical.
\end{proof}

\begin{proposition} [Traitor Alice] \label{prp:Traitor Alice Appendix}
	If Bob and Charlie are loyal and Alice is a traitor, the $3$ player EPRQDBA protocol will enable Bob and Charlie to reach agreement, in the sense that they will either follow the same order or abort.
\end{proposition}

\begin{proof}  [Proof of Proposition \ref{prp:Traitor Alice Appendix}] \

	Let us now consider the case where the commanding general Alice is a traitor, but both Bob and Charlie are loyal. We distinguish the following cases.

	\begin{itemize}
		\item	Alice sends a consistent command vector to Bob (Charlie) and an inconsistent command vector to Charlie (Bob). In this case, Bob (Charlie) will stick to his initial decision, but Charlie (Bob) will use the \textsc{CheckWBV} Algorithm \ref{alg:The $3$ Player CheckWBV Algorithm} to verify Bob's (Charlie's) command vector. Then, according to \textbf{Rule}$_{ 3, 5 }$ of the $3$ player EPRQDBA Protocol \ref{alg:The $3$ Player EPRQDBA Protocol}, Charlie (Bob) will change his decision to that of Bob (Charlie).
		\item	Alice sends to Bob and Charlie different orders with consistent command vectors. Let us assume, without loss of generality, that Alice sends the order $0$ together with a consistent command vector $\vmathbb{ 0 }_{ B }$ to Bob and the order $1$ together with a consistent command vector $\mathds{ 1 }_{ C }$ to Charlie.

		Bob knows the positions of almost all $00$ and $01$ pairs in $\mathbf { a }$. If Alice had forged even a single $10$ or $11$ pair, claiming to be either $00$ or $01$, then Bob, when using the \textsc{CheckAlice} Algorithm \ref{alg:The $3$ Player CheckAlice Algorithm}, he would have immediately detected the inconsistency. Therefore, Bob knows the sets
		\begin{align}
			\mathbb{ P }_{ 0, 0 } ( \vmathbb{ 0 }_{ B } )
			&=
			\{ p_{ 1 }, p_{ 2 }, \dots, p_{ t_{ 1 } } \} \ , \text{ where } t_{ 1 } \approx \frac { m } { 4 }
			\ , \ \text{ and } \tag{ P\ref{prp:Traitor Alice Appendix}.i }
			\\
			\mathbb{ P }_{ 0, 1 } ( \vmathbb{ 0 }_{ B } )
			&=
			\{ p'_{ 1 }, p'_{ 2 }, \dots, p'_{ t_{ 2 } } \} \ , \text{ where } t_{ 2 } \approx \frac { m } { 4 }
			\ , \tag{ P\ref{prp:Traitor Alice Appendix}.ii }
		\end{align}
		of the positions of the $00$ and $01$ pairs in $\mathbf { a }$.

		Symmetrically, Charlie knows the positions of almost all $01$ and $11$ pairs in $\mathbf { a }$:
		\begin{align}
			\mathbb{ P }_{ 0, 1 } ( \mathds{ 1 }_{ C } )
			&=
			\{ q_{ 1 }, q_{ 2 }, \dots, q_{ t_{ 3 } } \} \ , \text{ where } t_{ 3 } \approx \frac { m } { 4 }
			\ , \ \text{ and } \tag{ P\ref{prp:Traitor Alice Appendix}.iii }
			\\
			\mathbb{ P }_{ 1, 1 } ( \mathds{ 1 }_{ C } )
			&=
			\{ q'_{ 1 }, q'_{ 2 }, \dots, q'_{ t_{ 4 } } \} \ , \text{ where } t_{ 4 } \approx \frac { m } { 4 }
			\ . \tag{ P\ref{prp:Traitor Alice Appendix}.iv }
		\end{align}
		According to the rules of the protocol, Bob will sent to Charlie the command vector $\vmathbb{ 0 }_{ B }$, and, simultaneously, Charlie will send to Bob the command vector $\mathds{ 1 }_{ C }$. Both will use the \textsc{CheckWCV} Algorithm \ref{alg:The $3$ Player CheckWCV Algorithm} to verify the consistency of the others' command vector. Then, according to \textbf{Rule}$_{ 3, 3 }$, both Bob and Charlie will abort.
		\item	Alice sends to both Bob and Charlie inconsistent command vectors. Then, according to \textbf{Rule}$_{ 3, 1 }$ of the protocol, both Bob and Charlie agree to abort.
	\end{itemize}
	Ergo, both Bob and Charlie will either follow the same order or abort.
\end{proof}

\begin{proposition} [$n$ Player Loyal Alice] \label{prp:Loyal Alice General Case Appendix}
	If Alice is loyal, then the $n$ player EPRQDBA protocol will enable all loyal lieutenant generals to follow Alice's order.
\end{proposition}

\begin{proof} [Proof of Proposition \ref{prp:Loyal Alice General Case Appendix}] \

	We assume that Alice is loyal and, without loss of generality, we suppose that Alice conveys the command $0$ to all her lieutenant generals, along with the appropriate command vectors as proof. These command vectors satisfy the properties asserted in Lemma \ref{lem:$n$ Player Loyal Command Vector Properties}.

	Let us consider an arbitrary loyal lieutenant general LT$_i$, $0 \leq i \leq n - 2$. LT$_i$ will use the Algorithm \ref{alg:The $n$ Player CheckAlice Algorithm} to verify the validity of the received command vector $\vmathbb{ 0 }_{ i }$. Hence, LT$_i$ will have accepted order $0$ at the end of ($\mathbf{Round}_{ 2 }$).

	According to the rules of the $n$ player EPRQDBA Protocol \ref{alg:The $n$ Player EPRQDBA Protocol}, the only way for LT$_i$ to change his initial decision and abort is if another LT$_j$, $0 \leq i \neq j \leq n - 2$, sends to LT$_i$ his decision $1$ accompanied by a valid command vector $\mathds{ 1 }_{ j }$. Such an LT$_j$ would, of course, be a traitor that has in reality received the command vector $\vmathbb{ 0 }_{ j }$ during ($\mathbf{Round}_{ 1 }$). Can such an LT$_j$ construct a convincing $\mathds{ 1 }_{ j }$?

	LT$_i$, having received the command vector $\vmathbb{ 0 }_{ i }$, knows the positions of all the $( n - 1 )$-tuples in $\mathbf { a }$ that contain $0$ in the $i^{th}$ position and $1$ in the $j^{th}$ position.
	\begin{align}
		\mathbb{ T }_{ j \rightarrow 1 }^{ i \rightarrow 0 } ( \vmathbb{ 0 }_{ i } )
		=
		\{ p_{ 1 }, p_{ 2 }, \dots, p_{ t } \}
		\ , \text{ where } t \approx \frac { m } { 4 }
		\ . \tag{ P\ref{prp:Loyal Alice General Case Appendix}.i }
	\end{align}
	Knowing $\vmathbb{ 0 }_{ j }$, LT$_j$ knows the $( n - 1 )$-tuples of $\mathbf { a }$ that contain $0$ in their $j^{th}$ place, and the uncertain $( n - 1 )$-tuples. Both are $\approx \frac { m } { 2 }$ in number. LT$_j$ also knows that the $\approx \frac { m } { 2 }$ uncertain $( n - 1 )$-tuples contain $1$ in their $j^{th}$ position and $0$ or $1$, with equal probability $0.5$, in their $i^{th}$ position. However, LT$_j$ cannot know with certainty if a specific uncertain $( n - 1 )$-tuple contains $0$ or $1$ in the $i^{th}$ position. Therefore, when constructing $\mathds{ 1 }_{ j }$, LT$_j$ can make two detectable mistakes.
	\begin{enumerate}
		\renewcommand\labelenumi{(\textbf{M}$_\theenumi$)}
		\item	Place a $0$ in the $i^{th}$ place of a wrong $( n - 1 )$-tuple not contained in $\mathbb{ T }_{ j \rightarrow 1 }^{ i \rightarrow 0 } ( \vmathbb{ 0 }_{ i } )$.
		\item	Place a $1$ in the $i^{th}$ place of a wrong $( n - 1 )$-tuple that appears in $\mathbb{ T }_{ j \rightarrow 1 }^{ i \rightarrow 0 } ( \vmathbb{ 0 }_{ i } )$.
	\end{enumerate}
	The situation from a probabilistic point of view resembles the probability of picking the one correct configuration out of many. The total number of configurations is equal to the number of ways to place $\approx \frac { m } { 4 }$ identical objects into $\approx \frac { m } { 2 }$ distinguishable boxes. Hence, the probability that LT$_j$ places \emph{all} the $\approx \frac { m } { 4 }$ bits $0$ correctly in the $\approx \frac { m } { 2 }$ uncertain $( n - 1 )$-tuples is
	\begin{align}
		P( \text{ LT$_j$ places all $0$ correctly } )
		\approx
		\frac { 1 }
		{ \binom { \ m / 2 \ } { \ m / 4 \ } }
		\ , \tag{ P\ref{prp:Loyal Alice General Case Appendix}.ii }
	\end{align}
	which is practically zero for sufficiently large $m$. Thus, the cardinality of the symmetric difference
	\begin{align}
		\mathbb{ T }_{ j \rightarrow 1 }^{ i \rightarrow 0 } ( \vmathbb{ 0 }_{ i } )
		\bigtriangleup
		\mathbb{ T }_{ j \rightarrow 1 }^{ i \rightarrow 0 } ( \mathds{ 1 }_{ j } )
		=
		\left(
		\mathbb{ T }_{ j \rightarrow 1 }^{ i \rightarrow 0 } ( \vmathbb{ 0 }_{ i } )
		\setminus
		\mathbb{ T }_{ j \rightarrow 1 }^{ i \rightarrow 0 } ( \mathds{ 1 }_{ j } )
		\right)
		\bigcup
		\left(
		\mathbb{ T }_{ j \rightarrow 1 }^{ i \rightarrow 0 } ( \mathds{ 1 }_{ j } )
		\setminus
		\mathbb{ T }_{ j \rightarrow 1 }^{ i \rightarrow 0 } ( \vmathbb{ 0 }_{ i } )
		\right)
		\tag{ P\ref{prp:Loyal Alice General Case Appendix}.iii }
	\end{align}
	will be significantly greater than $0$ from a statistical point of view.

	Hence, when LT$_i$ uses the Algorithm \ref{alg:The $n$ Player CheckLTwCV Algorithm} during ($\mathbf{Round}_{ 3 }$) to verify the validity of the command vector $\mathds{ 1 }_{ j }$ sent by LT$_j$, he will detect inconsistencies and stick to his preliminary decision. Ergo, LT$_i$ and all loyal lieutenant generals will be included in $G_{ 0 }$ at the end of the protocol.

	The case where Alice conveys the command $1$ is identical.
\end{proof}

\begin{proposition} [$n$ Player Traitor Alice] \label{prp:Traitor Alice General Case Appendix}
	If Alice is a traitor, then the $n$ player EPRQDBA protocol will enable all loyal lieutenant generals to reach agreement, in the sense that they will either all follow the same order or all abort.
\end{proposition}

Proposition \ref{prp:Traitor Alice General Case Appendix} studies the situation where the commanding general Alice is a traitor. In this scenario, Alice has $3$ options regarding what to send to any one of her lieutenant generals. Specifically, she can send to lieutenant general LT$_i$, $0 \leq i \leq n - 2$, either:
\begin{enumerate}
	\item	the order $0$ together with a consistent command vector $\vmathbb{ 0 }_{ i }$, or
	\item	the order $1$ together with a consistent command vector $\mathds{ 1 }_{ i }$, or
	\item	an inconsistent command vector, so that LT$_i$'s initial decision would be to abort.
\end{enumerate}

We assume that disloyal lieutenant generals behave as specified in assumption ($\mathbf{T}_{ 2 }$). Moreover, if Alice is a traitor, then a disloyal lieutenant general may have $1$ or $2$ options when it comes to communicating with the other lieutenant generals. This analysis is based on the fact the it is probabilistically impossible for a lieutenant general to forge a consistent command vector, other than the one received directly from Alice. Accordingly, if he has received the consistent command vector $\vmathbb{ 0 }_{ i }$ ($\mathds{ 1 }_{ i }$), it is practically impossible to claim that he received the consistent command vector $\mathds{ 1 }_{ i }$ ($\vmathbb{ 0 }_{ i }$), or even $\vmathbb{ 0 }_{ k }$, for some $k \neq i$. Therefore, regarding the disloyal lieutenant general LT$_k$, $0 \leq k \leq n - 2$, the situation is as follows.
\begin{enumerate}
	\item	If he received from Alice an inconsistent command vector, then he is forced to behave honestly, since he can only say that his decision is to abort.
	\item	If he received from Alice the order $0$ ($1$) together with a consistent command vector $\vmathbb{ 0 }_{ k }$ ($\mathds{ 1 }_{ k }$), then he can either say that his decision is $0$ ($1$) and forward as proof the consistent command vector $\vmathbb{ 0 }_{ k }$ ($\mathds{ 1 }_{ k }$), or say that his decision is to abort.
\end{enumerate}

\begin{proof} [Proof of Proposition \ref{prp:Traitor Alice General Case Appendix}] \

	In terms of notation, to enhance the readability of this Proposition, we denote by $d_{ i }^{ ( 2 ) }$, $d_{ i }^{ ( 3 ) }$, and $d_{ i }^{ ( 4 ) }$ the initial, intermediary, and final decision of lieutenant general LT$_i$ at the end of ($\mathbf{Round}_{ 2 }$), ($\mathbf{Round}_{ 3 }$), and ($\mathbf{Round}_{ 4 }$). Analogously, we denote by $p_{ i }^{ ( 2 ) }$, $p_{ i }^{ ( 3 ) }$, and $p_{ i }^{ ( 4 ) }$ the initial, intermediary, and final proof of lieutenant general LT$_i$ at the end of ($\mathbf{Round}_{ 2 }$), ($\mathbf{Round}_{ 3 }$), and ($\mathbf{Round}_{ 4 }$). To capture the behavior of an arbitrary disloyal lieutenant general LT$_k$, we designate by $d_{ k \rightarrow i }^{ ( 2 ) }$, $d_{ k \rightarrow i }^{ ( 3 ) }$, and $d_{ k \rightarrow i }^{ ( 4 ) }$ the corresponding decisions he sends to LT$_i$at the end of ($\mathbf{Round}_{ 2 }$), ($\mathbf{Round}_{ 3 }$), and ($\mathbf{Round}_{ 4 }$), respectively, and by $p_{ k \rightarrow i }^{ ( 2 ) }$, $p_{ k \rightarrow i }^{ ( 3 ) }$, and $p_{ k \rightarrow i }^{ ( 4 ) }$ the corresponding proofs he sends to LT$_i$at the end of ($\mathbf{Round}_{ 2 }$), ($\mathbf{Round}_{ 3 }$), and ($\mathbf{Round}_{ 4 }$).

	\begin{itemize}
		\item	During ($\mathbf{Round}_{ 2 }$) all loyal lieutenant generals receive inconsistent command vectors from Alice.

		Accordingly, at the end of ($\mathbf{Round}_{ 2 }$), each loyal LT$_i$ communicates to every other lieutenant general his initial decision $d_{ i }^{ ( 2 ) } = \bot$ and his initial proof $p_{ i }^{ ( 2 ) } = \varepsilon$.

		Let us consider the possible actions of the disloyal lieutenant generals at the end of ($\mathbf{Round}_{ 2 }$).
		\begin{itemize}
			\item[$\square$]	At least one disloyal LT$_k$ sends to at least one loyal LT$_i$ his initial decision $d_{ k \rightarrow i }^{ ( 2 ) } = 0$ and his initial proof $p_{ k \rightarrow i }^{ ( 2 ) } = \vmathbb{ 0 }_{ k }$. Every other disloyal LT$_l$ tells each loyal LT$_j$ either
			\begin{itemize}
				\item[$\diamond$]	that his initial decision is $d_{ l \rightarrow j }^{ ( 2 ) } = \bot$ with initial proof $p_{ l \rightarrow j }^{ ( 2 ) } = \varepsilon$, or
				\item[$\diamond$]	that his initial decision is $d_{ l \rightarrow j }^{ ( 2 ) } = 0$, with initial proof $p_{ l \rightarrow j }^{ ( 2 ) } = \vmathbb{ 0 }_{ l }$.
			\end{itemize}
			\begin{tcolorbox}
				[
					colback = WordIceBlue,
					sharp corners,
					boxrule = 0.5 pt,
				]
				We clarify that each disloyal lieutenant general does not necessarily send the same decision to every loyal lieutenant general; he may send different decisions to different loyal lieutenant generals. The important point here is that at least one disloyal lieutenant general sends to at least one loyal lieutenant general the order $0$ along with a consistent command vector, and all the other disloyal lieutenant generals, either send $\bot$, or the same order $0$. The emphasis is on the ``same.'' There is nothing special about the specific order $0$; the case of the same order $1$ is entirely symmetrical.
			\end{tcolorbox}
			Then, during ($\mathbf{Round}_{ 3 }$), LT$_i$ will set his intermediary decision to $d_{ i }^{ ( 3 ) } = 0$, and, at the end of ($\mathbf{Round}_{ 3 }$), will communicate to every other lieutenant general his intermediary decision $d_{ i }^{ ( 3 ) } = 0$, together with his intermediary proof $p_{ i }^{ ( 3 ) } = \vmathbb{ 0 }_{ k }$.

			In this scenario, all loyal lieutenant generals to set their final decision to $0$ during ($\mathbf{Round}_{ 4 }$), thereby agreeing to follow order $0$ at the end of the protocol.
			\begin{tcolorbox}
				[
				colback = WordIceBlue,
				sharp corners,
				boxrule = 0.5 pt,
				]
				According to assumption ($\mathbf{T}_{ 2 }$), at the end of ($\mathbf{Round}_{ 3 }$), each disloyal LT$_k$ will transmit the same decision he sent at the end of ($\mathbf{Round}_{ 2 }$) to each other lieutenant general, i.e., $\bot$ or $0$. As a result, the decisions of the loyal lieutenant generals, will not be affected.
			\end{tcolorbox}
			\item[$\square$]	At least one disloyal LT$_k$ sends to at least one loyal LT$_i$ his initial decision $d_{ k \rightarrow i }^{ ( 2 ) } = 0$ and proof $p_{ k \rightarrow i }^{ ( 2 ) } = \vmathbb{ 0 }_{ k }$, and at least one disloyal LT$_l$ sends to at least one loyal LT$_j$ (conceivably $i = j$) his initial decision $d_{ l \rightarrow i }^{ ( 2 ) } = 1$ and proof $p_{ l \rightarrow j }^{ ( 2 ) } = \mathds{ 1 }_{ l }$. In this subcase, it does not matter what the remaining disloyal lieutenant generals do at the end of ($\mathbf{Round}_{ 2 }$).
			\begin{tcolorbox}
				[
					colback = WordIceBlue,
					sharp corners,
					boxrule = 0.5 pt,
				]
				The important point here is that the contradictory orders $0$ and $1$, both accompanied by consistent proofs, are sent to at least one loyal lieutenant general at the end of ($\mathbf{Round}_{ 2 }$).
			\end{tcolorbox}
			Then, during ($\mathbf{Round}_{ 3 }$), LT$_i$ will set his intermediary decision to $d_{ i }^{ ( 3 ) } = 0$, and, at the end of ($\mathbf{Round}_{ 3 }$), will communicate to every other lieutenant general his intermediary decision $d_{ i }^{ ( 3 ) } = 0$, together with his proof $p_{ i }^{ ( 3 ) } = \vmathbb{ 0 }_{ k }$. Symmetrically, during ($\mathbf{Round}_{ 3 }$), LT$_j$ will set his intermediary decision to $d_{ j }^{ ( 3 ) } = 1$, and, at the end of ($\mathbf{Round}_{ 3 }$), will communicate to every other lieutenant general his intermediary decision $d_{ j }^{ ( 3 ) } = 1$, together with his proof $p_{ j }^{ ( 3 ) } = \mathds{ 1 }_{ l }$. In the special case where $i = j$, during ($\mathbf{Round}_{ 3 }$), LT$_i$ will set his intermediary decision to $d_{ i }^{ ( 3 ) } = \bot$, and, at the end of ($\mathbf{Round}_{ 3 }$), will communicate to every other lieutenant general his intermediary decision $d_{ i }^{ ( 3 ) } = \bot$, together with his proof $p_{ i }^{ ( 3 ) } = \vmathbb{ 0 }_{ k }, \mathds{ 1 }_{ l }$. This will cause all other loyal lieutenant generals to set their final decision to $\bot$ during ($\mathbf{Round}_{ 4 }$), thereby achieving agreement by aborting at the end of the protocol.
			\item[$\square$]	Every disloyal LT$_k$ conveys to all loyal LT$_i$ that his initial decision is $d_{ l \rightarrow j }^{ ( 2 ) } = \bot$ with initial proof $p_{ l \rightarrow j }^{ ( 2 ) } = \varepsilon$.

			Then, during ($\mathbf{Round}_{ 3 }$), every loyal LT$_i$ will cling to his decision to abort, and, at the end of ($\mathbf{Round}_{ 3 }$), will communicate to every other lieutenant general his intermediary decision $d_{ i }^{ ( 3 ) } = \bot$, together with his proof $p_{ i }^{ ( 3 ) } = \varepsilon$. This will cause all other loyal lieutenant generals to set their final decision to $\bot$ during ($\mathbf{Round}_{ 4 }$), thereby achieving agreement by aborting at the end of the protocol. This will cause all loyal lieutenant generals to finalize their decision to abort during ($\mathbf{Round}_{ 4 }$), thereby achieving agreement by aborting at the end of the protocol.
			\begin{tcolorbox}
				[
					colback = WordIceBlue,
					sharp corners,
					boxrule = 0.5 pt,
				]
				According to assumption ($\mathbf{T}_{ 2 }$), at the end of ($\mathbf{Round}_{ 3 }$), each disloyal LT$_k$ will transmit the same decision $\bot$ he sent at the end of ($\mathbf{Round}_{ 2 }$) to each other lieutenant general. As a result, the decisions of the loyal lieutenant generals, will not be affected.
			\end{tcolorbox}
		\end{itemize}
		\item	During ($\mathbf{Round}_{ 2 }$) all loyal lieutenant generals receive and verify the same order from Alice. Without loss of generality, we may suppose that the order in question is $0$.

		Accordingly, at the end of ($\mathbf{Round}_{ 2 }$), each loyal lieutenant general LT$_i$ communicates to every other lieutenant general his initial decision $d_{ i }^{ ( 2 ) } = 0$ and his initial proof $p_{ i }^{ ( 2 ) } = \vmathbb{ 0 }_{ i }$.

		Let us consider the possible actions of the disloyal lieutenant generals at the end of ($\mathbf{Round}_{ 2 }$).
		\begin{itemize}
			\item[$\square$]	At least one disloyal LT$_k$ transmits to at least one loyal LT$_i$ his initial decision $d_{ k \rightarrow i }^{ ( 2 ) } = 1$ with a consistent initial proof $p_{ k \rightarrow i }^{ ( 2 ) } = \mathds{ 1 }_{ k }$.

			Then, during ($\mathbf{Round}_{ 3 }$), LT$_i$ will set his intermediary decision to $d_{ i }^{ ( 3 ) } = \bot$, and, at the end of ($\mathbf{Round}_{ 3 }$), will communicate to every other lieutenant general his intermediary decision $d_{ i }^{ ( 3 ) } = \bot$, together with his intermediary proof $p_{ i }^{ ( 3 ) } = \mathds{ 1 }_{ k }$. This will cause all other loyal lieutenant generals to set their final decision to $\bot$ during ($\mathbf{Round}_{ 4 }$), thereby achieving agreement by aborting at the end of the protocol.

			\item[$\square$]	Every disloyal LT$_k$ tells each loyal LT$_i$ either
			\begin{itemize}
				\item[$\diamond$]	that his initial decision is $d_{ k \rightarrow i }^{ ( 2 ) } = \bot$ with initial proof $p_{ k \rightarrow i }^{ ( 2 ) } = \varepsilon$, or
				\item[$\diamond$]	that his initial decision is $d_{ k \rightarrow i }^{ ( 2 ) } = 0$, with initial proof $p_{ k \rightarrow i }^{ ( 2 ) } = \vmathbb{ 0 }_{ k }$.
			\end{itemize}
			\begin{tcolorbox}
				[
					colback = WordIceBlue,
					sharp corners,
					boxrule = 0.5 pt,
				]
				We emphasize that each disloyal LT$_k$ does not necessarily send the same decision to every loyal lieutenant general; he may send different decisions to different loyal lieutenant generals.
			\end{tcolorbox}
			Then, during ($\mathbf{Round}_{ 3 }$), every LT$_i$ will cling to his initial decision and set his intermediary decision to $d_{ i }^{ ( 3 ) } = 0$, and, at the end of ($\mathbf{Round}_{ 3 }$), will communicate to every other lieutenant general his intermediary decision $d_{ i }^{ ( 3 ) } = 0$, together with his intermediary proof $p_{ i }^{ ( 3 ) } = \varepsilon$.

			In this scenario, all loyal lieutenant generals to set their final decision to $0$ during ($\mathbf{Round}_{ 4 }$), thereby agreeing to follow order $0$ at the end of the protocol.
			\begin{tcolorbox}
				[
					colback = WordIceBlue,
					sharp corners,
					boxrule = 0.5 pt,
				]
				According to assumption ($\mathbf{T}_{ 2 }$), at the end of ($\mathbf{Round}_{ 3 }$), each disloyal LT$_k$ will transmit the same decision he sent at the end of ($\mathbf{Round}_{ 2 }$) to each other lieutenant general, i.e., $\bot$ or $0$. As a result, the decisions of the loyal lieutenant generals, will not be affected.
			\end{tcolorbox}
		\end{itemize}
		\item	During ($\mathbf{Round}_{ 2 }$) at least one loyal lieutenant general LT$_i$ receives the order $0$ and the consistent command vector $\vmathbb{ 0 }_{ i }$ from Alice, and at least one loyal lieutenant general LT$_j$ receives the order $1$ and the consistent command vector $\mathds{ 1 }_{ j }$ from Alice. In this scenario, it does not matter what the remaining disloyal lieutenant generals do at the end of ($\mathbf{Round}_{ 2 }$).
		\begin{tcolorbox}
			[
				colback = WordIceBlue,
				sharp corners,
				boxrule = 0.5 pt,
			]
			The important point here is that the contradictory orders $0$ and $1$, both accompanied by consistent proofs, are sent to two loyal lieutenant generals at the end of ($\mathbf{Round}_{ 2 }$).
		\end{tcolorbox}
		Accordingly, at the end of ($\mathbf{Round}_{ 2 }$), LT$_i$ communicates to every other lieutenant general his initial decision $d_{ i }^{ ( 2 ) } = 0$ and his initial proof $p_{ i }^{ ( 2 ) } = \vmathbb{ 0 }_{ i }$, and LT$_j$ communicates to every other lieutenant general his initial decision $d_{ j }^{ ( 2 ) } = 1$ and his initial proof $p_{ j }^{ ( 2 ) } = \mathds{ 1 }_{ j }$.

		Then, during ($\mathbf{Round}_{ 3 }$), all loyal lieutenant generals set their intermediary decision to $\bot$, and, during ($\mathbf{Round}_{ 4 }$), all loyal lieutenant generals set their final decision to $\bot$, thereby achieving agreement by aborting at the end of the protocol.
		\item	During ($\mathbf{Round}_{ 2 }$) some loyal lieutenant generals receive and verify the same order from Alice, which, without loss of generality, we may suppose that is $0$, and the remaining loyal lieutenant generals receive inconsistent command vectors from Alice.

		Let us consider the possible actions of the disloyal lieutenant generals at the end of ($\mathbf{Round}_{ 2 }$).
		\begin{itemize}
			\item[$\square$]	Every other disloyal LT$_k$ tells each loyal LT$_i$ either
			\begin{itemize}
				\item[$\diamond$]	that his initial decision is $d_{ l \rightarrow k }^{ ( 2 ) } = \bot$ with initial proof $p_{ k \rightarrow i }^{ ( 2 ) } = \varepsilon$, or
				\item[$\diamond$]	that his initial decision is $d_{ l \rightarrow k }^{ ( 2 ) } = 0$, with initial proof $p_{ k \rightarrow i }^{ ( 2 ) } = \vmathbb{ 0 }_{ k }$.
			\end{itemize}
			Then, during ($\mathbf{Round}_{ 3 }$), every loyal LT$_i$ will set his intermediary decision to $d_{ i }^{ ( 3 ) } = 0$, and, during ($\mathbf{Round}_{ 4 }$), all loyal lieutenant generals set their final decision to $0$, thereby agreeing to follow order $0$ at the end of the protocol.
			\begin{tcolorbox}
				[
					colback = WordIceBlue,
					sharp corners,
					boxrule = 0.5 pt,
				]
				According to assumption ($\mathbf{T}_{ 2 }$), at the end of ($\mathbf{Round}_{ 3 }$), each disloyal LT$_k$ will transmit the same decision he sent at the end of ($\mathbf{Round}_{ 2 }$) to each other lieutenant general, i.e., $\bot$ or $0$. As a result, the decisions of the loyal lieutenant generals, will not be affected.
			\end{tcolorbox}
		\item[$\square$]	At least one disloyal LT$_k$ sends to at least one loyal LT$_i$ his initial decision $d_{ k \rightarrow i }^{ ( 2 ) } = 1$ and proof $p_{ k \rightarrow i }^{ ( 2 ) } = \mathds{ 1 }_{ k }$. In this subcase, it does not matter what the remaining disloyal lieutenant generals do at the end of ($\mathbf{Round}_{ 2 }$).
		\begin{tcolorbox}
			[
				colback = WordIceBlue,
				sharp corners,
				boxrule = 0.5 pt,
			]
			The important point here is that the contradictory order $1$, together with consistent proof, is sent to at least one loyal lieutenant general at the end of ($\mathbf{Round}_{ 2 }$).
		\end{tcolorbox}
		Then, during ($\mathbf{Round}_{ 3 }$), LT$_i$, having received two contradictory orders with consistent proofs, will set his intermediary decision to $d_{ i }^{ ( 3 ) } = \bot$, and, at the end of ($\mathbf{Round}_{ 3 }$), will communicate to every other lieutenant general his intermediary decision $d_{ i }^{ ( 3 ) } = \bot$, together with his proof $p_{ i }^{ ( 3 ) } = \vmathbb{ 0 }_{ j }, \mathds{ 1 }_{ k }$. This will cause all other loyal lieutenant generals to set their final decision to $\bot$ during ($\mathbf{Round}_{ 4 }$), thereby achieving agreement by aborting at the end of the protocol.
		\item[$\square$]	Every disloyal LT$_k$ conveys to all loyal LT$_i$ that his initial decision is $d_{ l \rightarrow j }^{ ( 2 ) } = \bot$ with initial proof $p_{ l \rightarrow j }^{ ( 2 ) } = \varepsilon$.

		Then, during ($\mathbf{Round}_{ 3 }$), every loyal LT$_i$ will set his intermediary decision to $0$, and, at the end of ($\mathbf{Round}_{ 3 }$), will communicate to every other lieutenant general his intermediary decision $d_{ i }^{ ( 3 ) } = 0$, together with his proof $p_{ i }^{ ( 3 ) } = \vmathbb{ 0 }_{ j }$. This will cause all other loyal lieutenant generals to set their final decision to $0$ during ($\mathbf{Round}_{ 4 }$), thereby achieving agreement by aborting at the end of the protocol. This will cause all loyal lieutenant generals to finalize their decision to abort during ($\mathbf{Round}_{ 4 }$), thereby agreeing to follow order $0$ at the end of the protocol.
		\begin{tcolorbox}
			[
				colback = WordIceBlue,
				sharp corners,
				boxrule = 0.5 pt,
			]
			According to assumption ($\mathbf{T}_{ 2 }$), at the end of ($\mathbf{Round}_{ 3 }$), each disloyal LT$_k$ will transmit the same decision $\bot$ he sent at the end of ($\mathbf{Round}_{ 2 }$) to each other lieutenant general. As a result, the decisions of the loyal lieutenant generals, will not be affected.
		\end{tcolorbox}
	\end{itemize}
\end{itemize}

Ergo, all loyal lieutenant generals will reach agreement.
\end{proof}

\bibliographystyle{ieeetr}
\bibliography{EPRQDBA}

\end{document}